\documentclass[acmsmall,screen,authorversion]{acmart} 

\acmPrice{}
\acmDOI{10.1145/3371136}
\acmYear{2020}
\copyrightyear{2020}
\acmJournal{PACMPL}
\acmVolume{4}
\acmNumber{POPL}
\acmArticle{68}
\acmArticleSeq{0}
\acmMonth{1}

\setcopyright{rightsretained}

\acmBadgeR{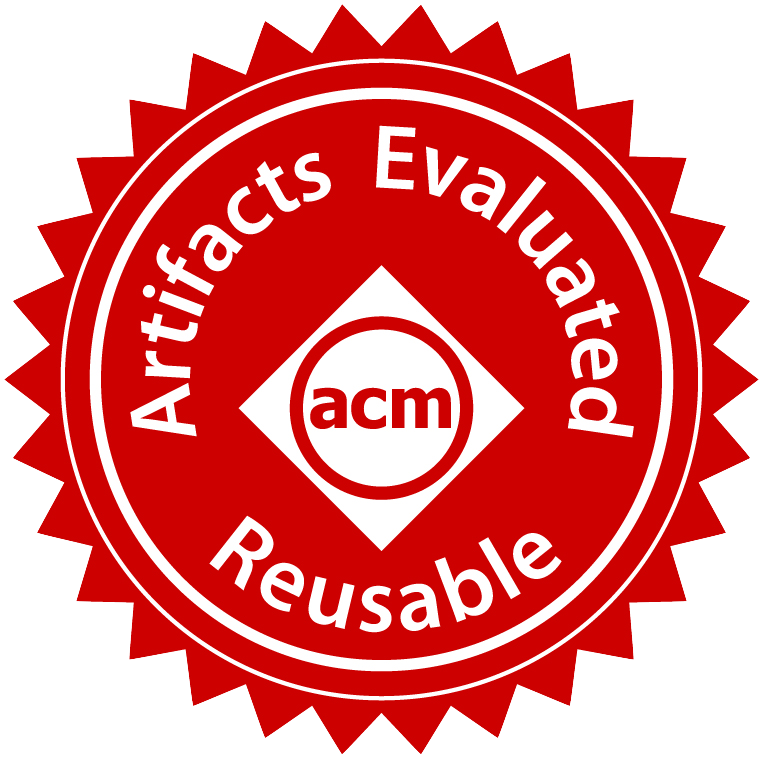}

\usepackage[utf8x]{inputenc}
\usepackage[american]{babel}
\usepackage[T1]{fontenc}
\usepackage{microtype}

\usepackage{booktabs}
\usepackage{subcaption}
\usepackage{paralist}
\usepackage{enumitem}
\usepackage{xcolor}
\usepackage{tabularx}
\usepackage{multirow}
\usepackage{wasysym}
\usepackage{mathpartir}
\usepackage{xspace}
\usepackage{pifont}
\usepackage{relsize}

\usepackage{wrapfig}
\newcommand{\setalignspacingtofiguremode}{%
  \abovedisplayskip=0pt%
  \belowdisplayskip=-2mm%
}

\usepackage[disable]{todonotes}

\usepackage{tikz}
\usetikzlibrary{automata,arrows,patterns,shapes,decorations,calc}
\usetikzlibrary{shadows.blur}
\tikzset{
    cell/.style={
        draw,
        rectangle split,
        rectangle split horizontal,
        rectangle split parts=2,
        rectangle split part align=base,
        align=center,
        rounded corners,
        text width=4mm,
        node distance=2cm,
    },
    next/.style={
    	circle,
    	minimum size=1.5mm,
    	inner sep=0pt,
    	fill=black,
    	draw=black,
    	node distance=3mm,
    },
    ptrbox/.style={
    	draw,
    	rectangle,
        text width=1.5mm,
        text height=1.5mm,
        node distance=1cm
    },
    ptr/.style={
    	circle,
    	minimum size=1.5mm,
    	inner sep=0pt,
    	fill=black,
    	draw=black,
    	node distance=0mm,
    },
    mybox/.style={
      rounded corners,
      blur shadow={shadow blur steps=5},
      shade,
      top color=white,
      bottom color=white,
      minimum width=1.9cm,
      minimum height=.85cm,
      anchor=center,
      font=\small
    },
    myboxarrow/.style={
      draw,
      line width=1mm,
      color=gray!80,
    },
    state/.style={
      circle,
      draw,
      inner sep=1pt,
      minimum size=.75cm,
      fill=none,
    },
}
\usetikzlibrary{external}
\usepackage{smartdiagram}

\newcounter{ObserverStateCounter}
\makeatletter
\newcommand{\observerstatelabel}[2]{%
   \protected@write \@auxout {}{\string \newlabel {#1}{{#2}{\thepage}{#2}{#1}{}} }%
   \hypertarget{#1}{}
}
\makeatother
\newcommand{\mkstatename}[1]{\stepcounter{ObserverStateCounter}$\apersitentlocation_{\arabic{ObserverStateCounter}}$\observerstatelabel{#1}{\ensuremath{\apersitentlocation_{\arabic{ObserverStateCounter}}}}}

\usepackage{listings}
\lstdefinestyle{condensed}{
  language=Java,
  basicstyle=\footnotesize\ttfamily,
  keywordstyle=\color{teal}\ttfamily,
  commentstyle=\color{gray}\ttfamily,
  tabsize=2,
  numberbychapter=false,
  morekeywords={NULL,CAS,shared,threadlocal,struct,atomic,delete},
  moredelim=**[is][\bfseries]{@}{@},
  moredelim=**[is][\bfseries\tiny]{§}{§},
  keepspaces=true,
  numbers=left,
  numberstyle=\scriptsize,
  numbersep=6pt,
  firstnumber=last,
  mathescape=true,
  xleftmargin=13pt,
  belowskip=-2mm,
}
\lstdefinestyle{inline}{
  basicstyle=\small\ttfamily,
  keywords={},
  mathescape=true,
}
\lstdefinestyle{typing}{
  delim=**[is][\color{teal}]{`}{`},
  escapebegin=\color{teal},
  keywords={}
}
\lstdefinestyle{typingnolines}{
  delim=**[is][\color{teal}]{`}{`},
  escapebegin=\color{teal},
  numbers=none,
  xleftmargin=0pt,
  keywords={}
}
\lstdefinestyle{condensedsimple}{
  delim=**[is][\color{teal}]{`}{`},
}
\newcommand{\code}[2][]{\lstinline[style=inline,#1]!#2!}
\newcommand{\mcode}[1]{\ensuremath{\text{\lstinline[style=inline]!#1!}}}

\usepackage{tcolorbox}
\tcbset{
  colback=black!3,
  boxrule=.5pt,
  boxsep=5pt,
  top=0pt,
  bottom=0pt,
  right=0pt,
  left=0pt,
  sharp corners,
}

\usepackage{amsmath,stackengine}
\stackMath
\newcommand\xxrightarrow[1]{\raisebox{-.75pt}{\ensuremath{\smash{\mathrel{%
  \setbox2=\hbox{\stackon{\scriptstyle#1}{\scriptstyle#1}}%
  \stackon[-3.5pt]{%
    \xrightarrow{\makebox[\dimexpr\wd2\relax]{}}%
  }{%
   \scriptstyle#1\,%
  }%
}}}}}

\newtheorem{assumption}{Assumption}
\newtheorem{remark}{Remark}

\usepackage{cleveref} 
\crefname{line}{line}{lines}
\Crefname{line}{Line}{Lines}
\crefname{rule}{rule}{rules}
\Crefname{rule}{rule}{Rules}
\crefname{assumption}{assumption}{assumptions}
\Crefname{assumption}{Assumption}{Assumptions}
\crefname{case}{case}{cases}
\crefname{property}{property}{properties}
\crefname{aux}{auxiliary}{auxiliaries}
\creflabelformat{aux}{#2(#1)#3}

\makeatletter
\newcommand{\rulelabel}[2]{%
   \protected@write \@auxout {}{\string \newlabel {#1}{{#2}{\thepage}{#2}{#1}{}} }%
   \hypertarget{#1}{}
}
\makeatother

\makeatletter
\providecommand*{\dashV}{%
  \mathrel{%
    \mathpalette\@dashV\Vdash
  }%
}
\newcommand*{\@dashV}[2]{%
  \reflectbox{$\m@th#1#2$}%
}
\makeatother

\newenvironment{casedistinction}{%
   \let\descriptionlabel\mycasedistinctionlabel
   \description[style=nextline,leftmargin=1pt,itemsep=1ex]
}{%
   \enddescription
}

\newcommand{\ad}[1]{\underline{\textit{Ad #1.}}}

\newcounter{auxpropcounter}[section]%
\newcounter{auxpropcounterSave}[section]%
\makeatletter%
\newenvironment{auxiliary}{%
  \setcounter{auxpropcounterSave}{\value{equation}}%
  \setcounter{equation}{\value{auxpropcounter}}%
  \align%
}{%
  \endalign%
  \setcounter{auxpropcounter}{\value{equation}}%
  \setcounter{equation}{\value{auxpropcounterSave}}%
}%
\makeatother

\newcommand{\mymathtt}[1]{\mathchoice
	{\text{\small\ttfamily#1}}
	{\text{\small\ttfamily#1}}
	{\text{\relscale{.9}\ttfamily#1}}
	{\text{\relscale{.9}\ttfamily#1}}
}

\newcommand{\set}[1]{\{#1\}}
\newcommand{\powersetof}[1]{2^{#1}}
\newcommand{\setcond}[2]{\{#1\mid\,#2\}}
\newcommand{\rangeof}[1]{\mathit{range}(#1)}
\newcommand{\domof}[1]{\mathit{dom}(#1)}
\newcommand{\restrict}[2]{#1|_{#2}}
\newcommand{\factorize}[2]{#1/_{#2}}
\newcommand{\project}[2]{{#1}\mkern-1.5mu\downarrow_{#2}}
\newcommand{\cardof}[1]{\,\mid\!#1\!\mid\,}
\newcommand{\cardnodistof}[1]{\mid\!#1\!\mid}
\renewcommand{\emptyset}{\varnothing}

\newcommand{\mkrulelabel}[1]{\rulelabel{rule:#1}{\textsc{(#1)}}}
\newcommand{\infrule}[3]{\inferrule[\textsc{(#1)}\mkrulelabel{#1}]{#2}{#3}}

\newcommand{\rconclusion}[2]{%
	\ifthenelse{#1 = 0}{\AxiomC{}\UnaryInfC{\( #2 \)}}{
	\ifthenelse{#1 = 1}{\UnaryInfC{\( #2 \)}}{%
	\ifthenelse{#1 = 2}{\BinaryInfC{\( #2 \)}}{%
	\ifthenelse{#1 = 3}{\TrinaryInfC{\( #2 \)}}{%
	\ifthenelse{#1 = 4}{\QuaternaryInfC{\( #2 \)}}{%
	\ifthenelse{#1 = 5}{\QuinaryInfC{\( #2 \)}}{%
		\textbf{ERROR}%
	}}}}}}%
}

\newcommand{\dom}{\ensuremath{\mathit{Dom}}}

\newcommand{\adr}{\mathit{Adr}}

\newcommand{\bnf}{\:\mid\:}
\newcommand{\acom}{\mathit{com}}

\newcommand{\aprog}{P}

\newcommand{\vars}{\mathit{Var}}
\newcommand{\svars}{\mathit{shared}}
\newcommand{\lvars}[1]{\mathit{local}_{#1}}
\newcommand{\pvars}{\mathit{PVar}}
\newcommand{\apavar}{\mathit{x}}
\newcommand{\apavarp}{\mathit{y}}
\newcommand{\apvar}{\mathit{p}}
\newcommand{\apvarp}{\mathit{q}}
\newcommand{\dvars}{\mathit{DVar}}
\newcommand{\advar}{\mathit{u}}
\newcommand{\advarp}{\mathit{u}'}
\newcommand{\avar}{\mathit{var}}
\newcommand{\lhs}{\mathit{lhs}}
\newcommand{\rhs}{\mathit{rhs}}
\newcommand{\psels}{\mathit{PSel}}
\newcommand{\dsels}{\mathit{DSel}}
\newcommand{\pexp}{\mathit{PExp}}
\newcommand{\dexp}{\mathit{DExp}}
\newcommand{\apexp}{\mathit{pexp}}
\newcommand{\apexpp}{\mathit{qexp}}

\newcommand{\anadr}{a}
\newcommand{\anadrp}{b}
\newcommand{\anadrpp}{c}
\newcommand{\advalue}{d}
\newcommand{\segval}{\mymathtt{seg}}

\newcommand{\aheap}{\mathit{m}}
\newcommand{\aheapp}{\mathit{m}'}

\newcommand{\athread}{t}
\newcommand{\athreadp}{t'}
\newcommand{\athreadpp}{t''}
\newcommand{\validof}[1]{\mathit{valid}_{#1}}
\newcommand{\heapcomput}[1]{\mem{#1}}
\newcommand{\heapcomputof}[2]{\memof{#1}{#2}}

\newcommand{\mem}[1]{\aheap_{#1}}
\newcommand{\memof}[2]{\mem{#1}({#2})}

\newcommand{\sem}[1]{\llbracket#1\rrbracket}
\newcommand{\asem}[3][\aprog]{\sem{#1}_{#2}^{#3}}
\newcommand{\asemobs}[4][\aprog]{#4\sem{#1}_{#2}^{#3}}
\newcommand{\allsem}[1][\aprog]{\asem[#1]{\adr}{\adr}}
\newcommand{\freesem}[1][\aprog]{\asem[#1]{\adr}{\emptyset}}
\newcommand{\nosem}[1][\aprog]{\asem[#1]{\emptyset}{\emptyset}}

\newcommand{\allsemobs}[1][\aprog]{\asemobs[#1]{\adr}{\adr}{\anobs}}
\newcommand{\freesemobs}[1][\aprog]{\asemobs[#1]{\adr}{\emptyset}{\anobs}}

\newcommand{\adrof}[1]{\mathit{adr}(#1)}
\newcommand{\vadrof}[1]{\adrof{\restrict{\heapcomput{#1}}{\validof{#1}}}}
\newcommand{\step}[1][]{\dashrightarrow_{#1}} 

\newcommand{\freeableof}[3][\anobs]{\mathcal{F}_{#1}(#2,\mkern+2mu#3)}

\newcommand{\retire}{\mymathtt{retire}}
\newcommand{\retireof}[1]{\retire(#1)}
\newcommand{\guard}{\mymathtt{protect}}

\newcommand{\leaveQ}{\mymathtt{leaveQ}}
\newcommand{\enterQ}{\mymathtt{enterQ}}

\newcommand{\malloc}{\mymathtt{malloc}}

\newcommand{\assert}{\mymathtt{assert}}
\newcommand{\assertof}[1]{\assert\ #1}
\newcommand{\invariant}{\texttt{@}\mymathtt{inv}}
\newcommand{\invariantof}[1]{\invariant\ #1}
\newcommand{\activeof}[1]{\mymathtt{active}(#1)}

\newcommand{\assume}{\mymathtt{assume}}
\newcommand{\assumeof}[1]{\assume\ #1}

\newcommand{\psel}[1]{#1.\mymathtt{next}}
\newcommand{\dsel}[1]{#1.\mymathtt{data}}

\newcommand{\op}{\mymathtt{op}}
\newcommand{\opof}[1]{\op(#1)}
\newcommand{\pred}{\mymathtt{pred}}
\newcommand{\predof}[1]{\pred(#1)}
\newcommand{\acond}{\mathit{cond}}

\newcommand{\enter}{\mymathtt{enter}}
\newcommand{\enterof}[1]{\enter\:#1}
\newcommand{\exit}{\mymathtt{exit}}
\newcommand{\exitof}[1]{\exit\:#1}
\newcommand{\cskip}{\mymathtt{skip}}

\newcommand{\choice}{\oplus}

\newcommand{\vecof}[1]{\bar{#1}}
\newcommand{\trans}[1]{\xxrightarrow{#1}}
\newcommand{\translab}[2]{\ensuremath{#1,\,#2}}
\newcommand{\translabbr}[2]{\ensuremath{#1},\\\ensuremath{#2}}
\newcommand{\evt}[2]{#1(#2)}
\newcommand{\afunc}{\mathit{func}}
\newcommand{\afuncof}[1]{\afunc(#1)}

\newcommand{\aval}[1][]{\mathit{v_{#1}}}
\newcommand{\history}{\mathcal{H}}

\newcommand{\historyof}[2][\!]{\history_{#1}({#2})}
\newcommand{\specof}[1]{\mathcal{S}({#1}\mkern+.5mu)}
\newcommand{\anevent}{\mathit{evt}}

\newcommand{\aguard}{\mathit{g}}
\newcommand{\apersitentlocation}{\mathit{L}} 
\newcommand{\initlocation}{\mathit{l}_{\mathit{init}}}
\newcommand{\alocation}{\mathit{l}}

\newcommand{\alocationp}{\alocation'}

\newcommand{\astate}{\mathit{s}}
\newcommand{\astatep}{\astate'}
\newcommand{\anobs}[1][]{\mathcal{O}_{#1}}
\newcommand{\baseobs}{\anobs[\mkern-1mu\mathit{Base}]}
\newcommand{\ebrobs}{\anobs[\mkern-1mu\mathit{EBR}]}
\newcommand{\hpobs}{\anobs[\mkern-1mu\mathit{HP}]}
\newcommand{\implobs}{\anobs[\mkern-1mu\mathit{Impl}]}
\newcommand{\smrobs}[1][\mkern-1mu\mathit{SMR}]{\anobs[#1]}
\newcommand{\ahist}{h}
\newcommand{\ahistp}{h'}

\newcommand{\controlof}[2][]{\mathit{ctrl}_{#1}(#2)}

\newcommand{\absevt}[2]{#1(#2)}
\newcommand{\hconcat}{\mkern+1mu.\mkern+2.5mu}

\newcommand{\adsraw}{\mathit{D}}
\newcommand{\ads}[1][]{\adsraw}
\newcommand{\retiredof}[1]{\mathit{retired}({#1})}

\newcommand{\anact}{\mathit{act}}

\newcommand{\anup}{\mathit{up}}

\newcommand{\anexp}{\mathit{e}}
\newcommand{\anexpp}{\mathit{e'}}
\newcommand{\comof}[1]{\acom({#1})}

\newcommand{\invalidof}[1]{\mathit{invalid}_{#1}}

\newcommand{\computequiv}{\sim}

\newcommand{\obsrel}[1][\!]{\lessdot_{#1}}

\newcommand{\free}{\mymathtt{free}}
\newcommand{\freeof}[1]{\free(#1)}
\newcommand{\freedof}[1]{\mathit{freed}({#1})}
\newcommand{\freshof}[1]{\mathit{fresh}({#1})}

\newcommand{\renamingof}[3]{#1[#2 / #3]}

\newcommand{\computrel}{\prec}

\newcommand{\astmt}{\mathit{stmt}}
\newcommand{\astmtof}[1]{\astmt(#1)}
\newcommand{\astmtp}{\mathit{stmt}'}

\newcommand{\threadof}[1]{\mathit{thrd}({#1})}

\newcommand{\return}[1][]{\mymathtt{return}}

\newcommand{\apc}{\mathit{pc}}
\newcommand{\pcinit}{\mathit{pc}_\mathit{init}}
\newcommand{\apcp}{\mathit{pc}'}

\newcommand{\freesof}[1]{\mathit{frees}(#1)}

\newcommand{\typerule}{\vdash}
\newcommand{\typecom}[4][]{\{\mkern+2mu#2\mkern+2mu\}\:#3\:\{\mkern+2mu#4\mkern+2mu\}}
\newcommand{\typestmt}[4][]{\typecom[#1]{#2}{#3}{#4}}
\newcommand{\TYPECOM}[4][]{\,\typerule_{#1}\mkern-1mu\typecom[#1]{#2}{#3}{#4}} 
\newcommand{\TYPESTMT}[4][]{\,\typerule_{#1}\mkern-1mu\typestmt[#1]{#2}{#3}{#4}} 
\newcommand{\typejudge}[4][]{\,\typerule_{#1}\mkern-1mu\typestmt[#1]{#2}{#3}{#4}}

\newcommand{\typechecks}[2][]{\,\vdash_{#1}\!#2}

\newcommand{\atype}{T}
\newcommand{\atypep}{T'}
\newcommand{\atypepp}{T''}
\newcommand{\typeof}[2]{#1\,{:}\,#2}
\newcommand{\typeofany}[1]{#1}

\newcommand{\glocal}{\mathbb{L}}
\newcommand{\gcustom}[1]{\mathbb{E}_{#1}}
\newcommand{\gactive}{\mathbb{A}}
\newcommand{\env}{\Gamma}
\newcommand{\envp}{\Gamma'}
\newcommand{\envpp}{\Gamma''}
\newcommand{\envof}[1]{\Gamma(#1)}
\newcommand{\envpof}[1]{\Gamma'(#1)}
\newcommand{\envinit}{\Gamma_\mathit{\!init}}
\newcommand{\envinitt}[1]{\renameof{\Gamma_{\mathit{\!init}}}{#1}}

\newcommand{\safecallof}[2]{\mathit{safeEnter}(#1,#2)}

\newcommand{\checkoftype}[5][]{#2,#3,#4\leadsto_{#1}#5}
\newcommand{\checkof}[4][]{#2,#3\leadsto_{#1}#4}
\newcommand{\checknocomof}[4][]{#2\leadsto_{#1}#4}
\newcommand{\rmtransientof}[1]{\mathit{rm}(#1)}
\newcommand{\isvalidof}[1]{\mathit{isValid}({#1})}
\newcommand{\invholdsof}[2][]{\mathit{inv}_{#1}(#2)}
\newcommand{\flatof}[2]{\mathit{flat}_{#1}(#2)}
\newcommand{\noaliasof}[2]{\mathit{noalias}_{#1}(#2)}
\newcommand{\renameof}[2]{#1^{[#2]}}
\newcommand{\renamevarof}[2]{#1_{#2}}
\newcommand{\atomicbegin}{\mymathtt{beginAtomic}}
\newcommand{\atomicend}{\mymathtt{endAtomic}}
\newcommand{\locksetof}[1]{\mathit{lock}(#1)}

\newcommand{\lreach}{\mathit{reach}}
\newcommand{\lreachof}[4][\anobs]{\lreach_{#1,#2,#3}(#4)}
\newcommand{\lpost}{\mathit{post}}
\newcommand{\lpostof}[3]{\lpost_{#1,#2}(#3)}
\newcommand{\alocset}{L}
\newcommand{\alocsetp}{L'}

\newcommand{\ghost}{\texttt{@}\mymathtt{inv}}
\newcommand{\ghostof}[1]{\ghost\ #1}
\newcommand{\chooseof}[1]{\mymathtt{angel}\ #1}
\newcommand{\containsof}[2]{#1\ \mymathtt{in}\ #2}
\newcommand{\aghostvar}{\mathit{r}}
\newcommand{\aghostvarp}{\mathit{r}'}

\newcommand{\denotation}[1]{\mathit{repr}_{\!#1}}
\newcommand{\denotationof}[2]{\mathit{repr}_{\!#1}(#2)}
\newcommand{\activeofcomp}[1]{\mathit{active}({#1})}
\newcommand{\gvars}{\mathit{AVar}}
\newcommand{\gsafeaccess}{\mathbb{S}}

\newcommand{\havocof}[1]{\mymathtt{havoc}(#1)}
\newcommand{\btrue}{\mymathtt{true}}
\newcommand{\bfalse}{\mymathtt{false}}
\newcommand{\instrumentation}{\mathit{inst}}
\newcommand{\instrumentationof}[1]{\instrumentation(#1)}
\newcommand{\issafeof}[1]{\mathit{safe}(#1)}
\newcommand{\update}[1]{#1}

\newcommand{\applyupdate}[2]{#1[#2]}

\newcommand{\alltypes}{\mathit{Types}}
\newcommand{\actypes}{\mathit{AntiChainTypes}}
\newcommand{\allvars}{\mathit{Vars}}
\newcommand{\allenvs}{\mathit{Envs}}
\newcommand{\fail}{\top}
\newcommand{\allenvsbot}{\allenvs_{\fail}}
\newcommand{\premiseof}[1]{\mathit{pre}_{#1}}
\newcommand{\updateof}[1]{\mathit{up}_{#1}}

\newcommand{\cvars}{\mathit{CVars}}
\newcommand{\IReq}{\mathit{Req}}
\newcommand{\IDep}{\mathit{Dep}}

\newcommand{\eqsys}{\mathit{\Phi}}
\newcommand{\eqsysof}[1]{\eqsys(#1)}
\newcommand{\spof}[1]{\mathit{sp}(#1)}

\newcommand{\lsol}{\mathit{lsol}}
\newcommand{\lsolof}[1]{\lsol(#1)}

\newcommand{\locsof}[1]{\mathit{Loc}({#1})}

\newcommand{\threadvar}{z_{t}}
\newcommand{\adrvar}{z_a}

\newcommand{\illush}{\mathit{shared}}
\newcommand{\gilluinv}{\gcustom{\mathit{inv}}}
\newcommand{\gilluisu}{\gcustom{\mathit{isu}}}
\newcommand{\gilluebr}{\gcustom{\mathit{acc}}}
\newcommand{\ccpp}{\code{C/C++}\xspace}
\newcommand{\cpp}{\code{C++}\xspace}
\newcommand{\cppeleven}{\code{C/C++11}\xspace}
\newcommand{\segfault}{\code{segfault}\xspace}
\newcommand{\cavetool}{\textsc{cave}\xspace}

\newcommand{\presection}{}
\newcommand{\thetool}{\textsc{seal}\xspace}

\newcommand{\rawsymbolYes}{{\color[RGB]{0,155,85}\smash{\ding{51}}}}
\newcommand{\rawsymbolNo}{{\color[RGB]{200,20,40}\smash{\ding{55}}}}
\newcommand{\rawsymbolTO}{{\smash{\scalebox{.13}{%
  \begin{tikzpicture}[line cap=rect,line width=3pt]
    \definecolor{clockColor}{RGB} {247,146,29}
    \draw [draw=clockColor] (0,0) circle [radius=1cm];
    \foreach \angle [count=\xi] in {0,180,270,60,30,...,-240}
    {
      \draw[clockColor,line width=2pt] (\angle:.8cm) -- (\angle:1cm);
    }
    \draw[fill=clockColor,draw=none] (0,0) -- (90:1cm) arc (90:-40:1cm) -- cycle;
  \end{tikzpicture}%
}}}}
\newcommand{\symbolYes}{\rawsymbolYes\xspace}
\newcommand{\symbolNo}{\rawsymbolNo\xspace}
\newcommand{\symbolTO}{\raisebox{-.75pt}{\rawsymbolTO}\xspace}

\begin{document}

	\title{Pointer Life Cycle Types for Lock-Free Data Structures with Memory Reclamation}
	\subtitle{Extended Version}

	\author{Roland Meyer}
	\orcid{0000-0001-8495-671X}
	\affiliation{%
		\institution{TU Braunschweig}
		\country{Germany}
	}
	\email{roland.meyer@tu-bs.de}
	\author{Sebastian Wolff}
	\orcid{0000-0002-3974-7713}
	\affiliation{%
		\institution{TU Braunschweig}
		\country{Germany}
	}
	\email{sebastian.wolff@tu-bs.de}


\begin{abstract}
	We consider the verification of lock-free data structures that manually manage their memory with the help of a safe memory reclamation (SMR) algorithm. 
	Our first contribution is a type system that checks whether a program properly manages its memory.
	If the type check succeeds, 
	it is safe to ignore the SMR algorithm and consider the program under garbage collection.
	Intuitively, our types track the protection of pointers as guaranteed by the SMR algorithm.
	There are two design decisions.
	The type system does not track any shape information, which makes it extremely lightweight.
	Instead, we rely on invariant annotations that postulate a protection by the SMR.
	To this end, we introduce angels, ghost variables with an angelic semantics.
	Moreover, the SMR algorithm is not hard-coded but a parameter of the type system definition.
	To achieve this, we rely on a recent specification language for SMR algorithms.
	Our second contribution is to automate the type inference and the invariant check.
	For the type inference, we show a quadratic-time algorithm.
	For the invariant check, we give a source-to-source translation that links our programs to off-the-shelf verification tools.
	It compiles away the angelic semantics.
	This allows us to infer appropriate annotations automatically in a guess-and-check manner.
	To demonstrate the effectiveness of our type-based verification approach, we check linearizability for various list and set implementations from the literature with both hazard pointers and epoch-based memory reclamation.
	For many of the examples, this is the first time they are verified automatically.
	For the ones where there is a competitor, we obtain a speed-up of up to two orders of magnitude.

\end{abstract}

	\maketitle


\section{Introduction}
\label{sec:introduction}

In the last decade we have experienced an upsurge in massive parallelization being available even in commodity hardware.
To keep up with this trend, popular programming languages include in their standard libraries features to make parallelization available to everyone.
At the heart of this effort are concurrent (thread-safe) data structures.
Consequently, efficient implementations are in high demand.
In practice, lock-free data structures are particularly efficient.

Unfortunately, lock-free data structures are also particularly hard to get correct.
The reason is the absence of traditional synchronization using locks and mutexes in favor of low-level synchronization using hardware instructions.
This calls for formal verification of such implementations.
In this context, the de-facto standard correctness property is linearizability \cite{DBLP:journals/toplas/HerlihyW90}.
It requires, intuitively, that each operation of a data structure implementation appears to execute atomically somewhen between its invocation and return.
For users of lock-free data structures, linearizability is appealing.
It provides the illusion of atomicity---they can use the data structure as if they were using it in a sequential setting.

Proving lock-free data structures linearizable has received a lot of attention (cf. \Cref{sec:related_work}).
\citet{DBLP:conf/forte/DohertyGLM04}, for instance, give a mechanized proof of a practical lock-free queue.
Such proofs require plenty of manual work and take a considerable amount of time.
Moreover, they require an understanding of the proof method and the data structure under consideration. 
To overcome this drawback, we are interested in automated verification.
The \cavetool tool by \citet{DBLP:conf/cav/Vafeiadis10,DBLP:conf/vmcai/Vafeiadis10}, for example, is able to establish linearizability for singly-linked data structures fully~automatically.

The problem with automated verification for lock-free data structures is its limited applicability.
Most techniques are restricted to implementations that assume a garbage collector (GC).
This assumption, however, does not apply to all programming languages. 
Take \ccpp as an example.
It does not provide an automatic garbage collector that is running in the background.
Instead, it is the programmer's obligation to avoid memory leaks by reclaiming memory that is no longer in use (using \code{delete}). 
In lock-free data structures, this task is much harder than it may seem at first glance.
The root of the problem is that threads typically traverse the data structure without synchronization.
Hence, there may be threads holding pointers to records that have already been removed from the structure.
If records are reclaimed immediately after the removal, those threads are in danger of accessing deleted memory.
Such accesses are considered unsafe (undefined behavior in \ccpp~\cite{iso-cppeleven}) and are a common cause for system crashes due to a \segfault.
The solution to this problem are so-called \emph{safe memory reclamation (SMR)} algorithms.
Their task is to provide lock-free means for deferring the reclamation/deletion until all unsynchronized threads have finished their accesses.
Typically, this is done by replacing explicit deletions with calls to a function \code{retire} provided by the SMR algorithm which defers the deletion.
Coming up with efficient and practical SMR implementations is difficult and an active field of research (cf. \Cref{sec:related_work}).

The use of SMR algorithms to manage manually the memory of lock-free data structures hinders verification, both manual and automated.
This is due to the high complexity of such algorithms.
As hinted before, an SMR implementation needs to be lock-free in order not to spoil the lock-free guarantee of the data structure using it.
In fact, SMR algorithms are quite similar to lock-free data structures implementation-wise.
This added complexity could not be handled by automatic verifiers up until recently.
\citet{DBLP:journals/pacmpl/MeyerW19} were the first to present a practical approach.
Their key insight is that the data structure can be verified as if it was relying on a garbage collector rather than an SMR algorithm, provided the data structure does not perform unsafe memory operations.
Since data structures from the literature are usually memory safe, the above insight is a powerful tool for verification.
Nevertheless, it leaves us with a hard task: establishing that all memory operations are safe in the presence of memory reclamation.
\citet{DBLP:journals/pacmpl/MeyerW19} were not able to conduct this check under GC.
Instead, they explore the entire state space of the data structure with SMR, restricting reallocations to a single address, to prove ABAs harmless (a criterion they require for soundness).
Unfortunately, their state space exploration does not scale well.

In the present paper we tackle the challenge of proving a lock-free data structure memory safe.
We present a type system to address this task.
That is, we present a syntax-centric approach to establish the semantic property of memory safety.
In particular, we no longer need expensive state space explorations that can handle SMR and memory reuse in order to prove memory safety.
This allows us to utilize the full potential of the above result:
if our type check succeeds, we remove the SMR code from the data structure and verify the resulting implementation using an off-the-shelf GC verifier.
The idea behind our type system is a life cycle common to lock-free data structures with manual memory management via SMR \cite{DBLP:conf/podc/Brown15}.
The life cycle, depicted in \Cref{fig:memory-life-cycle}, has four stages:
\begin{inparaenum}[(i)]
	\item local,
	\item active,
	\item retired, and
	\item not allocated.
\end{inparaenum}
Newly allocated records are in the local stage.
The record is known only to the allocating thread; it has exclusive read/write access.
The goal of the local stage is to prepare records for being published, i.e., added to the shared state of the data structure.
When a record is published, it enters the active stage.
In this stage, accesses to the record are safe because it is guaranteed to be allocated.
However, no thread has exclusive access and thus must fear interference by others.
It is worth pointing out that a publication is irreversible.
Once a record becomes active it cannot become local again.
A thread, even if it removes the active record from the shared structures, must account for other threads that have already acquired a pointer to that record.
To avoid memory leaks, removed records eventually become retired.
In this stage, threads may still be able to access the record safely.
Whether or not they can do so depends on the SMR algorithm used.
Finally, the SMR algorithm detects that the retired record is no longer in use and reclaims it.
Then, the memory can be reused and the life cycle begins anew.

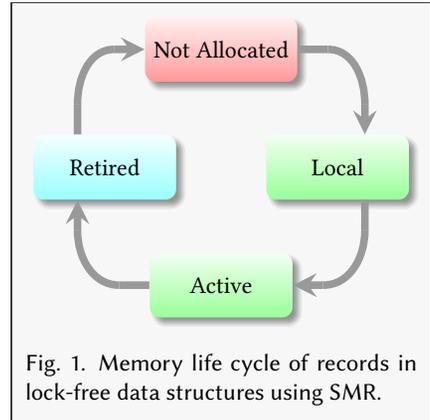
\begin{wrapfigure}{r}{5.6cm}
	\vspace{.5mm}
	\begin{tcolorbox}
		\center
		\begin{tikzpicture}
			\newcommand{\VD}{1.55}
			\newcommand{\HD}{1.55}
			\node[mybox, bottom color=red!40, top color=red!5] (unalloc) at (0,\VD) {Not Allocated};
			\node[mybox, bottom color=green!40, top color=green!5] (active) at (0,-\VD) {Active};
			\node[mybox, bottom color=green!40, top color=green!5] (local) at (\HD,0) {Local};
			\node[mybox, bottom color=cyan!40, top color=cyan!5] (retired) at (-\HD,0) {Retired};
			\path[myboxarrow,out=0,in=90] (unalloc.east) -- ([xshift=4mm]unalloc.east) edge ([yshift=4mm]local.50);
			\path[myboxarrow,-stealth] ([yshift=4mm]local.50) -> (local.50);
			\path[myboxarrow,out=270,in=0] (local.310) -- ([yshift=-4mm]local.310) edge ([xshift=4mm]active.east);
			\path[myboxarrow,-stealth] ([xshift=4mm]active.east) -> (active.east);
			\path[myboxarrow,out=180,in=270] (active.west) -- ([xshift=-4mm]active.west) edge ([yshift=-4mm]retired.230);
			\path[myboxarrow,-stealth] ([yshift=-4mm]retired.230) -> (retired.230);
			\path[myboxarrow,out=90,in=180] (retired.130) -- ([yshift=4mm]retired.130) edge ([xshift=-4mm]unalloc.west);
			\path[myboxarrow,-stealth] ([xshift=-4mm]unalloc.west) -> (unalloc.west);
		\end{tikzpicture}
		\caption{%
			Memory life cycle of records in lock-free data structures using SMR.
		}
		\label{fig:memory-life-cycle}
	\end{tcolorbox}
\end{wrapfigure}
The main challenge our type system has to address wrt. the above memory life cycle is the transition from the active to the retired stage.
Due to the lack of synchronization, this can happen without a thread noticing.
Programmers are aware of the problem.
They protect records while they are active such that the SMR guarantees safe access even though the record is retired.
To cope with this, our types integrate knowledge about the SMR algorithm.
A core aspect of our development is that the actual SMR algorithm is an input to our type system---it is not tailored towards a specific SMR algorithm.

An additional challenge arises from the type system performing a thread-local analysis, it considers the program code as if it was sequential.
This means the type system is not aware of the actual interference among threads, unlike state space explorations.
To address this, we use types that are stable under the actions of interfering threads \cite{DBLP:journals/acta/OwickiG76}.

In practice, protecting a record while it is active is non-trivial.
Between acquiring a pointer to the record and the subsequent SMR protection call, an interferer may retire the record, in which case the protection has no effect.
SMR algorithms usually offer no means to check whether a protection was successful.
Instead, programmers exploit intricate data structure invariants to perform this check.
A common such invariant, for instance, is \emph{all shared reachable records are active}.
A type system typically cannot detect such data structure shape invariants.
We turn this weakness into a strength.
We deliberately do not track shape invariants nor alias information.
Instead, we use simple annotations to mark pointers that point to active records.
To relieve the programmer from arguing about their correctness, we show how to discharge annotations automatically.
Interestingly, this can be done with off-the-shelf GC verifiers.
It is worth pointing out that the ability to automatically discharge invariants allows for an automated guess-and-check approach for placing invariant~annotations.

To increase the applicability of our type system, we use the theory of movers \cite{Lipton75} as an enabling technique.
Movers are a standard approach to transform a program into a \emph{more atomic} version while retaining its behavior.
That the resulting program is more atomic is beneficial for verification.
The transformations are practical:
\citet{DBLP:conf/popl/ElmasQT09}, for example, automate them.

To demonstrate the usefulness of our approach, we implemented a linearizability checker which realizes the techniques presented in this paper.
That is, our tool
\begin{inparaenum}[(i)]
	\item performs a type inference to establish memory safety relying on invariant annotations,
	\item discharges the annotations under GC using \cavetool as a back-end, and
	\item verifies linearizability under GC using \cavetool.
\end{inparaenum}
Additionally, we implemented a prototype for automatically inserting annotations and applying movers.
These program transformations are performed on demand, guided by a failed type inference.
Our tool is able to establish linearizability for lock-free data structures from the literature, like Michael\&Scott's lock-free queue \cite{DBLP:conf/podc/MichaelS96}, the Vechev\&Yahav CAS set \cite{DBLP:conf/pldi/VechevY08}, the Vechev\&Yahav DCAS set \cite{DBLP:conf/pldi/VechevY08}, and the ORVYY set \cite{DBLP:conf/podc/OHearnRVYY10}, for the well-known hazard pointer method \cite{DBLP:conf/podc/Michael02} as well as epoch-base reclamation \cite{DBLP:phd/ethos/Fraser04}.
We stress that our approach is not limited to \cavetool as a back-end but can use any verifier for garbage collection.
To the best of our knowledge, we are the first to automatically verify lock-free set implementations that use SMR.

We summarize our contributions and the outline of the paper:
\begin{compactitem}
	\item[§\ref{sec:type_system}]
		presents our type system for proving lock-free data structures memory safe wrt. a user-specified SMR algorithm,
	\item[§\ref{sec:type_checking}]
		presents an efficient type inference algorithm,
	\item[§\ref{sec:invariant-check}]
		presents an instrumentation of the data structure under scrutiny to discharge invariant annotations fully automatically with the help of a GC verifier, and
	\item[§\ref{sec:evaluation}]
		evaluates our approach on well-known lock-free data structures from the literature.
\end{compactitem}
We illustrate our contribution in §\ref{sec:illustration}, introduce the programming model in §\ref{sec:preliminaries}, discuss preliminary results in §\ref{sec:prf}, give a comprehensive example in §\ref{sec:example}, and discuss related work in §\ref{sec:related_work}.


\section{The Contribution on an Example}
\label{sec:illustration}


\begin{figure}
	\begin{tcolorbox}
	\center%
\begin{lstlisting}[style=condensed,belowskip=-4mm,aboveskip=0pt]
struct Node { data_t data; Node* next; };
shared Node* Head, Tail;
\end{lstlisting}%
	\begin{minipage}[t]{.47\textwidth}%
\begin{lstlisting}[style=condensed]

atomic init() {
	Head = Tail = new Node();
	Head->next = NULL;
}

void enqueue(data_t input) {
§E§	@leaveQ();@
	Node* node = new Node();
	node->data = input;
	node->next = NULL;
	while (true) {
		Node* tail = Tail;
§H§		@protect(tail, 0);@ $\label[line]{code:michaelscott:protect-tail}$
§H§		@if (tail != Tail) continue;@
		Node* next = tail->next;
		if (tail != Tail) continue;
		if (next != NULL) {
			CAS(&Tail, tail, next);
			continue;
		}
		if (CAS(&tail->next, next, node)) {
			CAS(&Tail, tail, node);
			break;
		}
	}
§E§	@enterQ();@
}
\end{lstlisting}%
	\end{minipage}%
	\hfill%
	\begin{minipage}[t]{.465\textwidth}%
\begin{lstlisting}[style=condensed]
data_t dequeue() { $\label[line]{code:michaelscott:dequeue}$
§E§	@leaveQ();@
	while (true) {
		Node* head = Head; $\label[line]{code:michaelscott:read-head}$
§H§		@protect(head, 0);@ $\label[line]{code:michaelscott:protect-head}$
§H§		@if (head != Head) continue;@ $\label[line]{code:michaelscott:check-head}$
		Node* tail = Tail; $\label[line]{code:michaelscott:read-tail}$
		Node* next = head->next; $\label[line]{code:michaelscott:read-next}$
§H§		@protect(next, 1);@ $\label[line]{code:michaelscott:protect-next}$
		if (head != Head) continue; $\label[line]{code:michaelscott:check-head-again}$
		if (next == NULL) {
§E§			@enterQ();@
			return EMPTY;
		}
		if (head == tail) {
			CAS(&Tail, tail, next);
			continue;
		} else {
			data_t output = next->data;
			if (CAS(&Head, head, next)) { $\label[line]{code:michaelscott:cas}$
				// delete head; $\label[line]{code:michaelscott:free}$
§HE§			@retire(head);@ $\label[line]{code:michaelscott:retire}$
§E§				@enterQ();@
				return output; $\label[line]{code:michaelscott:dequeue-return}$
			}
		}
	}
}
\end{lstlisting}%
	\end{minipage}%
	\caption{%
		Michael\&Scott's lock-free queue \cite{DBLP:conf/podc/MichaelS96} with two different safe memory reclamation techniques: epoch-based reclamation (EBR) \cite{DBLP:phd/ethos/Fraser04} and hazard pointers (HP) \cite{DBLP:conf/podc/Michael02}.
		The modifications needed to use EBR are marked with \code{E} and the modifications needed to use HP are marked with \code{H}.
		For HP, we assume two hazard pointers per thread.
	}
	\label{fig:michaelscott}
	\end{tcolorbox}
\end{figure}

We illustrate our approach on Micheal\&Scott's lock-free queue \cite{DBLP:conf/podc/MichaelS96}, \Cref{fig:michaelscott}, which is used, for example, as Java's \href{https://docs.oracle.com/javase/9/docs/api/java/util/concurrent/ConcurrentLinkedQueue.html}{\code{ConcurrentLinkedQueue}} and as C++ Boost's \href{https://www.boost.org/doc/libs/1_70_0/boost/lockfree/queue.hpp}{\code{lockfree:$\!$:queue}}.
The queue is organized as a \code{NULL}-terminated singly-linked list of nodes.
The \code{enqueue} operation appends new nodes to the end of the list.
To do so, an enqueuer first moves \code{Tail} to the last node as it may lack behind.
Then, the new node is appended by pointing \code{Tail->next} to it.
Last, the enqueuer tries to move \code{Tail} to the end of the list.
This can fail as another thread may already have moved \code{Tail} to avoid waiting for the enqueuer.
The \code{dequeue} operation removes the first node from the list.
Since the first node is a dummy node, \code{dequeue} reads out the data value of the second node in the list and then moves the \code{Head} to that node.
Additionally, \code{dequeue} maintains the property that \code{Head} does not overtake \code{Tail}.
This is done by moving \code{Tail} towards the end of the list if necessary.
(There is an optimized version due to \citet{DBLP:conf/forte/DohertyGLM04} which avoids this step.)
Note that updates to the shared list of nodes are performed exclusively with single-word atomic compare-and-swap (CAS).

So far, the discussed implementation assumes a garbage collector.
The nodes allocated by \code{enqueue} are not reclaimed explicitly after being removed from the shared list by \code{dequeue}: the queue leaks memory.
Unfortunately, there is no simple solution to this problem.
Uncommenting the explicit deletion from \Cref{code:michaelscott:free} avoids the leak.
However, it leads to use-after-free bugs.
Due to the lack of synchronization, threads may still hold and dereference pointers to the now deleted node.
A dereference of such a dangling pointer, however, is unsafe.
In \ccpp, for example, dereferencing a dangling pointer has \emph{undefined behavior} \cite{iso-cppeleven} and may make the system crash with~a~\segfault.

To solve the problem, programmers employ safe memory reclamation (SMR) algorithms.
Two well-known examples are epoch-based reclamation (EBR) \cite{DBLP:phd/ethos/Fraser04} and hazard pointers (HP) \citet{DBLP:conf/podc/Michael02}.
They offer a function \code{retire} that replaces the ordinary \code{delete}.
The difference is that \code{retire} does not immediately delete nodes.
Instead, it defers the deletion until it is safe.
In order to discover whether a deletion is safe, threads need to declare which nodes they access.
How this is done depends on the~SMR~algorithm.

Epoch-based reclamation offers two additional functions \code{leaveQ} and \code{enterQ}.
Threads use the former to announce that they are going to access the data structure and use the latter to announce that they have finished the access.
The function names, in particular the \code{Q}, refer to the fact that the threads are \emph{quiescent} \cite{McKenney1998ReadcopyUU} between \code{enterQ} and \code{leaveQ}, meaning they do not modify the data structure.
During the non-quiescent period, EBR guarantees that the shared reachable nodes are not reclaimed, even if they are removed from the data structure and retired.
To use EBR, the programmer simply replaces \code{delete} statements with calls to \code{retire} and adds calls to \code{leaveQ} (\code{enterQ}) at the beginning (end) of data structure operations.
Consider \Cref{fig:michaelscott} for an example; the lines marked by \code{E} are the modifications required to use EBR.
While easy to use, EBR implementations usually stop reclaiming memory altogether upon thread failure.
Hazard pointers do not suffer from this problem.

The hazard pointer method requires threads to declare which nodes they access in a per-node fashion.
To that end, HP offers an additional function: \code{protect}.
It signals that a deletion of the received node should be deferred.
To be precise, HP guarantees that the deletion of a node is deferred if it has been continuously protected since before it was retired \cite{DBLP:conf/esop/GotsmanRY13}.
While this method is conceptually simple, it is non-trivial to apply.

To use hazard pointers with Michael\&Scott's queue requires to add the code marked by \code{H} in \Cref{fig:michaelscott}.
As for EBR, \code{delete} statements are replaced with \code{retire}.
Moreover, pointers that are accessed need protection to defer their deletion.
Simply calling \code{protect} is usually insufficient as the \code{protect} may be too late.
A common pattern for protecting pointers is to first protect them and then check that they have not been retired since.
In Michael\&Scott's queue this is done by testing whether the protected nodes are still shared reachable---the queue maintains the invariant that nodes reachable from the shared pointers are never retired.
To make this precise, consider \Cref{code:michaelscott:read-head,code:michaelscott:protect-head,code:michaelscott:check-head}.
\Cref{code:michaelscott:read-head} reads in \code{head} from the shared pointer \code{Head}.
The dequeue operation will access (dereference) \code{head}.
Hence, it has to make sure that the referenced node remains allocated.
To do so, a protection of \code{head} is issued in \Cref{code:michaelscott:protect-head}.
However, the node pointed to by \code{head} may have been dequeue and retired since \code{head} was read.
To ensure that the protection is successful, that is, not too late, \Cref{code:michaelscott:check-head} restarts the dequeue operation in case \code{head} no longer coincides with \code{Head}.
The remaining protections in the code follow the same principle.

Our contribution is a method for verifying lock-free data structures which use an SMR algorithm, like Michael\&Scott's queue with EBR/HP from \Cref{fig:michaelscott}.
At the heart of our method lies a type system which proves safe all pointer operations in the data structure.
In the case of hazard pointers, for instance, this requires to prove all pointer accesses appropriately protected.
Once this property is established, we show that the actual verification does not need to consider the SMR algorithm: it suffices to verify the data structure under garbage collection; the SMR function invocations can be removed altogether.
This allows the use of off-the-shelf GC verifiers.

\presection
\subsection{A Type System to Simplify Verification}
\label{sec:illustration:types}

Our main contribution is a type system a successful type check of which proves a given program free from unsafe memory operations. 
The type assigned to a pointer specifies if it is safe to access that pointer. 
The types are influenced by both the memory life cycle from \Cref{sec:introduction} and the SMR algorithm used.
In the case of hazard pointers, a pointer may be protected and thus guaranteed not to be deleted. 
Hence, the protected pointer can be accessed without precautions.
For an unprotected pointer, on the other hand, threads may need to take additional steps to guarantee that the pointer is not dangling, for instance, by establishing that it (to be precise, its address) is in the active stage.

%
\begin{wrapfigure}{r}{5.9cm}
\vspace{1mm}
\begin{tcolorbox}
\newcommand{\mklineref}[1]{\textcolor{black}{\text{\textnormal{\smaller(\ref*{#1})}}}\!}
\begin{lstlisting}[style=typingnolines,aboveskip=0pt]
   `{ $\illush$:$\gactive$ }`
$\mklineref{code:michaelscott:read-head}$   Node* head = Head;
   `{ $\illush$:$\gactive$,$\!\!$ head:$\emptyset$ }`
$\mklineref{code:michaelscott:protect-head}$   protect(head, 0);
   `{ $\illush$:$\gactive$,$\!\!$ head:$\gilluisu$ }`
$\mklineref{code:michaelscott:check-head}$   assume(head == Head);
   `{ $\illush$:$\gactive$,$\!\!$ head:$\gilluisu\wedge\gactive$ }`
   `{ $\illush$:$\gactive$,$\!\!$ head:$\gsafeaccess$ }`
$\mklineref{code:michaelscott:read-next}$   Node* next = head->next;
   `{ $\illush$:$\gactive$,$\!\!$ head:$\gsafeaccess$,$\!\!$ next:$\emptyset$ }`
$\mklineref{code:michaelscott:protect-next}$   protect(next, 1);
   `{ $\illush$:$\gactive$,$\!\!$ head:$\gsafeaccess$,$\!\!$ next:$\gilluisu$ }`
$\mklineref{code:michaelscott:check-head-again}$   assume(head == Head);
   `{ $\illush$:$\gactive$,$\!\!$ head:$\gsafeaccess$,$\!\!$ next:$\gsafeaccess$ }`
\end{lstlisting}%
   \caption{%
      Idealized typing for the non-retrying branch of \Cref{code:michaelscott:read-head,code:michaelscott:protect-head,code:michaelscott:check-head,code:michaelscott:read-tail,code:michaelscott:read-next,code:michaelscott:protect-next,code:michaelscott:check-head-again}.
   }
   \label{fig:typing-illustration-simple}
\end{tcolorbox}
\end{wrapfigure}
We illustrate our type system on the \code{dequeue} operation of Michael\&Scott's queue.
The interesting part is the typing of the local pointers \code{head} and \code{next} in \Cref{code:michaelscott:read-head,code:michaelscott:protect-head,code:michaelscott:check-head,code:michaelscott:read-tail,code:michaelscott:read-next,code:michaelscott:protect-next,code:michaelscott:check-head-again}.
The type derivation is depicted in \Cref{fig:typing-illustration-simple}.
Let us assume for the moment that the shared pointers and the nodes reachable through them are in the active stage.
We denote this by $\typeof{\illush}{\gactive}$.
It is the only type binding at the beginning of the operation.
The first assignment, \Cref{code:michaelscott:read-head}, adds a type binding for \code{head} to the type environment.
The type for \code{head} is copied from the source pointer, \code{Head}.
However, we remove $\gactive$ immediately after the assignment so that the actual type of \code{head} is $\emptyset$.
The reason for this are interfering threads: as discussed above, an interferer can dequeue and retire the node pointed to by \code{head}.
As a consequence, we cannot guarantee that \code{head} is active; we indeed need to remove $\gactive$.
Next, \Cref{code:michaelscott:protect-head} protects \code{head}.
We set the type of \code{head} to $\gilluisu$, encoding that a protection has been issued.
Remembering that \code{head} is protected is crucial for the subsequent conditional.
\Cref{code:michaelscott:check-head} tests whether \code{Head} has changed since it was read into \code{head}.
If it has not, denoted by \code{assume(head == Head)} in \Cref{fig:typing-illustration-simple}, we join the type of \code{head} with the type of \code{Head}.
That is, \code{head} receives $\gactive$.
Now, we know that the protection has been issued before the node pointed to by \code{head} has been retired.
So the hazard pointer method guarantees that the node is not deleted.
The subsequent code can access \code{head} without precautions.
We incorporate this fact into the type of \code{head} by updating it to $\gsafeaccess$, indicating that accesses are safe.
(We skip the assignment to \code{tail} from \Cref{code:michaelscott:read-tail}, it does not affect the type check.)
Next, \Cref{code:michaelscott:read-next} dereferences \code{head}.
This dereference is safe since \code{head} has type $\gsafeaccess$, it is guaranteed to be allocated.
The type assigned to \code{next} is $\emptyset$ because we do not assign types to pointers within nodes, like~\code{head->next}.
Hence, \code{next} cannot obtain any guarantees from the assignment.
\Cref{code:michaelscott:protect-next} then protects \code{next}.
Similarly to the above, we set its type to $\gilluisu$.
The following conditional, \Cref{code:michaelscott:check-head-again}, tests again if \code{Head} has changed since the beginning of the operation.
Consider the case it has not.
If we remember that \code{next} is the successor of \code{head}, we know that \code{next} references a shared reachable node.
Hence, we can assign type $\gactive$ to \code{next}.
As in the case for \code{head}, this allows us to lift the type to $\gsafeaccess$.
That is, using \code{next} in the following code becomes safe due to the conditional guaranteeing its activeness.
The remainder of the type check is then straightforward since only protected and/or shared pointers are used.

We stress that the actual SMR algorithm is a parameter to our type system---it is not limited to analyzing programs using hazard pointers.

\subsection{Data Structure Invariants in the Type System}
\label{sec:illustration:annotations}

The type check as illustrated in \Cref{sec:illustration:types} is idealized.
We assumed that we maintain type $\gactive$ for shared pointers and the nodes reachable through them.
Moreover, we assumed that \code{next} remains the successor of \code{head} during an execution of \code{dequeue}.
Such invariants of the data structure are notoriously hard to derive. 
Typically, it requires a state-space exploration of all thread interleavings to find invariants of lock-free data structures.
A major challenge in exploring the state space is the need for an effective (symbolic) way of tracking the data structure shape \cite{DBLP:conf/tacas/AbdullaHHJR13,DBLP:conf/csl/OHearnRY01,DBLP:conf/lics/Reynolds02separationlogic,DBLP:conf/concur/OHearn04,DBLP:conf/concur/Brookes04}.

We tackle the above problem as follows: we do not track the data structure shape at all, not even pointer aliases.
Instead, we require the programmer to annotate which pointers/nodes are active.
This allows the type check to rely on data structure invariants which typically cannot be found by a type system.
To free the programmer from manually proving the correctness of such annotations, we automate the correctness check.
We give an instrumentation of the program under scrutiny such that an ordinary GC verifier can discharge the invariants.
A thing to note is that the simple nature of active annotations and the ability to automatically discharge them makes it possible to find appropriate annotations fully automatically (guided by a failed type check).

Revisiting the previous example, the type environments never contain $\typeof{\illush}{\gactive}$.
To arrive at type $\gsafeaccess$ for \code{head} in \Cref{code:michaelscott:check-head} nevertheless, we annotate the \code{assume(head == Head)} statement with an invariant stating that \code{head} is active.
Then, the type derivation for \Cref{code:michaelscott:check-head} remains the same as before.
We argue that, provided the queue implementation is memory safe, there must be a code location between the protection in \Cref{code:michaelscott:protect-head} and the subsequent dereference in \Cref{code:michaelscott:read-next} where an active annotation can be placed.
To see this, assume there is no such code location.
This means \code{head} is not active in \Cref{code:michaelscott:protect-head,code:michaelscott:check-head,code:michaelscott:read-tail,code:michaelscott:read-next}.
That is, it must have been retired before the protection in \Cref{code:michaelscott:protect-head}, rendering the protect unsuccessful.
Hence, the dereference in \Cref{code:michaelscott:read-next} is unsafe, contradicting our assumption of memory safety.
For pointer \code{next}, we proceed similarly and add an active annotation to the second assumption (\Cref{code:michaelscott:check-head-again}).

With the above annotations our type system can rely on aspects of the dynamic behavior without requiring the programmer to manually take over parts of the verification.
We believe that having annotations makes the type system more versatile (compared to having none) in the sense that data structures need not satisfy implicit invariants like \emph{all shared pointers and nodes are active}. 
Moreover, relying on annotations rather than shape invariants allows for a much simpler type system.

\subsection{Supporting Different SMR Algorithms}
\label{sec:illustration:parametric}

The above illustration focuses on hazard pointers.
The actual type system we develop in \Cref{sec:type_system} does not---it is not tailored towards a specific SMR algorithm.
To achieve this degree of freedom, our type system takes as a parameter a formal description of the SMR algorithm being used.
We rely on a recent specification language for SMR algorithms \cite{DBLP:journals/pacmpl/MeyerW19}: SMR automata.
Then, our types capture the locations of the SMR automaton that can be reached after having seen a sequence of SMR calls in the program (control-flow sensitive type system).
This allows the types to track the relevant sequences of SMR calls.
Moreover, it allows them to detect when the deletion of a node is guaranteed to be deferred in order to infer type $\gsafeaccess$.


\section{Programming Model}
\label{sec:preliminaries}

We introduce concurrent shared-memory programs that employ a library for safe memory reclamation (SMR) and are annotated by invariants.
A programming construct that is new to our model are angels, ghost variables with an angelic semantics.
Angels are second-order pointers holding sets of addresses.
When typing (cf. Section~\ref{sec:type_system}), angels will help us track the protected nodes.


\subsection{Programs}
We define a core language for concurrent shared-memory programs.
Invocations to a library for safe memory reclamation and invariant annotations will be added below.
Programs $\aprog$ are comprised of statements defined by
\begin{align*}
	\astmt\;&::=\phantom{\bnf}
		\astmt; \astmt
		\bnf
		\astmt \choice \astmt
		\bnf
		\astmt^*
		\bnf
		\acom
	\\
	\acom\;&::=\phantom{\bnf}
		\apvar:=\apvarp
		\bnf
		\apvar:=\psel{\apvarp}
		\bnf
		\psel{\apvar}:=\apvarp
		\bnf
		\advar:=\dsel{\apvarp}
		\bnf
		\dsel{\apvar}:=\advar
		\bnf
		\advar:=\opof{\vecof{\mkern-1mu\advar}}
		\\&\phantom{::=}
				\bnf
		\apvar:=\malloc
		\bnf
		\assumeof{\acond}
		\bnf
		\atomicbegin
		\bnf
		\atomicend
	\\
	\acond\;&::=\phantom{\bnf}
		\apvar = \apvarp
		\bnf
		\apvar \neq \apvarp
		\bnf
		\predof{\vecof{\mkern-2mu\advar}}
	\ .
\end{align*}
We assume a strict typing that distinguishes between data variables $\advar,\advarp\in\dvars$ and pointer variables 
$\apvar,\apvarp\in\pvars$. 
Notation $\vecof{\mkern-2mu\advar}$ is short for $\advar_1,\dots,\advar_n$.
The language includes sequential composition, non-deterministic choice, and Kleene iteration.
The primitive commands include assignments, memory accesses, memory allocations, assumptions, and atomic blocks.
They have the usual meaning. 
We make the semantics of commands precise in a moment.

\paragraph{Memory}
Programs operate over addresses from $\adr$ that are assigned to pointer expressions~$\pexp$.
A pointer expression is either a pointer variable from $\pvars$ or a pointer selector $\psel{\anadr}\in\psels$.
The set of shared pointer variables accessible by every thread is $\svars\subseteq\pvars$.
Additionally, we allow pointer expressions to hold the special value $\segval\notin\adr$ denoting undefined/uninitialized pointers.
There is also an underlying data domain $\dom$ to which data expressions $\dexp=\dvars\uplus\dsels$ with $\dsel{\anadr}\in\dsels$ evaluate.
A generalization of our development to further selectors is straightforward.

The memory is a partial function $\aheap:\pexp\uplus\dexp\nrightarrow\adr\uplus\set{\segval}\uplus\dom$ that respects the typing.
The initial memory is $\heapcomput{\epsilon}$.
Pointer variables $\apvar$ are uninitialized, $\heapcomputof{\epsilon}{\apvar}=\segval$.
Data variables $\advar$ have a default value, $\heapcomputof{\epsilon}{\advar}=0$.
We modify the memory with updates $\anup$ of the form $\update{\anexp\mapsto\aval}$.
Applied to a memory $\aheap$, the result is the memory $\aheapp=\applyupdate{\aheap}{\anexp\mapsto\aval}$ defined by $\aheapp(\anexp)=\aval$ and $\aheapp(\anexpp)=\aheap(\anexpp)$ for all $\anexpp\neq\anexp$.
Below, we define computations $\tau$ which give rise to sequences of updates.
We write $\heapcomput{\tau}$ for the memory resulting from the initial memory $\heapcomput{\varepsilon}$ when applying the sequence of updates~in~$\tau$.

\paragraph{Liberal Semantics}
We define a semantics where program $\aprog$ is executed by a possibly unbounded number of threads.
In this semantics some addresses may be freed non-deterministically by the runtime environment.
This behavior will be constrained by a memory reclamation algorithm in a moment.
Formally, the liberal semantics of program $\aprog$ is the set of computations $\asem[\aprog]{X}{Y}$.
It is defined relative to two sets $Y\subseteq X\subseteq\adr$ of addresses allowed to be reallocated and freed, respectively.
A computation is a sequence $\tau$ of actions of the form $\anact=(\athread,\acom,\anup)$.
The action indicates that thread $\athread$ executes command $\acom$ that results in the memory update $\anup$.
The definition of the liberal semantics is by induction.
The empty computation is always contained, $\epsilon\in\asem[\aprog]{X}{Y}$.
Then, action $\anact$ can be appended to computation $\tau$, denoted $\tau.\anact\in\asem[\aprog]{X}{Y}$, if $\tau\in\asem[\aprog]{X}{Y}$, $\anact$ respects the control flow of $\aprog$, and one of the following holds.
\begin{description}
	\item[(Assign)]
		If $\anact=(\athread,\psel{\apvar}:=\apvarp,\update{\psel{\anadr}\mapsto\anadrp})$ then $\heapcomputof{\tau}{\apvar}=\anadr$ and $\heapcomputof{\tau}{\apvarp}=\anadrp$.
		There are similar conditions for the remaining assignments.
	\item[(Assume)]
		If $\anact=(\athread,\assumeof{\lhs=\rhs},\emptyset)$ then $\heapcomputof{\tau}{\lhs}=\heapcomputof{\tau}{\rhs}$.
		There are similar conditions for the remaining assumptions.
	\item[(Malloc)]
		If $\anact=(\athread,\apvar:=\malloc, \anup)$, then $\anup$ has the form $\update{\apvar\mapsto\anadr,\psel{\anadr}\mapsto\segval,\dsel{\anadr}\mapsto\advalue}$ so that $\anadr\in\freshof{\tau}$ or $\anadr\in\freedof{\tau}\cap Y$, and $\advalue\in\dom$.
	\item[(Free)]
		If $\anact=(\bot,\freeof{\anadr},\emptyset)$ then $\anadr\in X$.
	\item[(Atomic)]
		If $\anact=(\athread,\atomicbegin,\emptyset)$ or $\anact=(\athread,\atomicend,\emptyset)$.
\end{description}

Note that Rule (Free) may spontaneously emit $\freeof{\anadr}$, although there is no $\free$ command in the programming language.
Indeed, the $\free$ command will be issued by the memory reclamation algorithm defined in the next section (it is not part of $\aprog$).
The rule allows us to define the set of allocatable addresses for rule (Malloc) as addresses that have never been allocated in the computation, denoted by $\freshof{\tau}$, and addresses which have been freed since their last allocation,~$\freedof{\tau}$.


\subsection{Safe Memory Reclamation} 
We consider programs that manage their memory with the help of a safe memory reclamation~(SMR) algorithm.
In this setting, threads do not free their memory themselves (no explicit $\free$ command), but request the SMR algorithm to do so.
The SMR algorithm will have means of understanding whether an address is still accessed by other threads, and only execute the $\free$ when it is safe to do so.
As a consequence, the semantics of the program depends on the SMR algorithm it invokes.

The means of detecting whether an address can be freed safely depend on the SMR algorithm.
Despite the variety of techniques, it was recently observed that the behavior of major SMR algorithms can be captured by a common specification language~\cite{DBLP:journals/pacmpl/MeyerW19}: SMR~automata.\footnote{Working on compositional verification, \citet{DBLP:journals/pacmpl/MeyerW19} call them \emph{observers}.}
Intuitively, the SMR automaton models the protection protocol of its SMR algorithm, while abstracting from implementation details.
We recall SMR automata and use them to restrict the liberal semantics to the frees performed by the SMR algorithm.

\paragraph{SMR Automata}
An SMR algorithm offers a set of functions $f(\vecof{r})$ for the programmer to provide information about the intended access to the data structure, like \code{leaveQ}, \code{enterQ}, and \code{retire} in the case of EBR (cf. \Cref{sec:illustration}).
An SMR automaton, as depicted in \Cref{fig:ebrobserver}, is a finite control structure the transitions of which are labeled with these function symbols.
Additionally, each transition comes with a guard.
The guard influences the flow of control in the SMR automaton based on the actual parameters of function calls.
To distinguish the parameters, the automaton maintains a finite set of local variables storing thread identifiers and addresses.
Guards may then compare the actual parameters with the values of variables.

What makes SMR automata a useful modeling language is their compactness: complex SMR algorithms can be captured by fairly small SMR automata.
This is achieved by an interesting definition of the semantics.
SMR automata accept bad behavior, $\free$ commands that should not be executed after a sequence of SMR function calls protecting the address.

What makes SMR automata interesting for automated verification are two technical restrictions that limit their expressiveness.
First, the variable values are chosen only once, in the beginning of the computation, and never changed.
This choice is non-deterministic.
The idea is that the automaton picks some protection to track.
Second, transition guards can only compare for equality.
That this is sufficient to properly model the behavior of SMR algorithms can be explained by the fact that SMR algorithms are designed to work with very different data structures, from stacks to queues to trees.
Hence, there is no point for the SMR algorithm to store information about the data structure more specific than the equality of pointers.

Syntactically, an SMR automaton $\anobs$ is a tuple consisting of a finite set of locations, a finite set of variables, and a finite set of transitions.
There is a dedicated initial location and a number of accepting locations.
Transitions are of the form $\alocation\trans{\translab{f(\vecof{r})}{\aguard}}\alocationp$   with locations $\alocation,\alocationp$, event $f(\vecof{r})$, and guard~$\aguard$.
Events $\absevt{f}{\vecof{r}}$ consist of a type $f$ and parameters $\vecof{r}=r_1,\dots,r_n$.
The guard is a Boolean formula over equalities of variables and parameters $\vecof{r}$.

Semantically, a (runtime) state $\astate$ of the SMR automaton is a tuple $(\alocation,\varphi)$ where $\alocation$ is a location and $\varphi$ maps variables to values.
Such a state is initial if $\alocation$ is initial, and similarly accepting if $\alocation$ is accepting.
Then, $(\alocation,\varphi)\trans{f(\vecof{v})}(\alocationp,\varphi)$ is an SMR step if $\alocation\trans{\translab{f(\vecof{r})}{\aguard}}\alocationp$ is a transition and $\varphi(\aguard[\vecof{r}\mapsto\vecof{v}])$ evaluates to $\mathit{true}$.
By $\varphi(\aguard[\vecof{r}\mapsto\vecof{v}])$ we mean $\aguard$ with the variables replaced by their $\varphi$-mapped values and the formal parameters $\vecof{r}$ replaced by the actual values $\vecof{v}$.
As mentioned before, the valuation $\varphi$ is chosen non-deterministically in the beginning; it is not changed by steps.
A \emph{history} $\ahist=f_1(\vecof{v}_1)\ldots f_n(\vecof{v}_n)$ is a sequence of events.
If there are SMR steps $\astate\trans{f_1(\vecof{v}_1)}\cdots\trans{f_n(\vecof{v}_n)}\astatep$, we write $\astate\trans{\ahist}\astatep$.
If $\astatep$ is accepting, we say that $\ahist$ is accepted by $\astate$.

Acceptance in SMR automata characterizes \emph{bad behavior}, and a history $\ahist$ is said to violate $\anobs$ if there is an initial state $\astate$ and an accepting state $\astatep$ such that $\astate\trans{\ahist}\astatep$.
The specification of $\anobs$ is the set of histories that are not accepted:
\begin{align*}
	\specof{\anobs}:=\setcond{\ahist}{\forall\astate,\astatep.~\astate\trans{\ahist}\astatep\wedge\astate\text{ initial}\implies\astatep\text{ not accepting}}\; .
\end{align*}
We also use a restriction of the specification. 
The set $\freeableof{\ahist}{\anadr}$ contains those continuations $\ahistp$ of $\ahist$ so that $\ahist.\ahistp\in\specof{\anobs}$ and moreover at most address $\anadr$ is freed in $\ahistp$.
As bad behavior means executing a forbidden free, we assume accepting states can only be reached by transitions labeled with $\free$ and cannot be left.


\begin{figure}
	\begin{tcolorbox}
	\begin{subfigure}[t]{\textwidth}
		\center
		\begin{tikzpicture}[->,>=stealth',shorten >=1pt,auto,node distance=4.7cm,thick,initial text={}]
			\node [xshift=-1.6cm,yshift=.4cm,draw,thin] {$\baseobs$};
			\tikzstyle{every state}=[minimum size=1.5em]
			\tikzset{every edge/.append style={font=\footnotesize}}
			\node[accepting,state]      (C)               {\mkstatename{obs:base:final}};
			\node[initial above, state] (A) [right of=C]  {\mkstatename{obs:base:init}};
			\node[state]                (B) [right of=A]  {\mkstatename{obs:base:retired}};
			\path
				([yshift=1mm]B.west) edge node [above]{\translab{\freeof{\anadr}}{\anadr=\adrvar}} ([yshift=1mm]A.east)
				([yshift=-1mm]A.east) edge node [below]{\translab{\evt{\enterof{\retire}}{\athread,\anadr}}{\anadr=\adrvar}} ([yshift=-1mm]B.west)
				(A) edge node [above]{\translab{\freeof{\anadr}}{\anadr=\adrvar}} (C)
				;
		\end{tikzpicture}
		\subcaption{%
			SMR automaton specifying that address $\adrvar$ may be freed only if it has been retired and not freed since.
			The automaton uses one variable $\adrvar$.
		}
		\label{fig:baseobs}
	\end{subfigure}
	\\[2mm]
	\begin{subfigure}[t]{\textwidth}
		\center
		\begin{tikzpicture}[->,>=stealth',shorten >=1pt,auto,node distance=3.8cm,thick,initial text={}]
			\node [xshift=-1.2cm,yshift=.6cm,draw,thin] {$\ebrobs$};
			\node[initial,state]    (A)              {\mkstatename{obs:ebr:init}};
			\node[state]            (E) [right of=A] {\mkstatename{obs:ebr:protected}};
			\node[state]            (C) [right of=E] {\mkstatename{obs:ebr:retired}};
			\node[accepting,state]  (D) [right of=C] {\mkstatename{obs:ebr:final}};
			\coordinate             [below of=A, yshift=+2.9cm]  (X)  {};
			\coordinate             [below of=E, yshift=+2.9cm]  (Y)  {};
			\coordinate             [below of=B, yshift=+2.9cm]  (Z)  {};
			\path
				(A) edge node[align=center] {\translabbr{\evt{\exitof{\leaveQ}}{\athread}}{\athread=\threadvar}} (E)
				(E) edge node[align=center] {\translabbr{\evt{\enterof{\retire}}{\athread,\anadr}}{\anadr=\adrvar}} (C)
				(C) edge node[align=center] {\translabbr{\freeof{\anadr}}{\anadr=\adrvar}} (D)
				(Y) edge[-,shorten >=0pt] node {
						\translab{\evt{\enterof{\enterQ}}{\athread}}{\athread=\threadvar}
					} ([xshift=-1mm]X)
				(E) edge[-,shorten >=0pt] ([yshift=1.5mm]Y.north)
				(C.south west) edge[-,shorten >=0pt] (Y)
				([xshift=-1mm]X) edge ([xshift=-1mm]A.south)
				([yshift=1.5mm]Y) edge[-,shorten >=0pt] ([xshift=1mm,yshift=1.5mm]X)
				([xshift=1mm,yshift=1.5mm]X) edge ([xshift=1mm]A.south)
				;
		\end{tikzpicture}
		\subcaption{%
			SMR automaton characterizing when EBR defers frees, using two variables $\threadvar$ and $\adrvar$.
			It states that address $\adrvar$ must not be freed if it was retired while $\threadvar$ is in-between $\leaveQ$ and $\enterQ$ calls.
		}
		\label{fig:ebrobs}
	\end{subfigure}
	\caption{%
		Epoch-based reclamation (EBR) is specified by the SMR automaton $\baseobs\times\ebrobs$.
		For legibility, we omit self loops on all locations for the events that are not given.
	}
	\label{fig:ebrobserver}
	\end{tcolorbox}
\end{figure}

To give an example, consider the SMR automaton $\baseobs\times\ebrobs$ from \Cref{fig:ebrobserver}.
It formalizes the informal specification of EBR from \Cref{sec:illustration}.
Automaton $\baseobs$, \Cref{fig:baseobs}, forbids an EBR implementation to free addresses that have not yet been retired or have not been retired since their last free.
Put differently, it forbids spurious frees and double-frees.
Automaton $\ebrobs$, \Cref{fig:ebrobs}, requires the EBR implementation to defer the $\free$ of retired nodes which could still be accessed by some thread.
A thread can still access the retired node if it has acquired a pointer to the node before it was retired (following the usage policy of EBR discussed in \Cref{sec:illustration}).
This is the case if the thread started accessing the data structure before the $\retire$, which it announces via a~call~to~$\leaveQ$.

While every SMR implementation has its own SMR automaton, the practically relevant SMR automata are products\footnote{%
	The product operation on SMR automata is defined as expected and leads to an intersection of the specifications.
} of $\baseobs$ with further SMR automata~\cite{DBLP:journals/pacmpl/MeyerW19}, like for EBR in the above example.
Our development relies on this.

We also assume that the SMR automaton has two distinguished variables $\threadvar$ and $\adrvar$. 
Intuitively, variable $\threadvar$ will store the thread for which the SMR automaton tracks the protection of the address stored in~$\adrvar$. 
All SMR algorithms we know can be specified with only two variables. 
A possible explanation is that SMR algorithms do not seem to use helping~\cite{HS08} to protect pointers. 
We are not aware of an SMR algorithm where the protection of an address would be inferred from communication with another address or, more ambitiously, another thread.

Moreover, we inherit from \cite{DBLP:journals/pacmpl/MeyerW19} the natural requirement that SMR algorithms do not remember addresses that have been freed in order to detect (and react to) reallocations.
Formally, an SMR automaton \emph{supports elision} if for all histories $\ahist$ the behavior on address $\anadr$ after~$\ahist$
\begin{inparaenum}[(i)]
	\item
		is not affected by a free of another address~$\anadrp$, $\freeableof{\ahist.\freeof{\anadrp}}{\anadr}=\freeableof{\ahist}{\anadr}$,
	\item
		is not affected by renaming another two addresses $\anadrp$ and~$\anadrpp$, $\freeableof{\ahist}{\anadr}=\freeableof{\renamingof{\ahist}{\anadrp}{\anadrpp}}{\anadr}$,
	\item
		is included in the behavior on $\anadr$ after another history $\ahistp$ provided $\anadr$ is fresh after $\ahistp$, $\freeableof{\ahist}{\anadr}\subseteq\freeableof{\ahistp}{\anadr}$, and
	\item
		contains the behavior on $\anadr$ after $\ahist.\freeof{\anadr}$, $\freeableof{\ahist.\freeof{\anadr}}{\anadr}\subseteq\freeableof{\ahist}{\anadr}$.
		To understand (iv), note that the task of the SMR algorithm is to protect addresses from being freed.
		Hence it is safe to delay frees.
\end{inparaenum}

For convenience, we summarize our assumptions on SMR automata. 
All SMR automata we encountered, including the ones from \cite{DBLP:journals/pacmpl/MeyerW19}, satisfy them. 
\begin{assumption}
	\label{assumption:observers-accepting-state}
	\label{Assumption:Product}
	\label{Assumption:Elision}
	SMR automata
	\begin{inparaenum}[(i)]
		\item reach accepting states only with \textnormal{$\free$} and do not leave them, 
		\item are products with $\baseobs$,
		\item have distinguished variables $\threadvar$ and $\adrvar$, and
		\item support elision.
	\end{inparaenum}
\end{assumption}


It will be convenient to have a post-image $\lpostof{\apvar}{\acom}{\alocset}$ on the locations of SMR automata. 
The post-image yields a set of locations $\alocsetp$ reachable by taking a $\acom$-labeled transition from~$\alocset$.
The considered transition is restricted in two ways.
First, its guard $\aguard$ must allow $\threadvar$ to track thread~$\athread$ executing $\acom$.
Second, if $\apvar$ appears as a parameter in~$\acom$, then guard $\aguard$ must allow $\adrvar$ to track $\apvar$.
Formally, these requirements translate to the satisfiability of $\aguard\wedge\athread=\threadvar$ and $\aguard\wedge\apvar=\adrvar$, respectively.
The parameterization in $\apvar$ makes the post-image precise. 
For an example, consider $\baseobs$ and the command $\acom=\enterof{\retireof{\apvar}}$.
We expect the post-image of $\ref{obs:base:init}$ wrt. $\acom$ and $\apvar$ to be $\lpostof{\apvar}{\acom}{\ref{obs:base:init}}=\set{\ref{obs:base:retired}}$. 
The address has definitely been retired. 
Without the parametrization in $\apvar$, we would get $\set{\ref{obs:base:init},\ref{obs:base:retired}}$. The transition could choose not to track $\apvar$.

\paragraph{SMR Semantics}
To incorporate SMR automata into our programming model, we generalize the set of program commands $\acom$ to include calls to and returns from SMR functions:
\begin{align*}
		\acom\;&::=\phantom{\bnf}\acom\bnf \enterof{\afuncof{\vecof{\apvar},\vecof{\advar}}} \bnf \exitof{\afunc}
		\ .
\end{align*}
We add corresponding actions to the liberal semantics $\asem[\aprog]{X}{Y}$.
They make visible the function call/return but do not lead to memory updates.
\\[1mm]\begin{minipage}{\textwidth}
	\begin{minipage}{.48\textwidth}
		\begin{description}
			\item[(Enter)]
				$\anact=(\athread,\enterof{\afuncof{\vecof{\apvar},\vecof{\advar}}},\emptyset)$.
		\end{description}
	\end{minipage}
	\hfill
	\begin{minipage}{.48\textwidth}
		\begin{description}
			\item[(Exit)]
				$\anact=(\athread,\exitof{\afunc},\emptyset)$.
		\end{description}
	\end{minipage}
\end{minipage}\vspace{1mm}

To use SMR automata in the context of computations, we convert a computation~$\tau$ into a history~$\ahist$ by projecting $\tau$ to the $\enter$, $\exit$, and $\free$ commands and replacing the formal parameters with the actual values.
To be precise, we use as events the function names offered by the SMR algorithm plus $\free$.
The parameters to an event are the values of the actual parameters as well as the executing thread.
In the case of $\exit$ events, we drop the actual parameters and in case of $\free$ events we drop the executing thread.
For example, $\historyof{\tau.(\athread,\enterof{\afuncof{\apvar}},\emptyset)}=\historyof{\tau}.\afuncof{\athread,\heapcomputof{\tau}{\apvar}}$.

The SMR semantics of a program is the restriction of the liberal semantics to the specification of the SMR automaton of interest.
More precisely, given an SMR automaton $\anobs$ and sets $Y\subseteq X \subseteq \adr$ of reallocatable and freeable addresses, the SMR semantics induced by $\anobs, X, Y$ of program $\aprog$ is \[
	\anobs\asem[\aprog]{X}{Y}\;:=\;\setcond{\tau}{\tau\in \asem[\aprog]{X}{Y}\wedge \historyof{\tau}\in \specof{\anobs}}
	\ .
\]

SMR algorithms only restrict the execution of free commands, their functions can always be invoked by the program.
SMR automata mimic this by including in their specification all histories that do not respect the control flow.
In particular, we have the following property.
In the absence of frees, the SMR automaton does not play a role.
The resulting semantics, $\nosem$, is garbage~collection~(GC).
\begin{lemma}
	\label{Lemma:GC}
	$\anobs\nosem=\nosem$ for every SMR automaton $\anobs$.
\end{lemma} 
To see the lemma, note that only accepting states in $\anobs$ may rule out computations from $\nosem$.
By \Cref{assumption:observers-accepting-state}, only events $\freeof{\anadr}$ may lead to such accepting states.


Reconsider the SMR automaton $\baseobs$. 
For this automaton to properly restrict the frees in a program, the program should not perform \emph{double retires}, that is, not retire an address again before it is freed.
The point is that SMR algorithms typically misbehave after a double retire (perform double frees), which is not reflected in $\baseobs$ (it does not allow for double frees after a double retire).
Our type system will establish the absence of double retires for a given program. 


\subsection{Angels}
Angels can be understood as ghost variables with an angelic semantics. 
Like for ghosts, their purpose is verification: angels store information about the computation that can be used in invariants but that cannot be used to influence the control flow. 
This information is a set of addresses, which means angels are second-order pointers. 
The set of addresses is determined by an angelic choice, a non-deterministic assignment that is beneficial for the future of the computation. 

The idea behind angels is the following. 
When typing, some invariants of the runtime behavior may not be deducible by the type system. 
Angels allow the programmer to make them explicit in the program and thus available to the type check.
Consider EBR's \code{leaveQ} function.
It guarantees that all currently active addresses remain allocated, i.e., will not be reclaimed even if they are retired.
An angelic choice is convenient for selecting the set.
Subsequent dereferences can then use invariant annotations to ensure that the dereferenced pointer holds an address in the set captured by the angel.
With this, our type system is able to detect that the access is safe.

To incorporate angels and invariant annotations into our programming model, we generalize the set of commands as follows \[
	\acom\;::=\phantom{\bnf}\acom\bnf
	\ghostof{\chooseof{\aghostvar}}
	\bnf
	\invariantof{\apvar=\apvarp}
	\bnf
	\invariantof{\containsof{\apvar}{\aghostvar}}
	\bnf 		
	\invariantof{\activeof{\apvar}}
	\bnf 		
	\invariantof{\activeof{\aghostvar}}
	\ .
\]
Angels are local variables $\aghostvar$ from the set $\gvars$.
Invariant annotations include allocations of angels with the keyword $\chooseof{\aghostvar}$.  
Intuitively, this will map the angel to a set of addresses. 
Conditionals behave as expected.
The membership assertion $\containsof{\apvar}{\aghostvar}$ checks that the address of $\apvar$ is included in the set of addresses held by the angel $\aghostvar$.
The predicate $\activeof{\apvar}$ expresses that the address pointed to by $\apvar$ currently is neither freed nor retired, and similar for $\activeof{\aghostvar}$.
We use $\apavar$ to uniformly refer to pointers $\apvar$ and angels $\aghostvar$.

In the liberal semantics $\asem[\aprog]{X}{Y}$, the above commands do not lead to memory updates:
\begin{description}
	\item[(Invariant)]
		$\anact=(\athread,\invariantof{\bullet}\:,\emptyset)$.
\end{description}

\begin{figure}
	\begin{tcolorbox}
	\setalignspacingtofiguremode
	\begin{align*}
		\invholdsof{\tau}\; &:=\; \invholdsof[\epsilon]{\tau} \\\smallskip
		\invholdsof[\sigma]{\epsilon}\; &:=\; \mathit{true} \\
		\invholdsof[\sigma]{\anact.\tau}\; &:=\; \exists {\aghostvar}.\invholdsof[\sigma.\anact]{\tau}&&\text{if }\anact=(\athread,\ghostof{\chooseof{\aghostvar}},\emptyset)\\
		\invholdsof[\sigma]{\anact.\tau}\; &:=\; \heapcomputof{\sigma}{\acond}\wedge \invholdsof[\sigma.\anact]{\tau}&&\text{if }\anact=(\athread,\invariantof{\acond},\emptyset)\\
		\invholdsof[\sigma]{\anact.\tau}\; &:=\; \heapcomputof{\sigma}{\apvar}\in\aghostvar \wedge \invholdsof[\sigma.\anact]{\tau} &&\text{if }\anact=(\athread,\ghostof{\containsof{\apvar}{\aghostvar}},\emptyset)\\
		\invholdsof[\sigma]{\anact.\tau}\; &:=\; \heapcomputof{\sigma}{\apvar}\in\activeofcomp{\sigma} \wedge \invholdsof[\sigma.\anact]{\tau}&&\text{if }\anact=(\athread,\invariantof{\activeof{\apvar}},\emptyset)\\
		\invholdsof[\sigma]{\anact.\tau}\; &:=\; \aghostvar\subseteq\activeofcomp{\sigma}\wedge \invholdsof[\sigma.\anact]{\tau} &&\text{if }\anact=(\athread,\ghostof{\activeof{\aghostvar}},\emptyset)\\
		\invholdsof[\sigma]{\anact.\tau}\; &:=\; \invholdsof[\sigma.\anact]{\tau} &&\text{otherwise.}
	\end{align*}
	\caption{Formula capturing the correctness of invariant annotations in a computation $\tau$.}
	\label{Figure:FormulaInvariants}
	\end{tcolorbox}
\end{figure}
Invariant annotations behave like assertions, they do not influence the semantics but it has to be verified that they hold for all computations. 
To make precise what it means for invariant annotations to hold for a computation $\tau$, 
we construct a formula $\invholdsof{\tau}$. 
The invariant annotations are defined to hold for $\tau$ iff $\invholdsof{\tau}$ is valid.
The construction of the formula is given in \Cref{Figure:FormulaInvariants}. 
There, $\activeofcomp{\sigma}$ is the set of addresses that are neither freed nor retired after computation $\sigma$.
We only consider programs leading to closed formulas, meaning every angel is allocated (and hence quantified) before it is used.
The semantics of the formula is as expected: angels evaluate to sets of addresses, equality of addresses is the identity, and membership is as usual for sets.
\Cref{sec:invariant-check} shows how to automatically prove the correctness of invariant annotations for all computations.


\section{Getting Rid of Memory Reclamation}
\label{sec:prf}

Despite the compact formulation of SMR algorithms as SMR automata, analyzing programs in the presence of memory reclamation remains difficult.
Unlike for programs running under garbage collection, ownership guarantees~\cite{DBLP:conf/sas/Boyland03,DBLP:conf/popl/BornatCOP05} and the resulting thread-local reasoning techniques \cite{DBLP:conf/csl/OHearnRY01,DBLP:conf/lics/Reynolds02separationlogic,DBLP:conf/concur/OHearn04,DBLP:conf/concur/Brookes04} do not apply.
\Citet{DBLP:journals/pacmpl/MeyerW19} bridge this gap.
They show that it is sound and complete to conduct the verification under garbage collection provided the program properly manages its memory.
So one can establish this requirement and then perform the actual verification under the simpler semantics.
Their statement is as follows; we give the missing definitions in a moment.
\begin{theorem}[Consequence of Theorem 5.20 in \cite{DBLP:journals/pacmpl/MeyerW19}]
	\label{thm:PRF-guarantee}
	If the semantics $\freesemobs$ is pointer-race-free, then $\allsemobs$ corresponds to $\nosem$.
\end{theorem}
With the above \namecref{thm:PRF-guarantee}, the only property to be checked in the presence of memory reclamation is the premise of pointer race freedom. 
However, \citet{DBLP:journals/pacmpl/MeyerW19} report on this task as being rather challenging, requiring an intricate state space exploration of a semantics much more complicated than garbage collection.
The contribution of the present paper is a type system to tackle exactly this challenge (cf. \Cref{sec:type_system}).

We elaborate on pointer races and the correspondence between the semantics.

\paragraph{Pointer Race Freedom}
Pointer races generalize the concept of memory errors by taking into account the SMR algorithm \cite{DBLP:conf/vmcai/HazizaHMW16,DBLP:journals/pacmpl/MeyerW19}.
A memory error is an access through a dangling pointer, a pointer to an address that has been freed.
Such accesses are prone to system crashes, for example, due to {\segfault}s.
Indeed, the \cppeleven standard considers programs with memory errors to have an undefined semantics (catch-fire semantics)~\cite{iso-cppeleven}.

To make precise which pointers in a computation are dangling, \citet{DBLP:conf/vmcai/HazizaHMW16} introduce the notion of \emph{validity}.
A pointer is then dangling if it is invalid.
Initially, all pointers are invalid.
A pointer is rendered valid if it receives its value from an allocation or from a valid pointer.
A pointer becomes invalid if its address is freed or it receives its value form an invalid pointer.
It is worth pointing out that $\freeof{\anadr}$ invalidates all pointers to address $\anadr$ but a subsequent reallocation of $\anadr$ validates only the receiving pointer.
We denote the set of valid pointers after a computation~$\tau$~by~$\validof{\tau}$.

We already argued that dereferences of invalid pointers may lead to system crashes.
Consequently, passing invalid pointers to the SMR algorithm may also be unsafe.
Consider a call to $\retireof{\apvar}$ requesting the SMR algorithm to free the address of $\apvar$.
If $\apvar$ is invalid, then its address has already been freed, resulting in a system crash due to a double free.
Yet, we cannot forbid invalid pointers from being passed to SMR functions altogether.
For instance, $\guard$ may be invoked with invalid pointers in \Cref{code:michaelscott:protect-tail,code:michaelscott:protect-head} of Michael\&Scott's queue from \Cref{sec:illustration}.
To support such calls, one deems a command $\enterof{\afuncof{\vecof{\apvar},\vecof{\advar}}}$ unsafe, if replacing the actual values of invalid pointer arguments with arbitrary values may exhibit new (and potentially undesired) SMR behavior.
We inherit the the formal definition from \citet{DBLP:journals/pacmpl/MeyerW19} as it is an integral part of their proof strategy.
\begin{definition}[Definition 5.12 in \citet{DBLP:journals/pacmpl/MeyerW19}]
	\label{def:unsafe-enter}
	Consider a computation $\tau$ with history~$\ahist$. 
	A subsequent action $\anact$ is an \emph{unsafe call} if its command is $\enterof{\afuncof{\vecof{\apvar},\vecof{\advar}}}$ with $\apvar_i\notin\validof{\tau}$ for some $i$, ${\heapcomputof{\tau}{\vecof{\apvar}}=\vecof{\anadr}}$, $\heapcomputof{\tau}{\vecof{\mkern-1mu\advar\mkern+1mu}}=\vecof{\advalue}$, and:
	\begin{align*}
		\exists\,\anadrpp~
		\exists\,\vecof{\anadrp}.~~
		\big(
			\forall i.~ 
			(
			\anadr_i=\anadrpp\vee\apvar_i\in\validof{\tau}
			)
			\implies
			\anadr_i=\anadrp_i
		\big)
		\wedge
		\freeableof{\ahist.\afunc(\athread,\vecof{\anadrp},\vecof{\advalue})}{\anadrpp}
		\not\subseteq
		\freeableof{\ahist.\afunc(\athread,\vecof{\anadr},\vecof{\advalue})}{\anadrpp}
		\ .
	\end{align*}
\end{definition}
\begin{definition}[Following Definition 5.13 in \citet{DBLP:journals/pacmpl/MeyerW19}]
	\label{def:prf-brief}
	A computation $\tau.\anact$ is a \emph{pointer race} if $\anact$
	\begin{inparaenum}[(i)]
		\item dereferences an invalid pointer, 
		\item is an assumption comparing an invalid pointer for equality,
		\item retires an invalid pointer, or
		\item is an unsafe call.
	\end{inparaenum}
\end{definition}

\paragraph{Correspondence}
\Cref{thm:PRF-guarantee} establishes a correspondence between the behavior of full $\anobs\smash{\allsem}$ and the simpler, garbage collected semantics $\nosem$.
It states that we find for every computation $\tau\in\anobs\allsem$ another computation $\sigma\in\nosem$ such that $\sigma$ mimics $\tau$.
We denote this by $\tau\computrel\sigma$.
Relation $\computrel$ requires $\tau$ and $\sigma$ to agree on the control locations of all threads and the valid memory of $\tau$.
Intuitively, this means that any pointer-race-free action after $\tau$ has the same effect after $\sigma$ because the action cannot access the invalid part of the memory without raising a pointer race.
Put differently, threads cannot distinguish whether they execute in $\tau$ or in $\sigma$.
So they cannot distinguish whether or not memory is reclaimed.

Technically, $\tau$ and $\sigma$ agree on the valid memory of $\tau$ if $\restrict{\heapcomput{\tau}}{\validof{\tau}}=\restrict{\heapcomput{\sigma}}{\validof{\tau}}$.
Here, $\restrict{\heapcomput{\tau}}{\validof{\tau}}$ denotes the restriction of $\heapcomput{\tau}$ to its valid part $\validof{\tau}$.
It restricts the domain of $\heapcomput{\tau}$ to $\validof{\tau}$ and to data variables and to data selectors of addresses referenced from $\validof{\tau}$.
It is worth pointing out the asymmetry in the definition of $\tau\computrel\sigma$: $\heapcomput{\sigma}$ is restricted to $\validof{\tau}$.
This is necessary because there are no $\free$ commands in $\sigma$ and thus pointer expressions that are invalidated in $\tau$ are never invalidated in $\sigma$.
The correspondence is precise enough for verification results of safety properties to carry over from one semantics to the other. 


\presection
\section{A Type System to Prove Pointer Race Freedom}
\label{sec:type_system}

We present a type system a successful type check of which entails pointer race freedom as required by \Cref{thm:PRF-guarantee}.
The guiding idea of our types is to under-approximate the validity of pointers.
To achieve this, our types incorporate the SMR algorithm in use and the guarantees it provides.
It~does so in a modular way: a parameter of the type system definition is an SMR automaton specifying the SMR algorithm.

A key design decision of our type system is to track no information about the data structure shape.
Instead, we deduce runtime specific information from automatically dischargeable annotations.
We still achieve the necessary precision because the same SMR algorithm may be used with different data structures.
Hence, shape information should not help tracking its behavior.


\subsection{Overview}
\label{sec:type_system:overview}

Towards a definition of our type system, recall the memory life cycle from \Cref{sec:introduction}.
The transition from the active to the retired stage requires care.
The type system has to detect that a thread is guaranteed safe access to a retired node.
This means finding out that an SMR protection was successful.
Additionally, types need to be stable under interference.
Nodes can be retired without a thread noticing.
Hence, types need to ensure that the guarantees they provide cannot be spoiled by actions of other threads.

To tackle those problems, we use intersection types capturing which \emph{access guarantees} a thread has for each pointer.
We point out that this means we track information about nodes in memory through pointers to them.
We use the following guarantees.

\begin{description}[leftmargin=7mm,labelwidth=*]
	\item[\textnormal{$\glocal$:}]
		Thread-local pointers referencing nodes in the local stage. 
		The guarantee comes with two more properties. 
		There are no valid aliases of the pointer and the referenced node is not retired.
		This gives the thread holding the pointer exclusive access.

	\item[\textnormal{$\gactive$}]
		Pointers to nodes in the active stage.
		Active pointers are guaranteed to be valid, they can be accessed safely.

	\item[\textnormal{$\gsafeaccess$}]
		Pointers to nodes which are protected by the SMR algorithm from being reclaimed.
		Such pointers can be accessed safely although the referenced node might be in the retired stage.

	\item[\textnormal{$\gcustom{\alocset}$}]
		SMR-specific guarantee that depends on a set of locations in the given SMR automaton. 
		The idea is to track the history of SMR calls performed so far. 
		This history is guaranteed to reach a location in $\alocset$. 
		The information about $\alocset$ bridges the (SMR-specific) gap between $\gactive$ and $\gsafeaccess$. 
		Accesses to the pointer are potentially unsafe. 
\end{description}

The interplay among these guarantees tackles the aforementioned challenges as follows.
Consider a thread that just acquired a pointer $\apvar$ to a shared node.
In the case of hazard pointers, this pointer comes without access guarantees.
Hence, the thread issues a protection of $\apvar$.
We denote this with an SMR-specific type $\gcustom{}$.
For the protection to be successful, the programmer has to make sure that $\apvar$ is active during the invocation.
The type system detects this through an annotation that adds guarantee $\gactive$ to $\apvar$.
We then deduce from the SMR automaton that $\apvar$ can be accessed safely because the protection was successful.
This adds guarantee $\gsafeaccess$.
(We have seen this on~an~example~in~\Cref{sec:illustration}.)


\subsection{Types}

Throughout the remainder of the section we fix an SMR automaton $\anobs$ relative to which we describe the type system.
The SMR automaton induces a set of intersection types~\cite{CD78,Pierce} defined by the following grammar:
\begin{align*}
	\atype\; ::=\; \emptyset\bnf \glocal \bnf \gactive \bnf \gsafeaccess\bnf \gcustom{\alocset}\bnf \atype\wedge\atype\ .
\end{align*}
The meaning of the guarantees $\glocal$ to $\gcustom{\alocset}$ is as explained above.
We also write a type $\atype$ as the set of its guarantees where convenient.
We define the predicate $\isvalidof{\atype}$ to hold if $\atype\cap \set{\gsafeaccess, \glocal, \gactive}\neq \emptyset$.
The three guarantees serve as syntactic under-approximations of the semantic notion of validity from the definition of pointer races (cf. \Cref{sec:prf}).

There is a restriction on the sets of locations $\alocset$ for which we provide guarantees $\gcustom{\alocset}$.
To understand it, note that our type system infers guarantees about the protection of pointers thread-locally from the code, that is, as if the code was sequential.
Soundness then shows that these guarantees carry over to the computations of the overall program where threads interfere.
To justify this sequential to concurrent lifting, we rely on the concept of interference freedom due to \citet{DBLP:journals/acta/OwickiG76}.
A set of locations $\alocset$ in the SMR automaton $\anobs$ is \emph{closed under interference from other threads}, if no SMR command issued by a thread different from $\threadvar$ (whose protection we track) can leave the locations.
Formally, we require that for every transition $\alocation\trans{\translab{f(\athreadp, *)}{\aguard}}\alocationp$ with $\alocation\in\alocset$ and $\alocationp\notin\alocset$ we have guard $\aguard$ implying $\athreadp=\threadvar$.
We only introduce guarantees $\gcustom{\alocset}$ for sets of locations $\alocset$ that are closed under interference from other threads.

Type environments $\env$ are total functions that assign a type to every pointer and every angel in the code being typed.
To fix the notation, $\envof{\apavar}=\atype$ or $\typeof{\apavar}{\atype}\in \env$ means $\apavar$ is assigned $\atype$ in environment~$\env$.
We write $\env,\typeof{\apavar}{\atype}$ for $\env\uplus\set{\typeof{\apavar}{\atype}}$.
If the type of $\apavar$ does not matter, we just write $\env,\typeofany{\apavar}$.
The initial type environment  $\envinit$  assigns $\emptyset$ to every pointer and angel.

Our type system will be control-flow sensitive~\cite{CWM99,DBLP:conf/pldi/FosterTA02,HS06}, which means type judgements take the form
\begin{align*}
	\typestmt{\env_\mathit{pre}}{\astmt}{\env_\mathit{post}}\ .
\end{align*}
The thing to note is that the type assigned to a pointer/angel is not constant throughout the program but depends on the commands that have been executed.
So we may have the type assignment $\typeof{\apavar}{\atype}$ in $\env_\mathit{pre}$ but $\typeof{\apavar}{\atypep}$ in the type environment $\env_\mathit{post}$ with $\atype\neq \atypep$.

Control-flow sensitivity requires us to formulate how types change under the execuction of SMR commands.
Towards a definition, we associate with every type a set of locations in $\anobs = \baseobs\times \anobs'$.
Guarantee~$\gcustom{\alocset}$ already comes with a set of locations.
Guarantee~$\gsafeaccess$ grants safe access to the tracked address.
In terms of locations, it should not be possible to free the address stored in $\adrvar$.
We define $\mathit{SafeLoc}(\anobs)$ to be the largest set of locations in the SMR automaton that is closed under interference from other threads and for which there is no transition $\alocation\trans{\translab{\mathtt{free}(\anadr)}{\aguard}}\alocationp$ with $\alocation\in\mathit{SafeLoc}(\anobs)$, $\alocationp$ not accepting, and $\aguard$ implying $\anadr=\adrvar$.
Guarantee~$\gactive$ is characterized by location $\ref{obs:base:init}$ in $\baseobs$.
Indeed, a pointer is active iff  $\baseobs$ is in its initial location.
For $\glocal$ we also use location $\ref{obs:base:init}$.
The~discussion~yields:
\begin{align*}
	\locsof{\emptyset}\; &:=\; \locsof{\anobs}&\locsof{\gcustom{\alocset}}\; &:=\; \alocset \\
	\locsof{\gactive}\; &:=\; \set{\ref{obs:base:init}}\times \locsof{\anobs'}
	&\locsof{\gsafeaccess}\; &:=\; \mathit{SafeLoc}(\anobs)\\
	\locsof{\glocal}\; &:=\; \set{\ref{obs:base:init}}\times \locsof{\anobs'}&
	\locsof{\atype_1\wedge \atype_2}\;&:=\;\locsof{\atype_1}\cap\locsof{\atype_2}\ .
\end{align*}

The set of locations associated with a type is defined to over-approximate the locations reachable in the SMR automaton by the (history of the) current computation.
With this understanding, it should be possible for command $\acom$ to transform $\typeof{\apavar}{\atype}$ into $\typeof{\apavar}{\atypep}$ if the locations associated with~$\atypep$ over-approximate the post-image under $\apavar$ and $\acom$ of the locations associated with $\atype$.
We define the \emph{type transformer relation} $\checkoftype{\atype}{\apavar}{\acom}{\atypep}$ by the following conditions:
\begin{eqnarray*}
	\lpostof{\apavar}{\acom}{\locsof{\atype}} & \subseteq & \locsof{\atypep} \\
	\isvalidof{\atypep} & \Rightarrow & \isvalidof{\atype} \\
	\set{\mkern+1mu\glocal,\gactive\mkern+1mu}\cap\atypep & \subseteq & \set{\mkern+1mu\glocal,\gactive\mkern+1mu}\cap\atype
	\ .
\end{eqnarray*}
The over-approximation of the post-image is the first inclusion.
The implication states that SMR commands cannot validate pointers.
We can, however, deduce from the fact that the address has not been retired ($\gactive$ or $\glocal$) and the SMR command has been executed, that it is safe to access the address~($\gsafeaccess$).
The last inclusion states that SMR commands cannot establish the guarantees $\glocal$ and~$\gactive$.
It is worth pointing out that the relation $\checkoftype{\atype}{\apavar}{\acom}{\atypep}$ only depends on the SMR automaton, up to a choice of variable names.
This means we can tabulate it to guarantee quick access when typing a program.
We also write $\checkof{\env}{\acom}{\envp}$ if we have $\checkoftype{\envof{\apavar}}{\apavar}{\acom}{\envpof{\apavar}}$ for all pointers/angels $\apavar$.
We write $\checknocomof{\env}{\acom}{\envp}$ if we take the post-image to be the identity.
For an example, refer to \Cref{sec:example:type-transformer}.


Guarantees $\glocal$ and $\gactive$ are special in that their sets of locations, $\locsof{\glocal}$ and $\locsof{\gactive}$, are not closed under interference.
For $\glocal$, the type rules ensure interference freedom.
They do so by enforcing that $\retire$ is not invoked with invalid pointers. 
Hence, the fact that $\glocal$-pointers have no valid aliases implies that other threads cannot retire them.
So $\baseobs$ remains in $\ref{obs:base:init}$ no matter the interference.
For~$\gactive$, the type rules account for interference.
We define an operation $\rmtransientof{\env}$ that takes an environment and removes all $\gactive$ guarantees for thread-local pointers and angels:
\begin{align*}
	\rmtransientof{\env}
	\;:=\;
	\setcond{\typeof{\apavar}{\atype\setminus\set{\gactive}}}{\typeof{\apavar}{\atype}\in\env \wedge \apavar\notin\svars}
	\:\cup\:
	\setcond{\typeof{\apavar}{\emptyset}}{\apavar\in\svars}
	\ .
\end{align*}
The operation also has an effect on shared pointers and angels where it removes all guarantees.
The reasoning is as follows.
An interference on a shared pointer or angel may change the address being pointed to.
Guarantees express properties about addresses, indirectly via their pointers.
As we do not have any information about the new address, the pointer receives the empty set of guarantees.


\subsection{Type System}
\label{sec:type_system:rules}

Our type system is given in \Cref{fig:type-rules}.
We write $\typejudge{\envinit}{\astmt}{\env}$ to indicate that $\typestmt{\envinit}{\astmt}{\env}$ is derivable with the given rules.
We write $\typechecks{\astmt}$ if there is an environment $\env$ so that $\typejudge{\envinit}{\astmt}{\env}$.
In this case, we say the program \emph{type~checks}.
Soundness will show that a type check entails pointer race freedom and the absence of double retires.

We distinguish between rules for statements and rules for primitive commands.
We assume that primitive commands $\acom$ are wrapped inside an atomic block, like $\atomicbegin;\,\acom;\,\atomicend$.
With this assumption, the rules for primitive commands need not handle the fact that guarantee $\gactive$ is not closed under interference.
Interference will be taken into account by the rules for statements.
The assumption of atomic blocks can be established by a simple preprocessing of the program.
We do not make it explicit but assume it has been applied.


\begin{figure}
	\begin{tcolorbox}
	\def\MathparLineskip{\lineskip=6.5mm}
	\begin{subfigure}[t]{\textwidth}
		\small
		\begin{mathpar}
			\infrule{assign1}
				{\atypep=\atype\setminus\set{\glocal}}
				{\typecom{\env,\typeofany{\apvar},\typeof{\apvarp}{\atype}}{\apvar:=\apvarp}{\env,\typeof{\apvar}{\atypep},\typeof{\apvarp}{\atypep}}}
			\and
			\infrule{assign2}
				{\envof{\apvarp}={\atype} \\ \isvalidof{\atype}}
				{\typecom{\env,\typeofany{\apvar}}{\apvar:=\psel{\apvarp}}{\env,\typeof{\apvar}{\emptyset}}}
			\and
			\infrule{assign3}
				{\envof{\apvar}={\atype} \\ \isvalidof{\atype} \\ \atypepp=\atypep\setminus\set{\glocal}}
				{\typecom{\env,\typeof{\apvarp}{\atypep}}{\psel{\apvar}:=\apvarp}{\env,\typeof{\apvarp}{\atypepp}}}
			\and
			\infrule{assign4}
				{~}
				{\typecom{\env}{\advar:=\opof{\vecof{\advar}}}{\env}}
			\and
			\infrule{assign5}
				{\envof{\apvarp}={\atype} \\ \isvalidof{\atype}}
				{\typecom{\env}{\advar:=\dsel{\apvarp}}{\env}}
			\and
			\infrule{assign6}
				{\envof{\apvar}={\atype} \\ \isvalidof{\atype}}
				{\typecom{\env}{\dsel{\apvar}:=\advar}{\env}}
			\and
			\infrule{malloc}
				{\apvar\notin\svars \\ \atype=\set{\glocal}}
				{\typecom{\env,\typeofany{\apvar}}{\apvar:=\malloc}{\env,\typeof{\apvar}{\atype}}}
			\and
			\infrule{assume1}
				{\isvalidof{\atype} \\ \isvalidof{\atypep} \\ \atypepp=(\atype\wedge\atypep)\setminus\set{\glocal}}
				{\typecom{\env,\typeof{\apvar}{\atype},\typeof{\apvarp}{\atypep}}{\assumeof{\apvar=\apvarp}}{\env,\typeof{\apvar}{\atypepp},\typeof{\apvarp}{\atypepp}}}
			\and
			\infrule{assume2}
				{\acond\not\equiv\apvar=\apvarp}
				{\typecom{\env}{\assumeof{\acond}}{\env}}
			\and
			\infrule{equal}
				{\atypepp=\atype\wedge\atypep}
				{\typecom{\env,\typeof{\apvar}{\atype},\typeof{\apvarp}{\atypep}}{\invariantof{\apvar=\apvarp}}{\env,\typeof{\apvar}{\atypepp},\typeof{\apvarp}{\atypepp}}}
			\and
			\infrule{active}
				{\atypep=\atype\wedge\set{\gactive}}
				{\typecom{\env,\typeof{\apavar}{\atype}}{\invariantof{\activeof{\apavar}}}{\env,\typeof{\apavar}{\atypep}}}
			\and 
			\infrule{angel}
				{\aghostvar\notin\svars}
				{\typecom{\env,\typeofany{\aghostvar}}{\ghostof{\chooseof{\aghostvar}}}{\env,\typeof{\aghostvar}{\emptyset}}}
			\and
			\infrule{member}
				{\env(\aghostvar)=\atypep \\ \atypepp=\atype\wedge\atypep}
				{\typecom{\env,\typeof{\apvar}{\atype}}{\ghostof{\containsof{\apvar}{\aghostvar}}}{\env,\typeof{\apvar}{\atypepp}}}
			\and 
			\infrule{enter}
				{
					\safecallof{\env}{\afuncof{\vecof{\apvar},\vecof{\advar}}}
					\\
					\checkof{\env}{\enterof{\afuncof{\vecof{\apvar},\vecof{\advar}}}}{\envp}
					\\\\
					\afuncof{\vecof{\apvar},\vecof{\advar}}\equiv\retireof{\apvar}
					\wedge
					\envof{\apvar}={\atype}
					\implies
					\gactive\in\atype
				}
				{\typecom{\env}{\enterof{\afuncof{\vecof{\apvar},\vecof{\advar}}}}{\envp}}
			\and
			\infrule{exit}
				{\checkof{\env}{\exitof{\afunc}}{\envp}}
				{\typecom{\env}{\exitof{\afunc}}{\envp}}
		\end{mathpar}
		\vspace{1mm}
		\subcaption{Type rules for primitive commands.}
		\label{fig:type-rules:com}
	\end{subfigure}
	\\[4mm]
	\begin{subfigure}[t]{\textwidth}
		\small
		\begin{mathpar}
			\infrule{infer}
				{\checknocomof{\env_1}{\epsilon}{\env_2}\\ \typestmt{\env_2}{\astmt}{\env_3} \\ \checknocomof{\env_3}{\epsilon}{\env_4}}
				{\typestmt{\env_1}{\astmt}{\env_4}}
			\and
			\infrule{begin}
				{~}
				{\typestmt{\env}{\atomicbegin}{\env}}
			\and
			\infrule{end}
				{~}
				{\typestmt{\env}{\atomicend}{\rmtransientof{\env}}}
			\and
			\infrule{seq}
				{\typestmt{\env}{\astmt_1}{\envp} \\ \typestmt{\envp}{\astmt_2}{\envpp}}
				{\typestmt{\env}{\astmt_1;\astmt_2}{\envpp}}
			\and
			\infrule{choice}
				{\typestmt{\env}{\astmt_1}{\envp} \\ \typestmt{\env}{\astmt_2}{\envp}}
				{\typestmt{\env}{\astmt_1\oplus\astmt_2}{\envp}}
			\and
			\infrule{loop}
				{\typestmt{\env}{\astmt}{\env}}
				{\typestmt{\env}{\astmt^*}{\env}}
		\end{mathpar}
		\vspace{1mm}
		\subcaption{Type rules for statements.}
		\label{fig:type-rules:stmt}
	\end{subfigure}
	\caption{%
		Type rules.
	}
	\label{fig:type-rules}
	\end{tcolorbox}
\end{figure}

The rules for primitive commands, \Cref{fig:type-rules:com}, that are not related to SMR are straightforward.
Rule~\ref{rule:assign1} copies the type of the right-hand side pointer to the left-hand side pointer of the assignment.
Additionally, both pointers lose their $\glocal$ qualifier since the command creates an alias.
Rule~\ref{rule:assign2} ensures that the dereferenced pointer is valid and then sets the type of the assigned pointer to the empty type.
The assigned pointer does not receive any guarantees since we do not track guarantees for selectors.
Rule~\ref{rule:assign3} checks the dereferenced pointer for validity and removes $\glocal$ from the pointer that is aliased.
Data assignments, Rules \ref{rule:assign4}, \ref{rule:assign5}, and \ref{rule:assign6}, simply check dereferenced pointers for validity.
Allocations grant the target pointer the $\glocal$ guarantee, Rule~\ref{rule:malloc}.
Assumptions of the form $\apvar=\apvarp$ check that both pointers are valid and join the type information, Rule~\ref{rule:assume1}.
Guarantee $\glocal$ is removed due to the alias.
All other assumptions have no effect on the type environment, Rule~\ref{rule:assume2}.
Similarly, Rule~\ref{rule:equal} joins type information in the case of assertions.
However, no validity check is performed and $\glocal$ is not removed.
Rule~\ref{rule:active} adds the $\gactive$ guarantee. Note that $\apavar$ is a pointer or an angel.
Angels are always local variables.
Their allocation does not justify any guarantees, in particular not $\glocal$, as they hold full sets of addresses, Rule~\ref{rule:angel}.
We can also assert membership of an address held by a pointer in a set of addresses held by an angel, Rule~\ref{rule:member}.

SMR-related commands may change the entire type environment, rather than manipulating only the pointers that occur syntactically in the command.
This is because of pointer aliasing on the one hand, and because of the SMR automaton on the other hand (for instance, $\enterQ$ has an effect on all pointers).
The post type environment of Rules~\ref{rule:enter} and \ref{rule:exit} simply infers guarantees wrt. the pre type environment and the emitted event.
Note that this is the only way to infer SMR-specific guarantees $\gcustom{\alocset}$, i.e., these guarantees solely depend on the SMR commands.
Moreover, Rule~\ref{rule:enter} performs a pointer race check as defined in \Cref{sec:prf}.
Predicate $\safecallof{\env}{\afuncof{\vecof{\apvar},\vecof{\advar}}}$ evaluates to \textit{true} iff the command $\enterof{\afuncof{\vecof{\apvar},\vecof{\advar}}}$ is guaranteed to be free from pointer races given the types in $\env$.
The formalization coincides with \Cref{def:unsafe-enter} except that it replaces  $\validof{}$ by the under-approximation $\isvalidof{\cdot}$.
A special case of Rule~\ref{rule:enter} is the invocation of $\retireof{\apvar}$, which requires the argument $\apvar$ to be active.
This will prevent double retires.

The rules for statements are given in \Cref{fig:type-rules:stmt}. Rule~\ref{rule:infer} allows for type transformations at any point, in particular to establish the proper pre/post environments for the Rules~\ref{rule:choice} and \ref{rule:loop}.
Entering an atomic block, Rule~\ref{rule:begin}, has no effect on the type environment.
Exiting an atomic block allows for interference.
Hence, Rule~\ref{rule:exit} removes any type information from the type environment that can be tampered with by other threads.
Sequences of statements are straightforward, Rule~\ref{rule:seq}.
Choices require a common pre and post type environment, Rule~\ref{rule:choice}.
Loops require a type environment that is stable under the loop body, Rule~\ref{rule:loop}.


\subsection{Soundness}
\label{sec:type_system:soundness}

Our goal is to show that a successful type check $\typechecks{\astmt}$ implies pointer race freedom and the absence of double retires.
There are two challenges.
We already commented on the problematic sequential to concurrent lifting that motivated the definition of interference freedom.
The second difficulty is that the type system relies on the program's invariant annotations.
The set of computations ignores these annotations.
To reconcile the assumptions about the program, we have to prove the invariant annotations correct.
Interestingly, we can use garbage collection for this purpose, meaning the invariant annotations only have to hold in $\nosem$, although the following results refer to the computations in $\freesemobs$.
Intuitively, garbage collection is sufficient because we have elision support (cf. \Cref{sec:preliminaries}): it allows us to remove frees from a computation and then apply \Cref{Lemma:GC}.

Pointer race freedom and the absence of double retires will be consequences of a more general soundness result that makes explicit the information tracked by our type system. 
We give some auxiliary definitions that ease the formulation.
We write $\tau\models_{\varphi} \atype$ if there is a location $\alocation\in\locsof{\atype}$ associated with the type $\atype$ so that $(\initlocation, \varphi)\trans{\historyof{\tau}}(\alocation, \varphi)$. 
The definition is parameterized in the valuation~$\varphi$ determining the thread and the address to be tracked. 
We write $\tau, \athread \models \typeof{\apavar}{\atype}$ if for every address $\anadr\in\heapcomputof{\tau}{\apavar}$ we have $\tau\models_{\varphi} \atype$, with $\varphi=\set{\threadvar\mapsto \athread, \adrvar\mapsto \anadr}$. 
The thread is given. 
The address is the one held by the pointer or among the ones held by the angel, as determined by the computation. 
We write $\tau, \athread \models \env$ if we have $\tau, \athread \models \typeof{\apavar}{\atype}$ for all type assignments $\typeof{\apavar}{\atype}\in \env$.

Soundness states that a type environment annotating a program point approximates the history of every computation reaching this point.  
Moreover, $\isvalidof{\cdot}$ approximates validity. 
To make this precise, we define the relation $\models \typestmt{\envinit}{\astmt}{\env}$.
It requires for every $\tau\in\freesemobs$  where thread $\athread$ executes $\astmt$ to completion that
\begin{inparaenum}[(i)]
	\item $\tau, \athread \models \env$ holds, and
	\item for every $\typeof{\apvar}{\atype}\in \env$ with $\isvalidof{\atype}$ we have $\apvar\in\validof{\tau}$.
\end{inparaenum}
The soundness result now lifts the syntactic derivation relation $\typechecks{\,}$ to the semantic soundness relation $\models$.

\begin{theorem}[Soundness]\label{Theorem:Soundness}
	If $\invholdsof{\nosem}$, then $\vdash \typestmt{\envinit}{\astmt}{\env}$ implies $\models \typestmt{\envinit}{\astmt}{\env}$.
\end{theorem}

\begin{proof}[Proof Sketch]
	We proceed by induction on the length of computations $\tau$ from $\freesemobs$.
	During the proof, we need to access the types encountered by thread $\athread$ along the execution of $\astmt$.
	To make them explicit, we define the straight-line version $\astmtof{\tau, \athread}$ of $\astmt$ induced by $\tau$ and $\athread$.
	It is obtained by projecting $\tau$ to the commands of thread $\athread$.
	One can show that if we can derive a typing for the program then we can derive it for the induced straight-line program: \[
		\vdash \typestmt{\envinit}{\astmt}{\env}\qquad\text{implies}\qquad\vdash \typestmt{\envinit}{\astmtof{\tau, \athread}}{\env}
		\ .
	\]
	The implication should be intuitive. 
	The typing of the overall program can be seen as an intersection over the typings of the induced straight-line programs.

	The induction hypothesis links the current type environment $\env$ derived for the straight-line program to the semantic information carried by the computation. 
	The hypothesis strengthens the requirements (i) and (ii) 
	in the definition of soundness by the following two conditions, where we assume $\envof{\apavar}={\atype}$. 
	\begin{inparaenum}[(i)]
		\setcounter{enumi}{2}
		\item
			If $\glocal\in \atype$, then $\apavar$ is a pointer that does not have valid aliases.
			That is, $\heapcomputof{\tau}{\apavar} = \heapcomputof{\tau}{\apvarp}$ implies $\apvarp\notin\validof{\tau}$.
			Note that angels cannot obtain $\glocal$ according to the type rules.
		\item
			If $\gactive\in\atype$, then thread $\athread$ is in an atomic block. 
	\end{inparaenum}
	The interesting argumentation in the induction step is in the case when another thread appends an action, $\tau.\anact$. 
	It can be summarized as follows.
	Property (i) continues to hold for $\tau.\anact$ because the type $\atype$ of $\apavar$ is closed under interference; for $\glocal$ and $\gactive$ we argue separately in the following.
	If $\glocal\in\atype$, then $\anact$ cannot use a valid alias of $\apavar$.
	In particular, it cannot retire $\apavar$ according to the premise of Rule~\ref{rule:enter}.
	If $\gactive\in\atype$, then thread $\athread$ is in an atomic block and there is no chance to append action $\anact$ of another thread.
	The case does not occur.
	Consider property (ii).
	Assume $\isvalidof{\atype}$ holds.
	That is, $\atype$ contains one of $\gactive,\glocal,\gsafeaccess$.
	If $\glocal\in\atype$ or $\gactive\in\atype$, then the above reasoning for (i) already implies~(ii).
	Otherwise, we have $\gsafeaccess\in\atype$.
	It implies (ii) because $\gsafeaccess$ is closed under interference.
	Property (iii) follows from the fact that $\anact$ cannot contain, and thus cannot create, a valid alias of $\apavar$.
	Lastly, to conclude property (iv), note that another thread cannot append an action while $\athread$ is inside an atomic block.
\end{proof}

The first consequence of soundness is that a successful type check implies pointer race freedom.
Phrased differently, the rules from \Cref{fig:type-rules} allow for a successful typing only if there are no pointer races. 
That is, our type system performs a pointer race freedom check indeed.
\begin{proposition}
	\label{proposition:typecheckprf}
	If $\invholdsof{\nosem}$, then $\typechecks{\aprog}$ implies that $\freesemobs[\aprog]$ is pointer-race-free.
\end{proposition}
The proposition gives an effective means of checking the premise of \Cref{thm:PRF-guarantee}: determine a typing using the proposed type system (cf. \Cref{sec:type_checking}) and discharge the invariant annotations using an off-the-shelf verification tool (cf. \Cref{sec:invariant-check}).

\begin{proof}[Proof Sketch]
	To see the proposition, consider $\tau.\anact\in\anobs\freesem[\aprog]$.
	We focus on the case where the last action is a dereference, say due to command $\acom$ being $\apvar:=\psel{\apvarp}$.
	The remaining cases in the definition of pointer races are similar.
	We show that the dereference is safe, $\apvarp\in\validof{\tau}$.
	Let thread~$\athread$ perform the dereference.
	Let $\astmtof{\tau.\anact, \athread} = \astmt;\acom$ be the induced straight-line program.
	As observed above, since the program type checks also $\astmt;\acom$ will type check.
	Say we can derive $\typestmt{\envinit}{\astmt;\acom}{\env}$.
	The only way to type a composition $\astmt;\acom$ is Rule~\ref{rule:seq}.
	It requires an environment $\envp$ so that $\typestmt{\envinit}{\astmt}{\envp}$ and $\typestmt{\envp}{\acom}{\env}$ are derivable.
	The only way to type an assignment $\apvar:=\psel{\apvarp}$ is Rule~\ref{rule:assign2}.
	By its premise, $\envp(\apvarp)=\atype$ with~$\isvalidof{\atype}$.
	\Cref{Theorem:Soundness} yields $\apvarp\in \validof{\tau}$.
	The dereference of $\apvarp$ is safe.
\end{proof}

The second consequence of soundness is that a successful type check means the program does not perform double retires.
This is the precondition for a meaningful application of $\baseobs$ (cf. \Cref{sec:preliminaries}).
\begin{proposition}
	\label{proposition:retire-check}
	If $\invholdsof{\nosem}$, then $\typechecks{\aprog}$ implies that $\allsemobs[\aprog]$ does not perform double retires.
\end{proposition}
The argumentation is along the lines of \Cref{proposition:typecheckprf}.
To perform a retire, Rule~\ref{rule:enter} requires the pointer to be active.
This, in turn, means $\baseobs$ is in state \ref{obs:base:init}.
The state, however, can only be reached if there were no earlier retires of the address or the earlier retires have been followed by a free.
In both cases, we do not have a double retire.

The next section gives an in-depth example on how to apply our type system.
The two sections thereafter automate the checks in \Cref{Theorem:Soundness}:
we give an efficient algorithm for type inference $\typechecks{\aprog}$ and show how to discharge the invariants $\invholdsof{\nosem[\aprog]}$ with the help of \mbox{off-the-shelf verification tools.}


\section{Example}
\label{sec:example}

We apply our type system to Michael\&Scott's queue with EBR from \Cref{sec:illustration}.
Here, a single custom guarantee $\gilluebr$ is sufficient.
We define $\locsof{\gilluebr}$ to be those locations where thread $\threadvar$ is guarantee to have returned from a call to $\leaveQ$ but has not yet invoked $\enterQ$.
That is, $\gilluebr$ captures when $\threadvar$ is accessing the data structure.
The sets of locations represented by $\gactive$, $\gsafeaccess$, and $\gilluebr$ can be read of the cross-product SMR automaton $\baseobs\times\ebrobs$ in \Cref{fig:ebr-types}.
It is worth pointing out that $\locsof{\gsafeaccess}$ does not contain location $(\ref*{obs:base:init},\ref*{obs:ebr:init})$.
For a set containing $(\ref*{obs:base:init},\ref*{obs:ebr:init})$ to be closed under interference we would need to have $(\ref*{obs:base:retired},\ref*{obs:ebr:init})$ in that set.
However, $(\ref*{obs:base:retired},\ref*{obs:ebr:init})$ allows for a $\free$ of $\adrvar$ and thus must not belong to $\locsof{\gsafeaccess}$ by definition.

%
\begin{figure}
	\begin{tcolorbox}
	\center
	\definecolor{colorF}{RGB} {227, 39, 119}
	\definecolor{colorR}{RGB} {41, 83, 186}
	\definecolor{colorA}{RGB} {208, 176, 227}
	\definecolor{colorE}{RGB} {214, 216, 145}
	\definecolor{colorS}{RGB} {136, 187, 146}
	\definecolor{colorMix}{RGB} {200, 200, 200}

	\begin{tikzpicture}[->,>=stealth',shorten >=1pt,auto,node distance=5.3cm,thick,initial text={}]
		\node [xshift=.0cm,yshift=4.8cm,draw,thin] {$\baseobs\times\ebrobs$};

		\node[initial,state,text width=8mm,align=center,yshift=-3mm,yshift=.5cm]    (A)              {$\ref*{obs:base:init},\ref*{obs:ebr:init}$};
		\node[state,text width=8mm,align=center]            (E) [right of=A] {$\ref*{obs:base:init},\ref*{obs:ebr:protected}$};
		\node[accepting,state,text width=8mm,align=center,double=colorMix]  (D) [right of=E] {\mkstatename{obs:base-ebr:final}};
		\coordinate             [below of=A, yshift=+4.2cm]  (X)  {};
		\coordinate             [below of=E, yshift=+4.2cm]  (Y)  {};

		\node[state,text width=8mm,align=center]            (F) [above of=A,yshift=-3.2cm]  {$\ref*{obs:base:retired},\ref*{obs:ebr:init}$};
		\node[state,text width=8mm,align=center]            (J) [right of=F] {$\ref*{obs:base:retired},\ref*{obs:ebr:protected}$};
		\node[state,text width=8mm,align=center]            (H) [right of=J] {$\ref*{obs:base:retired},\ref*{obs:ebr:retired}$};
		\coordinate             [above of=F, yshift=-4.3cm]  (U)  {};
		\coordinate             [above of=J, yshift=-4.3cm]  (V)  {};

		\path
			(H) edge[colorF] node[above right,align=center] {F} (D)
			(A.324) edge[colorF, bend right=7] node[align=center,pos=.7] {F} (D.216)
			(E) edge[colorF] node {F} (D)
			(F) edge[colorF, bend right] node[align=center] {F} (A)
			(J) edge[colorF] node[align=center] {F} (E)
			;

		\path
			(E) edge[colorR] node[below right,align=center,pos=.66] {R} (H)
			(J) edge[colorR] node[below,align=center] {R} (H)
			(A) edge[colorR, bend right] node[align=center] {R} (F)
			;

		\path
			(A) edge node[align=center] {\translab{\evt{\exitof{\leaveQ}}{\athread}}{\athread=\threadvar}} (E)
			(E) edge[-,shorten >=0pt] ([yshift=1mm]Y.north)
			([yshift=1mm]Y) edge[-,shorten >=0pt] node {
					\translab{\evt{\enterof{\enterQ}}{\athread}}{\athread=\threadvar}
				} ([xshift=0mm,yshift=1mm]X)
			([xshift=0mm,yshift=1mm]X) edge ([xshift=0mm]A.south)
			;

		\path
			(F) edge node[align=center] {\translab{\evt{\exitof{\leaveQ}}{\athread}}{\athread=\threadvar}} (J)
			(V) edge[-,shorten >=0pt] node[above] {
					\translab{\evt{\enterof{\enterQ}}{\athread}}{\athread=\threadvar}
				} ([xshift=-.75mm]U)
			(J) edge[-,shorten >=0pt] ([yshift=-1mm]V.south)
			(H.north west) edge[-,shorten >=0pt] (V)
			([xshift=-.75mm]U) edge ([xshift=-.75mm]F.north)
			([xshift=.75mm,yshift=-1mm]U) edge[-,shorten >=0pt] ([xshift=0mm,yshift=-1mm]V)
			([xshift=.75mm,yshift=-1mm]U) edge ([xshift=.75mm]F.north)
			;

		\node [xshift=7.5cm,yshift=4.8cm] {
			\(\begin{aligned}
				\color{colorF}{F} ~&:=~ \translab{\freeof{\anadr}}{\anadr=\adrvar}
				&\quad
				\color{colorR}{R} ~&:=~ \translab{\evt{\enterof{\retire}}{\athread,\anadr}}{\anadr=\adrvar}
			\end{aligned}\)
		};

		\begin{scope}[xshift=6.9cm, yshift=4.0cm]
			\node [minimum width=1.4cm, minimum height=.7cm,ellipse,fill=colorA,draw=none,xshift=0.0cm] (TA) {$\gactive$};
			\node [minimum width=1.4cm, minimum height=.7cm,ellipse,fill=colorE,draw=none,xshift=1.7cm] (TE) {$\gilluebr$};
			\node [minimum width=1.4cm, minimum height=.7cm,ellipse,fill=colorS,draw=none,xshift=3.4cm] (TS) {$\gsafeaccess$};
		\end{scope}

		\begin{scope}[-,on background layer]
			\newcommand{\spaceA}{0.7cm}
			\newcommand{\spaceE}{1.1cm}
			\newcommand{\spaceS}{0.9cm}
			\draw[fill=colorE,draw=colorE,line width=.5mm,fill opacity=.4] ([xshift=-\spaceE,yshift=-\spaceE]E) rectangle ([xshift=+\spaceE,yshift=+\spaceE]H);
			\draw[fill=colorS,draw=colorS,line width=.5mm,fill opacity=.4] ([xshift=-\spaceS,yshift=-\spaceS]E.center)
			                -- ([xshift=-\spaceS,yshift=+\spaceS]E.center)
			                -- ([xshift=-\spaceS,yshift=+\spaceS]H.center)
			                -- ([xshift=+\spaceS,yshift=+\spaceS]H.center)
			                -- ([xshift=+\spaceS,yshift=-\spaceS]D.center)
			                -- cycle;
			\draw[fill=colorA,draw=colorA,line width=.5mm,fill opacity=.4] ([xshift=-\spaceA,yshift=-\spaceA]A) rectangle ([xshift=+\spaceA,yshift=+\spaceA]D);

			\draw [-,color=colorA,line width=.7mm,shorten <= -2pt] ([yshift=-0.5mm]TA.south) -- ([yshift=-2.75cm]TA.south);
			\draw [-,color=colorE,line width=.7mm,shorten <= -2pt] ([yshift=-0.5mm]TE.south) -- ([yshift=-2.75mm]TE.south);
			\draw [-,color=colorS,line width=.7mm,shorten <= -2pt] ([yshift=-0.5mm]TS.south) -- ([yshift=-4.95mm]TS.south);
		\end{scope}
	\end{tikzpicture}%
	\caption{%
		Cross-product SMR automaton for $\baseobs\times\ebrobs$ and EBR-specific types.
	}
	\label{fig:ebr-types}
	\end{tcolorbox}
\end{figure}

In the following, we illustrate the type transformer relation, the use of angels, the typing of programs, and explain how to find suitable annotations for the type inference to go through.

\subsection{Type Transformer Relation}
\label{sec:example:type-transformer}

We illustrate the computation of the type transformer relation for $\exitof{\leaveQ}$ and the inference~of~$\gsafeaccess$.

First, we establish the type transformer relation \( \emptyset,\apavar,\exitof{\leaveQ} \,\leadsto\, \gilluebr \).
This boils down to checking \( \lpostof{\apavar}{\exitof{\leaveQ}}{\locsof{\emptyset}} \,\subseteq\, \locsof{\gilluebr} \) because the remaining properties of the type transformer relation are trivially satisfied (we do not add any of $\{\gactive,\glocal,\gsafeaccess\}$).
The empty type corresponds to no knowledge about previously executed SMR commands, which means $\locsof{\emptyset}=L$ with $L$ the set of all locations of $\baseobs\times\ebrobs$.
We compute the post-image of $L$ wrt. $\apavar$ and $\exitof{\leaveQ}$ in the SMR automaton from \Cref{fig:ebr-types}. 
To this end, we consider all transitions labeled with $\exitof{\leaveQ(\athread)}$. The pointer or angel $\apavar$ does not play a role.
We derive the desired inclusion as follows: \[
	\lpostof{\apavar}{\exitof{\leaveQ}}{\locsof{\emptyset}}
	\,=\,
	\lpostof{\apavar}{\exitof{\leaveQ}}{L}
	\,=\,
	L \setminus \{ (\ref*{obs:base:init},\ref*{obs:ebr:init}), (\ref*{obs:base:retired},\ref*{obs:ebr:init}) \}
	\,=\,
	\locsof{\gilluebr}
	\ .
\]

Second, we show how to infer $\gsafeaccess$.
From \Cref{fig:ebr-types} we know that $\gilluebr$ alone does not yield $\gsafeaccess$ because of location $(\ref*{obs:base:retired},\ref*{obs:ebr:protected})$; we also need $\gactive$.
We establish \( \gilluebr\wedge\gactive \,\leadsto\, \gilluebr\wedge\gactive\wedge\gsafeaccess \).
Since $\gilluebr\wedge\gactive$ is valid and we do not add $\glocal$, the key task is to establish \( \locsof{\gilluebr\wedge\gactive} \;\subseteq\; \locsof{\gilluebr\wedge\gactive\wedge\gsafeaccess} \).
As $\locsof{\gilluebr\wedge\gactive}\subseteq\locsof{\gilluebr\wedge\gactive}$ trivially holds, it suffices to show $\locsof{\gilluebr\wedge\gactive}\subseteq\locsof{\gsafeaccess}$: \[
	\locsof{\gilluebr\wedge\gactive}
	\,=\,
	\locsof{\gilluebr}\cap\locsof{\gactive}
	\,=\,
	\{ (\ref*{obs:base:init},\ref*{obs:ebr:protected}), \ref*{obs:base-ebr:final} \}
	\,\subseteq\,
	\{ (\ref*{obs:base:init},\ref*{obs:ebr:protected}), (\ref*{obs:base:retired},\ref*{obs:ebr:retired}), \ref*{obs:base-ebr:final} \}
	\,=\,
	\locsof{\gsafeaccess}
	\ .
\]

\subsection{Angels}
\label{sec:example:angels}

\begin{wrapfigure}{r}{5.7cm}
	\vspace{-1.75mm}
	\begin{tcolorbox}
\begin{lstlisting}[style=condensedsimple]
 @inv angel r;                $\label[line]{fig:msebr:angel}$
 beginAtomic                  $\label[line]{fig:msebr:beginatomic}$
   enter leaveQ();            $\label[line]{fig:msebr:leave}$
   exit leaveQ;               $\label[line]{fig:msebr:exit}$
   @inv active(r);            $\label[line]{fig:msebr:active}$
 endAtomic                    $\label[line]{fig:msebr:endatomic}$
 // ...
 Node* head = Head;           $\label[line]{fig:msebr:readhead}$
 Node* tail = Tail;           $\label[line]{fig:msebr:readtail}$
 @inv head in r;              $\label[line]{fig:msebr:headcovered}$
 Node* next = head->next;     $\label[line]{fig:msebr:readnext}$
 // ...
 @inv next in r;              $\label[line]{fig:msebr:nextcovered}$
 data_t output = next->data;  $\label[line]{fig:msebr:deref}$
 // ...
 enter exitQ();
 exit exitQ;
 // ..
\end{lstlisting}%
	\caption{%
		Excerpt of \code{dequeue} using angel~$\aghostvar$.
	}
	\label{fig:msebr}
	\end{tcolorbox}
\end{wrapfigure}
To illustrate the use of angels, consider the excerpt of the \code{dequeue} operation depicted in \Cref{fig:msebr}. 
Note that calls to SMR functions lead to two consecutive commands. 
The atomic block ensures the commands are executed without interruption by other threads.
To infer it, we rely on standard moverness arguments \cite{Lipton75}: command $\enterof{\leaveQ()}$ is a right-mover because it does not affect the memory nor the observer $\baseobs\times\ebrobs$.
The call to $\leaveQ$ guarantees that no currently active address is reclaimed until $\enterQ$ is called. 
It thus protects an unbounded number of addresses
before a thread acquires a pointer to them. 
Later, when a thread acquired a pointer to such an address in order to access it, the address may no longer be active and thus the type system may not be able to infer $\gsafeaccess$ (cf. \Cref{sec:example:type-transformer}). 
To overcome this problem, we use an angel~$\aghostvar$.
Given its angelic semantics, $\aghostvar$ will capture all addresses that are protected by the \code{leaveQ} call, \Cref{fig:msebr:angel,fig:msebr:beginatomic,fig:msebr:leave,fig:msebr:exit,fig:msebr:active}.
Later, upon accessing/dereferencing a pointer $\apvar$, we make sure that $\aghostvar$ captures the address pointed to by~$\apvar$, \Cref{fig:msebr:headcovered,fig:msebr:nextcovered}.

\subsection{Typing}
\label{sec:example:type-inference}

We give a typing for the code from \Cref{fig:msebr} in \Cref{fig:msebr-typing}.
We start in \Cref{typing:msebr:t:start} with type $\emptyset$ for all pointers and the angel $\aghostvar$.
The allocation of $\aghostvar$ in \Cref{typing:msebr:c:angel} has no effect on the type assignment.
The same holds when entering an atomic block, \Cref{typing:msebr:c:beginatomic}.
\Cref{typing:msebr:c:leave} invokes \code{leaveQ}.
Again, the types are not affected because the SMR automaton has no transitions labeled with $\enterof{\leaveQ}$.
Next, the invocation returns, \Cref{typing:msebr:c:exit}.
Following the discussion from \Cref{sec:example:type-transformer}, we obtain $\gilluebr$ for~$\aghostvar$, \Cref{typing:msebr:t:exit}.
It is worth pointing out that $\aghostvar$ is treated like an ordinary pointer when it comes to the type transformer relation.

%
\begin{figure}
  \begin{tcolorbox}
	\center%
	\begin{minipage}[t]{.47\textwidth}%
\begin{lstlisting}[style=typing]
  `{ Head,head,next,r:$\emptyset$ }`                                   $\label[line]{typing:msebr:t:start}$
    @inv angel r;                                                      $\label[line]{typing:msebr:c:angel}$
  `{ Head,head,next,r:$\emptyset$ }`                                   $\label[line]{typing:msebr:t:angel}$
    beginAtomic                                                        $\label[line]{typing:msebr:c:beginatomic}$
  `{ Head,head,next,r:$\emptyset$ }`                                   $\label[line]{typing:msebr:t:beginatomic}$
    enter leaveQ();                                                    $\label[line]{typing:msebr:c:leave}$
  `{ Head,head,next,r:$\emptyset$ }`                                   $\label[line]{typing:msebr:t:leave}$
    exit leaveQ;                                                       $\label[line]{typing:msebr:c:exit}$
  `{ Head,head,next:$\emptyset$, r:$\gilluebr$ }`                      $\label[line]{typing:msebr:t:exit}$
    @inv active(r);                                                    $\label[line]{typing:msebr:c:activer}$
  `{ Head,head,next:$\emptyset$, r:$\gilluinv\wedge\gactive$ }`        $\label[line]{typing:msebr:t:activer}$
  `{ Head,head,next:$\emptyset$, r:$\gilluinv\wedge\gactive\wedge\gsafeaccess$ }`    $\label[line]{typing:msebr:t:infer}$
    endAtomic                                                          $\label[line]{typing:msebr:c:endatomic}$
  `{ Head,head,next:$\emptyset$, r:$\gilluinv\wedge\gsafeaccess$ }`    $\label[line]{typing:msebr:t:endatomic}$
    // ...
\end{lstlisting}%
	\end{minipage}%
	\hfill%
	\begin{minipage}[t]{.465\textwidth}%
\begin{lstlisting}[style=typing]
    // ...
  `{ Head,head,next:$\emptyset$, r:$\gilluinv\wedge\gsafeaccess$ }`    $\label[line]{typing:msebr:t:endatomicrepeat}$
    Node* head = Head;                                                 $\label[line]{typing:msebr:c:readhead}$
  `{ Head,head,next:$\emptyset$, r:$\gilluinv\wedge\gsafeaccess$ }`    $\label[line]{typing:msebr:t:readheadfirst}$
    // ...
  `{ Head,head,next:$\emptyset$, r:$\gilluinv\wedge\gsafeaccess$ }`    $\label[line]{typing:msebr:t:readhead}$
    @inv head in r;                                                    $\label[line]{typing:msebr:c:headinr}$
  `{ Head,next:$\emptyset$, head,r:$\gilluinv\wedge\gsafeaccess$ }`    $\label[line]{typing:msebr:t:headinr}$
    Node* next = head->next;                                           $\label[line]{typing:msebr:c:readnext}$
    // ...
  `{ Head,next:$\emptyset$, head,r:$\gilluinv\wedge\gsafeaccess$ }`    $\label[line]{typing:msebr:t:readnext}$
    @inv next in r;                                                    $\label[line]{typing:msebr:c:nextinr}$
  `{ Head:$\emptyset$, next,head,r:$\gilluinv\wedge\gsafeaccess$ }`    $\label[line]{typing:msebr:t:nextinr}$
    data_t output = next->data;                                        $\label[line]{typing:msebr:c:deref}$
    // ...
\end{lstlisting}%
	\end{minipage}%
	\caption{Typing for EBR using angels.}
	\label{fig:msebr-typing}
  \end{tcolorbox}
\end{figure}

To capture in the type system the set of addresses that can be safely accessed in the subsequent code, we want to lift $\gilluebr$ of $\aghostvar$ to $\gsafeaccess$.
We annotate $\aghostvar$ to hold a set of active addresses, \Cref{typing:msebr:c:activer}.
This yields type $\gilluebr\wedge\gactive$ for $\aghostvar$, \Cref{typing:msebr:t:activer}.
As explained above, we can now lift this type to $\gilluebr\wedge\gactive\wedge\gsafeaccess$, \Cref{typing:msebr:t:infer}.
Recall that the allocation of $\aghostvar$ in \Cref{typing:msebr:c:angel} is angelic.
That is, the addresses held by $\aghostvar$ will indeed be chosen to be active.

In the subsequent code, we already added annotations (cf. \Cref{sec:example:angels}) ensuring that accessed/dereferenced pointers are captured by the angel $\aghostvar$.
For instance, \Cref{typing:msebr:c:headinr} requires the address of \code{head} to be captured by $\aghostvar$.
That this is the case indeed is established when the annotations are discharged.
For the typing, we can copy $\gilluebr\wedge\gsafeaccess$ from $\aghostvar$ over to \code{head}.
As a consequence, the dereference of \code{head} in \Cref{typing:msebr:c:readnext} is safe.
Similarly, we require \code{next} to be captured by $\aghostvar$ in \Cref{typing:msebr:c:nextinr} such that the dereference in \Cref{typing:msebr:c:deref} is safe.

\subsection{Annotations}
\label{sec:example:annotations}

We explain our algorithm to automatically add to the program in \Cref{fig:michaelscott} the annotations in \Cref{fig:msebr} in order to arrive at the typing in \Cref{fig:msebr-typing}.
We focus on the dereference of \code{head} in \Cref{fig:msebr:readnext}.
Without annotations, the type inference will fail because it cannot conclude that \code{head} is guaranteed to be valid.
To fix this, we implemented a sequence of tactics that we invoke one after the other.
If none of them fixes the issue, we give up the type inference and report the failure to the user.

The first tactic simply adds an $\invariantof{\activeof{\mcode{head}}}$ annotation to \Cref{fig:msebr:readnext}.
This makes \code{head} valid and the type inference go through for \Cref{fig:msebr:readnext}.
However, we should only add the annotation if it actually holds.
To check this, we employ the technique from \Cref{sec:invariant-check}.
In this particular case, we will find that the annotation does not hold; so we try to fix the problem with another tactic.

The second tactic adds an angel $\aghostvar$ to the (syntactically) most recent $\leaveQ$ call.
We use a template to transform the sequence $\enterof{\leaveQ()};\exitof{\leaveQ};$ to the code from Lines~\ref{fig:msebr:angel}-\ref{fig:msebr:endatomic}.
(A subsequent use of this tactic will skip this step and reuse the existing angel.)
Then, we fix \Cref{fig:msebr:readnext} by adding the annotation $\ghostof{\containsof{\mcode{head}}{\mathtt{r}}}$ before it, as shown in \Cref{fig:msebr:headcovered}.
This makes \code{head} valid.
Whether the annotation holds is again checked with the technique from \Cref{sec:invariant-check}.

It is worth pointing out that the second tactic is EBR-specific.
From our experience, every SMR algorithm/automaton comes with a small set of tactics that significantly help finding the right annotations---EBR requires the above tactic and HP requires two specific tactics.
We do not believe that there is a silver bullet of tactics since SMR algorithms may vary greatly, as seen in the cases of EBR and HP.
Theoretically speaking, one could find the annotations by an exhaustive search (finitely many angels will suffice), but this will not scale.

\subsection{Hazard Pointers}

Our approach applies to lock-free data structures with hazard pointers just as well as in the case of EBR.
The main difference is that HP typically does not require angels because pointers are protected after they are acquired.
However, the size of the SMR automaton for HP grows in the number of hazard pointers.
For two hazard pointers it consists of $17$ locations.
Consider \Cref{appendix:hpexample}~for~details.


\section{Type Inference}
\label{sec:type_checking}

We show that type inference is surprisingly efficient, namely quadratic time.
\begin{theorem}
	Given a program $\astmt$, the type inference $\typechecks{\astmt}$ is computable in time $\mathcal{O}(\cardnodistof{\astmt}^2)$.
\end{theorem}
As common in type systems~\cite{Pierce}, our algorithm for type inference is constraint-based.
We associate with the program $\astmt$ a constraint system $\eqsysof{\envinit, \astmt, X}$.
The variables $X$ are interpreted over the set of type environments enriched with a value $\fail$ for a failed type inference.
The correspondence between solving the constraint system and type inference will be the following.
An environment $\env$ can be assigned to $X$ in order to solve the constraint system if and only if $\typestmt{\envinit}{\astmt}{\env}$.
As a consequence, a non-trivial solution to $X$ will show $\typechecks{\astmt}$.

Our type inference algorithm will be a fixed-point computation.
The canonical choice for a domain over which to compute would be the set of types ordered by $\leadsto$.
The problem is that types of the form $\gcustom{\alocset}$ and $\gcustom{\alocset}\wedge \gcustom{\alocsetp}$ with $\alocset\subseteq \alocsetp$ are comparable, $\gcustom{\alocset} \leadsto \gcustom{\alocset}\wedge \gcustom{\alocsetp}$ and $\gcustom{\alocset}\wedge \gcustom{\alocsetp}\leadsto \gcustom{\alocset}$.
This is not merely a theoretical issue of the domain being a quasi instead of a partial order.
It means we compute over too large a domain, namely a powerset lattice where we should have used a lattice of antichains~\cite{WDHR06}.
We factorize the set of all types along such equivalences $\leadsto\cap \leadsto^{-1}$.
The resulting $\actypes:=(\factorize{\alltypes}{\leadsto\cap \leadsto^{-1}}, \leadsto)$ is a complete lattice \cite{birkhoff1948lattice}.

Type environments can be understood as total functions into this antichain lattice.
We enrich the set of functions by a value $\fail$ to indicate a failed type inference. 
The result is the complete lattice of enriched type environments \[ \allenvsbot\; :=\; (\actypes^{\allvars} \cup \set{\fail}, \sqsubseteq) \ . \]
Between environments, we define $\env\sqsubseteq \envp$ to hold if for all  $\apavar\in\allvars$ we have $\envof{\apavar}\leadsto \envp(\apavar)$.
This lifts $\leadsto$ to the function domain.
Value~$\fail$ is defined to be the largest element.

\begin{figure}
  \begin{tcolorbox}
	\setalignspacingtofiguremode
	\begin{align*}
		\eqsysof{X, \acom, Y}:\hspace{-1.5cm}&&  \spof{X, \acom}\sqsubseteq Y\\
		\eqsysof{X, \astmt_1;\astmt_2, Y}:\hspace{-1.5cm}&& \eqsysof{X, \astmt_1, Z}&\wedge \eqsysof{Z, \astmt_2, Y},\quad \text{$Z$ fresh}\\
		\eqsysof{X, \astmt_1\choice\astmt_2, Y} :\hspace{-1.5cm}&&\eqsysof{X, \astmt_1, Y}&\wedge \eqsysof{X, \astmt_2, Y}\\
		\eqsysof{X, \astmt^*, Y}:\hspace{-1.5cm}&& \eqsysof{Y, \astmt, Y}&\wedge  X\sqsubseteq Y
	\end{align*}
	\caption{Constraint system $\eqsysof{X, \astmt, Y}$.}
	\label{Figure:ConstraintSystem}
  \end{tcolorbox}
\end{figure}

The constraint system $\eqsysof{\envinit, \astmt, X}$ is defined in Figure~\ref{Figure:ConstraintSystem}.
We proceed by induction over the structure of statements and maintain triples $(X, \astmt, Y)$.
The idea is that statement $\astmt$ will turn the enriched type environment stored in variable $X$ into an environment upper bounded by~$Y$. 
Consider the case of basic commands.
We will define $\spof{X, \acom}$ to be the strongest enriched type environment resulting from the environment in $X$ when applying command $\acom$.
The constraint $\spof{X, \acom}\sqsubseteq Y$ requires $Y$ to be an upper bound.
Note that $Y$ still contains safe type information.
For a sequential composition, we introduce a fresh variable $Z$ for the enriched type environment determined by $\astmt_1$ from $X$.
We then require $\astmt_2$ to transform this environment into $Y$.
For a choice, $Y$ should upper bound the effects of both $\astmt_1$ and $\astmt_2$ on $X$.
This guarantees that the type information is valid independent of which branch is chosen.
For iterations, we have to make sure $Y$ is an upper bound for the effect of arbitrarily many applications of $\astmt$ to $X$.
This means the environment in~$Y$ is at least $X$ because the iteration may be skipped.
Moreover, if we apply $\astmt$ to $Y$ then we should again obtain at most the environment in $Y$.

It remains to define $\spof{X, \acom}$, the strongest enriched type environment resulting from $X$ under command $\acom$.
We refer to the typing rules in~\Cref{fig:type-rules} and extract $\premiseof{\acom}$ and $\updateof{\acom}$.
The former is a predicate on environments capturing the premise of the rule associate with command $\acom$.
To give an example, for Rule~\ref{rule:assign2} the predicate $\premiseof{\acom}(\env)$ is $\isvalidof{\atype}$ with $\atype=\envof{\apvarp}$.
The latter is a function on environments.
It captures the update of the given environment as defined in the consequence of the corresponding rule.
For~\ref{rule:assign2}, the update $\updateof{\acom}(\env)$ is \mbox{$\env[\apvar\mapsto\emptyset]$}.
The strongest enriched environment preserves the information that a type inference has failed,  \mbox{$\spof{\fail,\acom}:=\fail$}, for all commands.
For a given environment, we set \[ \spof{\env, \acom}\; :=\; \premiseof{\acom}(\env) ~\,{?}~\, \updateof{\acom}(\env) ~\,{:}~\, \fail \ . \]
We evaluate the premise of the rule.
If it does not hold, the type inference will fail and return~$\fail$.
Otherwise, we determine the update of the current type environment, $\updateof{\acom}(\env)$.
We rely on the fact that $\spof{\cdot, \acom}$ is monotonic and hence (as the domains are finite) continuous.

We apply a Kleene iteration to obtain the least solution to the constraint system $\eqsysof{\envinit, \astmt, X}$.
The least solution is a function $\lsol$ that assigns to each variable in the system an enriched type environment.
We focus on variable $X$ that captures the effect of the overall program on the initial type environment.
Then $\lsolof{X}$ is the strongest type environment that can be obtained by a successful type inference.
This is the key correspondence.
\begin{proposition}[Principle Types]
	\label{Proposition:TypeCheck}
	Consider $\eqsysof{\envinit, \astmt, X}$. 
	Then  $\lsolof{X} = \bigsqcap_{\typejudge{\,\envinit}{\astmt}{\env}}\env$.
	Hence, $\lsolof{X}\neq \fail$ if and only if $\typechecks{\astmt}$.
\end{proposition}

It remains to check the complexity of the Kleene iteration.
In the lattice of enriched type environments, chains have length at most ${\cardof{\vars\mkern+1mu}}\,{\cdot}\,{\cardof{\set{\gactive, \glocal, \gsafeaccess, \gcustom{1}, \ldots, \gcustom{n}}}}\,{+}\,1$.
This is linear in the size of the program as the guarantees only depend on the SMR algorithm, which is not part of the input.
With one variable for each program point, also the number of variables in the constraint system is linear in the size of the program.
It remains to compute $\spof{\cdot, \acom}$ for the Kleene approximants.
This can be done in constant time.
The premise and the update of a rule only modify a constant number of variables.
Moreover,  we can look-up the effect of commands on a type in constant time.
Combined, we obtain the overall quadratic complexity.


\section{Invariant Checking}
\label{sec:invariant-check}

The type system from \Cref{sec:type_system} relies on invariant annotations in the program under scrutiny in order to incorporate runtime behavior that is typically not available to a type system.
For the soundness of our approach, we require those annotations to be correct.
Recall from \Cref{sec:type_system} that the annotations need only hold in the garbage collected (GC) semantics.
We now show how to use an off-the-shelf GC verifier to discharge the invariant annotations fully automatically.
In our experiments, we rely on \cavetool \cite{DBLP:conf/cav/Vafeiadis10,DBLP:conf/vmcai/Vafeiadis10,DBLP:conf/vmcai/Vafeiadis09}.

Making the link to tools is non-trivial.
Our programs feature programming constructs that are typically not available in off-the-shelf verifiers.
We present a source-to-source translation that replaces those constructs by standard ones.
The constructs to be replaced are SMR commands, invariants guaranteeing pointers to be active (not retired), and invariants centered around angels.
For the translation, we only rely on ordinary assertions $\assertof{\acond}$ and non-deterministic assignments $\havocof{\apvar}$ to pointers.
Both are usually available in verification tools.

The correspondence between the original program $\aprog$ and its translation $\instrumentationof{\aprog}$ is documented in \Cref{Theorem:Instrumentation} and as required.
Predicate $\issafeof{\cdot}$ evaluates to true iff the assertions hold, i.e., verification is successful.
Recall that $\nosem[\aprog]$ is the GC semantics where addresses are neither freed nor reclaimed.
Note that this semantics is the weakest a tool can assume.
Our instrumentation also works if the GC tool collects and subsequently reuses garbage nodes.

\begin{theorem}[Soundness and Completeness]
	\label{Theorem:Instrumentation}
	We have $\invholdsof{\nosem[\aprog]}$ iff $\,\issafeof{\nosem[\instrumentationof{\aprog}]}$. The source-to-source translation is linear in size.
\end{theorem}

\begin{figure}
  \begin{tcolorbox}
	\setalignspacingtofiguremode
	\center
	\setlength\extrarowheight{2.5pt} 
	\newcommand{\isdefinedas}{\;:=\;}
	\newcommand{\codenot}{\mymathtt{!}}
	\newcommand{\codeor}{\:\vee\:}
	\newcommand{\codeand}{\:\wedge\:}
	\newcommand{\codesemi}{;\mkern+2mu}
	\newcolumntype{L}[1]{>{\raggedright\let\newline\\\arraybackslash\hspace{0pt}}p{#1}}
	\newcolumntype{C}[1]{>{\centering\let\newline\\\arraybackslash\hspace{0pt}}p{#1}}
	\newcolumntype{R}[1]{>{\raggedleft\let\newline\\\arraybackslash\hspace{0pt}}p{#1}}
	\begin{tabular}{R{3.2cm}@{\(\isdefinedas\)}L{3.5cm}@{\hspace{.5cm}}R{3.2cm}@{\(\isdefinedas\)}L{3cm}}
		\( \instrumentationof{\astmt^*} \)
			&\( \instrumentationof{\astmt}^* \)
		&
		\( \instrumentationof{\enterof{\afuncof{\vecof{\apvar},\vecof{\advar}}}} \)
			&\( \cskip \)
		\\
		\( \instrumentationof{\astmt_1\choice \astmt_2} \)
			&\( \instrumentationof{\astmt_1}\choice \instrumentationof{\astmt_2} \)
		&
		\( \instrumentationof{\exitof{\afunc}} \)
			&\( \cskip \)
		\\
		\( \instrumentationof{\astmt_1\codesemi \astmt_2} \)
			&\( \instrumentationof{\astmt_1}\codesemi \instrumentationof{\astmt_2} \)
		&
		\( \instrumentationof{\invariantof{\apvar=\apvarp}} \)
			&\( \assertof{\apvar=\apvarp\mkern+2mu} \)
		\\
		\( \instrumentationof{\acom} \)
			&\( \acom \)
		\\
	\end{tabular}\vspace{0.05cm}
	\begin{tabular}{R{3.2cm}@{\(\isdefinedas\)}L{10.3cm}}
		%
		\( \instrumentationof{\enterof{\retireof{\apvarp}}} \)
			&\( \cskip\choice\big(
				\mymathtt{retire\_ptr}:=\apvarp\codesemi
				\mymathtt{retire\_flag}:=\btrue
			\big) \)
		\\
		\( \instrumentationof{\invariantof{\activeof{\apvar}}} \)
			&\( \assertof{\codenot\mymathtt{retire\_flag}\codeor\mymathtt{retire\_ptr}\neq\apvar} \)
		\\[0.1cm]
		%
		\( \instrumentationof{\ghostof{\chooseof{\aghostvar}}} \)
			&\( \havocof{\aghostvar}\codesemi
			\mymathtt{included\_}\aghostvar:=\bfalse\codesemi
			\mymathtt{failed\_}\aghostvar:=\bfalse \)
		\\
		\( \instrumentationof{\ghostof{\containsof{\apvarp}{\aghostvar}}} \)
			&\( \cskip\choice\big(
				\assumeof{\apvarp=\aghostvar}\codesemi
				\assertof{\codenot\mymathtt{failed\_}\aghostvar}\codesemi
				\mymathtt{included\_}\aghostvar:=\btrue
			\big) \)
		\\
		\( \instrumentationof{\ghostof{\activeof{\aghostvar}}} \)
			&\( \cskip\choice\big(
				\assumeof{\mymathtt{retire\_flag}\codeand\mymathtt{retire\_ptr}=\aghostvar}\codesemi
				\)\newline\(\phantom{\cskip\choice\big(}
				\assertof{\codenot\mymathtt{included\_}\aghostvar}\codesemi
				\mymathtt{failed\_}\aghostvar:=\btrue
			\big) \)
	\end{tabular}
	\vspace{-2mm}
	\caption{%
		Source-to-source translation replacing SMR commands and invariant annotations. 
	}
	\label{Figure:Instrumentation}
  \end{tcolorbox}
\end{figure}

The source-to-source translation is defined in \Cref{Figure:Instrumentation}.
It preserves the structure of the program and does not modify ordinary commands.
SMR calls and returns will be taken care of by the type system.
They are ignored, except for retire.
Invariants guaranteeing pointer equality yield~assertions.

The purpose of invariants $\invariantof{\activeof{\apvar}}$ is to guarantee that the address held by the pointer has not been retired since its last allocation.
The idea of our translation is to guess the moment of failure, the retire function after which such an invariant will be checked.
We instrument the program by an additional pointer $\mymathtt{retire\_ptr}$ and a Boolean variable $\mymathtt{retire\_flag}$.
Both are shared.
A retire translates into a non-deterministic choice between skipping the command or being the retire after which an invariant will fail.
In the latter case, the address is stored in $\mymathtt{retire\_ptr}$ and $\mymathtt{retire\_flag}$ is raised.
Note that the instrumentation is tailored towards garbage collection.
As long as $\mymathtt{retire\_ptr}$ points to the address, it will not be reallocated.
Therefore, we do not run the risk of the address becoming active ever again.
The invariant $\invariantof{\activeof{\apvar}}$ now translates into an assertion that checks the address of $\apvar$ for being the retired one and the flag for being raised.
A thing to note is that the instrumentation of the retire function is not atomic.
Hence, there may be an interleaving where a pointer has been stored in $\mymathtt{retire\_ptr}$ but the flag has not yet been raised.
The assertion would consider this interleaving safe.
However, if there is such an interleaving, there is also one where the assertion fails.
Hence, atomicity is not needed.

For invariants involving angels, the idea of the instrumentation is the same as for pointers, guessing the moment of failure.
What makes the task more difficult is the angelic semantics.
We cannot just guess a value for the angel and show that it makes an invariant fail.
Instead, we have to show that, no matter how the value is chosen, it inevitably leads to an invariant failure.
This resembles the idea of having a strategy to win against an opponent in a turn-based game,  a common phenomenon when quantifier alternation is involved~\cite{GTW2500}.
Another source of difficulty is the fact that angels are second-order variables storing sets.
We tackle the problem by guessing an element in the set for which verification fails.

The instrumentation proceeds as follows.
We consider angels $\aghostvar$ to be ordinary pointers.
For each angel, we add two Boolean variables $\mymathtt{included\_}\aghostvar$ and  $\mymathtt{failed\_}\aghostvar$ that are local to the thread.
When we allocate an angel using $\ghostof{\chooseof{\aghostvar}}$, we guess the address that (i) will inevitably belong to the set of addresses held by the angel and (ii) for which an active invariant will fail.
To document that we are sure of (i), we raise flag $\mymathtt{included\_}\aghostvar$.
For (ii), we use $\mymathtt{failed\_}\aghostvar$.
If we are sure of both facts, we let verification fail.
Note that we can derive the facts in arbitrary order.

An invariant $\ghostof{\containsof{\apvarp}{\aghostvar}}$ forces the angel to contain the address of $\apvarp$.
This may establish (i).
The reason it does not establish (i) for sure is that the angel denotes a set of addresses,  and the address of $\apvarp$ could be different from the one for which an active invariant fails.
Hence, we non-deterministically choose between skipping the invariant or comparing $\apvarp$ to $\aghostvar$.
If the comparison succeeds, we raise  $\mymathtt{included\_}\aghostvar$.
Moreover, we check (ii).
If the address has been retired, we report a bug.

Invariant $\ghostof{\activeof{\aghostvar}}$ forces all addresses held by the angel to be active.
In the instrumented program, $\aghostvar$ is a pointer that we compare to $\mymathtt{retire\_ptr}$ introduced above.
If the address has been retired, we are sure about (ii) and document this by raising $\mymathtt{failed\_}\aghostvar$.
If we already know (i), the address inevitably belongs to the set held by the angel, verification fails.


\section{evaluation}
\label{sec:evaluation}

We implemented our approach in a \cpp tool called \thetool.\footnote{%
	Available at: \url{https://wolff09.github.io/seal/}
}.
As stated before, we use the state-of-the-art tool \textsc{cave} \cite{DBLP:conf/cav/Vafeiadis10,DBLP:conf/vmcai/Vafeiadis10,DBLP:conf/vmcai/Vafeiadis09} as a back-end verifier for discharging annotations and checking linearizability.
For the type inference, our tool computes the most precise guarantees $\gcustom{L}$ on-the-fly; there is no need for the user to manually specify them.
To substantiate the claim of usefulness of our approach, we evaluated \thetool on examples from the literature.
We considered the following data structures: Treiber's stack \cite{opac-b1015261,DBLP:conf/podc/Michael02}, Michael\&Scott's lock-free queue \cite{DBLP:conf/podc/MichaelS96,DBLP:conf/podc/Michael02}, the DGLM queue \cite{DBLP:conf/forte/DohertyGLM04}, the Vechev\&Yahav CAS set \cite{DBLP:conf/pldi/VechevY08}, the Vechev\&Yahav DCAS set \cite{DBLP:conf/pldi/VechevY08}, the ORVYY set \cite{DBLP:conf/podc/OHearnRVYY10}, and Michael's set \cite{DBLP:conf/spaa/Michael02}.
Our benchmarks include a version of each data structure for hazard pointers (HP) \cite{DBLP:conf/podc/Michael02} and epoch-based reclamation~(EBR) \cite{DBLP:phd/ethos/Fraser04}.
We adapted the GC implementations of the Vechev\&Yahav DCAS set, the Vechev\&Yahav CAS set, and the ORVYY set given in the literature to use HP/EBR.


\begin{table}%
  \begin{tcolorbox}
	\caption{Experimental results for verifying singly-linked data structures using safe memory reclamation. The experiments were conducted on an Intel i5-8600K@3.6GHz with 16GB of RAM.}%
	\label{table:benchmarks}%
	\begin{minipage}{\columnwidth}
		\center%
		\newcommand{\cellSpacer}{\hspace{3mm}}
		\newcommand{\widthLeft}{1cm}
		\newcommand{\widthRight}{.5cm}
		\newcommand{\cellDunno}[1]{\makebox[\widthLeft][r]{\(#1\)}\cellSpacer\makebox[\widthRight][l]{\rawsymbolTO}}
		\newcommand{\cellTO}{\makebox[\widthLeft][r]{---}\cellSpacer\makebox[\widthRight][l]{\rawsymbolTO}}
		\newcommand{\cellYes}[1]{\makebox[\widthLeft][r]{\(#1\)}\cellSpacer\makebox[\widthRight][l]{\rawsymbolYes}}
		\newcommand{\cellNo}[2][]{\makebox[\widthLeft][r]{\(#2\)}\cellSpacer\makebox[\widthRight][l]{\rawsymbolNo#1}}
		\newcommand{\seprule}{\midrule[.2pt]}
		\begin{tabularx}{\textwidth}{lXccc}%
			\toprule
			SMR
				&Program
				& Type Inference
				& Annotations
				& Linearizability
				\\
			\midrule
			\multirow{8}{*}{HP}
			&Treiber's stack
				& \cellYes{0.7s}
				& \cellYes{12s}
				& \cellYes{1s}
				\\
			&Opt. Treiber's stack
				& \cellYes{0.5s}
				& \cellYes{11s}
				& \cellYes{1s}
				\\
			&Michael\&Scott's queue
				& \cellYes{0.6s}
				& \cellYes{12s}
				& \cellYes{4s}
				\\
			&DGLM queue
				& \cellYes{0.6s}
				& \cellNo[\footnote{False-positive due to imprecision in the back-end verifier.\label{footnote:imprecision}}]{1s}
				& \cellYes{5s}
				\\
			&Vechev\&Yahav DCAS set
				& \cellYes{1.2s}
				& \cellYes{13s}
				& \cellYes{98s}
				\\
			&Vechev\&Yahav CAS set
				& \cellYes{1.2s}
				& \cellYes{3.5h} 
				& \cellYes{42m} 
				\\
			&ORVYY set
				& \cellYes{1.2s}
				& \cellYes{3.2h} 
				& \cellYes{47m} 
				\\
			&Michael's set
				& \cellYes{1.2s}
				& \cellNo[\textsuperscript{\ref{footnote:imprecision}}]{90s}
				& \cellTO
				\\
			\midrule
			\multirow{8}{*}{EBR}
			&Treiber's stack
				& \cellYes{0.6s}
				& \cellYes{10s}
				& \cellYes{1s}
				\\
			&Michael\&Scott's queue
				& \cellYes{0.7s}
				& \cellYes{16s}
				& \cellYes{5s}
				\\
			&DGLM queue
				& \cellYes{0.7s}
				& \cellNo[\textsuperscript{\ref{footnote:imprecision}}]{1s}
				& \cellYes{6s}
				\\
			&Vechev\&Yahav DCAS set
				& \cellYes{0.8s}
				& \cellYes{38s}
				& \cellYes{200s}
				\\
			&Vechev\&Yahav CAS set
				& \cellYes{0.8s}
				& \cellYes{7h}
				& \cellYes{42m}
				\\
			&ORVYY set
				& \cellYes{0.9s}
				& \cellYes{7h}
				& \cellYes{47m}
				\\
			&Michael's set
				& \cellYes{0.2s}
				& \cellNo[\textsuperscript{\ref{footnote:imprecision}}]{22s}
				& \cellTO
				\\
			\bottomrule
		\end{tabularx}
	\end{minipage}
  \end{tcolorbox}
\end{table}

Our findings are listed in \Cref{table:benchmarks}.
The experiments were conducted on an Intel i5-8600K@3.6GHz with 16GB of RAM.
The \namecref{table:benchmarks} includes the time taken
\begin{inparaenum}[(i)]
	\item for the type inference,
	\item for discharging the invariant annotations, and
	\item to check linearizability.
\end{inparaenum}
We mark tasks with~\symbolYes if they were successful, with~\symbolNo if they failed, and with~\symbolTO if they timed out after $12h$ wall time.

Our approach is capable of verifying most of the lock-free data structures we considered.
Comparing the total runtime with our competitors \cite{DBLP:journals/pacmpl/MeyerW19}, the only other approach capable of handling lock-free data structures with general SMR algorithms, we experience a speed-up of over two orders of magnitude on examples like Michael\&Scott's queue.
Besides the speed-up, we are the first to automatically verify lock-free set algorithms that use SMR.

We were not able to discharge the annotations of the DGLM queue and Michael's set.
Imprecision in the thread-modular abstraction of our back-end verifier resulted in false-positives being reported.
Hence, we cannot guarantee the soundness of our analysis in these cases. 
This is no limitation of our approach, it is a shortcoming of the back-end verifier.
\Citet{DBLP:journals/pacmpl/MeyerW19} reported a similar issue that they solved by manually providing hints to improve the precision of their analysis.


The annotation checks for set implementations are interesting.
While the HP version of an implementation is typically more involved than the corresponding version using EBR, the annotation checks for the HP version are more efficient.
The reason for this could be that EBR implementations require angels.
The conjecture suggests that discharging angels is harder for \cavetool than discharging active annotations although our instrumentation uses the same idea for both annotation types.

For the benchmarks from \Cref{table:benchmarks} we preprocessed the implementations by applying mover types~\cite{Lipton75}, a well-known program transformation (cf. \Cref{sec:related_work}).
Intuitively, a command is a mover if it can be reordered with commands of other threads.
This allows for the command to be moved to the next command of the same thread, effectively constructing an atomic block containing the mover and the next command.
What is remarkable in our setting is that SMR commands ($\enter$, $\exit$, $\free$) always move over ordinary memory commands, and vice versa.
(Technically, this requires $\enter$ commands to contain only thread-local variables, a property than be checked/established easily.)
As a result, we can find movers for memory commands using existing techniques.
For SMR commands, movers can be read of the SMR automaton.
Our tool is able to find and apply movers.
Due to space constraints, we omit a thorough discussion of the matter.


\presection
\section{Related Work}
\label{sec:related_work}

\paragraph{Safe Memory Reclamation}
Besides EBR and HP there is another basic SMR technique: reference counting (RC).
RC extends records with an integer field counting the number of pointers to the record.
Safely modifying counters in a lock-free manner, however, requires hazard pointers \cite{DBLP:journals/tocs/HerlihyLMM05} or a mostly unavailable CAS for two arbitrary memory locations \cite{DBLP:conf/podc/DetlefsMMS01}.

Recent efforts in developing SMR algorithms have mostly combined existing SMR techniques.
For example, \emph{DEBRA} \cite{DBLP:conf/podc/Brown15} is an optimized EBR implementation.
\Citet{DBLP:conf/wdag/Harris01} modifies EBR to store epochs inside records.
\emph{Hyaline} \cite{DBLP:conf/podc/0001R19} is used like EBR.
Optimized HP implementations include the work by \citet{DBLP:conf/podc/AghazadehGW13a}, the work by \citet{DBLP:conf/iwmm/DiceHK16}, and \emph{Cadence} \cite{DBLP:conf/spaa/BalmauGHZ16}.
\emph{Dynamic Collect} \cite{DBLP:conf/podc/DragojevicHLM11}, \emph{StackTrack} \cite{DBLP:conf/eurosys/AlistarhEHMS14}, and \emph{ThreadScan} \cite{DBLP:conf/spaa/AlistarhLMS15} are HP-esque implementations exploring the use of operating system and hardware support.
\emph{Drop the Anchor} \cite{DBLP:conf/spaa/BraginskyKP13}, \emph{Optimistic Access} \cite{DBLP:conf/spaa/CohenP15}, \emph{Automatic Optimistic Access} \cite{DBLP:conf/oopsla/CohenP15}, \emph{QSense} \cite{DBLP:conf/spaa/BalmauGHZ16}, \emph{Hazard Eras} \cite{DBLP:conf/spaa/RamalheteC17}, and \emph{Interval-based Reclamation} \cite{DBLP:conf/ppopp/WenICBS18} combine EBR and HP.
\emph{Free Access} \cite{DBLP:journals/pacmpl/Cohen18} automates the application of Automatic Optimistic Access.
While the method promises to be correct by construction, we believe that performance-critical applications choose the SMR technique based on performance rather than ease of use.
The demand for automated verification remains.
\emph{Beware\&Cleanup} \cite{DBLP:conf/ispan/GidenstamPST05} combines HP and RC.
\emph{Isolde} \cite{DBLP:conf/iwmm/YangW17} combines EBR and RC.
We believe our approach can handle other SMR algorithms besides EBR and HP as well.

\paragraph{Memory Safety}
We use our type system to show that a program is free from pointer races, meaning that it is memory safe.
There are a number of related tools that can check pointer programs for memory safety.
For example:
\begin{inparaitem}[]
	\item a combination of \textsc{ccured} \cite{DBLP:conf/popl/NeculaMW02} and \textsc{blast} \cite{DBLP:conf/spin/HenzingerJMS03} due to \citet{DBLP:conf/fase/BeyerHJM05},
	\item \textsc{invader} \cite{DBLP:conf/cav/YangLBCCDO08},
	\item \textsc{xisa} \cite{DBLP:conf/esop/LavironCR10},
	\item \textsc{slayer} \cite{DBLP:conf/cav/BerdineCI11},
	\item \textsc{infer} \cite{DBLP:conf/nfm/CalcagnoD11},
	\item \textsc{forester} \cite{DBLP:conf/cav/HolikLRSV13},
	\item \textsc{predator} \cite{DBLP:conf/sas/DudkaPV13,DBLP:conf/hvc/HolikKPSTV16}, and
	\item \textsc{aprove} \cite{DBLP:journals/jar/StroderGBFFHSA17}.
\end{inparaitem}
These tools can only handle sequential code.
Moreover, unlike our type system, they include memory/shape abstractions to identify unsafe pointer operations.
We delegate this task to a back-end verifier with the help of annotations.
That is, if the related tools were to support concurrent programs, they were candidates for the back-end.
We used \cavetool \cite{DBLP:conf/cav/Vafeiadis10,DBLP:conf/vmcai/Vafeiadis10} as it can also prove linearizability.

Despite the differences, we point out that the combination of \textsc{blast} and \textsc{ccured} \cite{DBLP:conf/fase/BeyerHJM05} is close to our approach in spirit.
\textsc{ccured} performs a type check of the program under scrutiny which checks for unsafe memory operations.
While doing so, it annotates pointer operations in the program with run-time checks in case the type check could not establish the operation to be safe.
The run-time checks are then discharged using \textsc{blast}.
The approach is limited to sequential programs.
Moreover, we incorporate the behavior of the SMR.
Finally, our type system is more lightweight and we discharge the invariants in a simpler semantics without memory deletions.

\citet{DBLP:conf/ecoop/CastegrenW17} give a type system that guarantees the absence of data races.
Types encode a notion of ownership that prevents non-owning threads from accessing a node.
Their method is tailored towards GC and requires to rewrite programs with appropriate type specifiers.
Recently, \citet{DBLP:conf/esop/KuruG19} presented a type system for checking the correct use of RCU.
Unlike our approach, they integrate a fixed shape analysis and a fixed RCU specification.
This makes the type system considerably more complicated and the type check potentially more expensive.
Unfortunately, \citet{DBLP:conf/esop/KuruG19} did not implement their approach.

Besides memory safety, tools like \textsc{invader}, \textsc{slayer}, \textsc{infer}, \textsc{forester}, \textsc{predator}, and the type system by \citet{DBLP:conf/esop/KuruG19} discover memory leaks.
A successful type check with our type system does not imply the absence of memory leaks. 
We believe that the outcome of our analysis could help a leak detection tool.
For example, by performing a linearizability check to find the abstract data type the data structure under consideration implements.
We consider a closer investigation of the matter as future work.

\paragraph{Typestate}
Typestate \cite{DBLP:journals/tse/StromY86} extents an object's type to carry a notion of state.
The methods of an object can be annotated to modify this state and to be available only in a certain state.
Existing analyses checking for methods being called only in the appropriate state include \cite{DBLP:conf/pldi/FosterTA02,DBLP:conf/pldi/FahndrichD02,DBLP:conf/ecoop/DeLineF04,DBLP:conf/issta/FinkYDRG06,DBLP:conf/oopsla/BierhoffA07}.
Our types can be understood as typestates for pointers (and the objects they reference) geared towards SMR.
However, whereas an object's typestate has a global character, our types reflect a thread's local perception.
\citet{DBLP:conf/pldi/DasLS02} give a typestate analysis based on symbolic execution to increase precision.
Similarly, we increase the applicability of our approach by using annotations that are discharged by a back-end verifier.
For a more detailed overview on typestate, refer to \cite{DBLP:journals/ftpl/AnconaBB0CDGGGH16}.

\paragraph{Program Logics}
There are several program logics for verifying concurrent programs with heap.
Examples are:
\begin{inparaitem}[]
	\item \textsc{sagl} \cite{DBLP:conf/esop/FengFS07},
	\item \textsc{rgsep} \cite{DBLP:conf/concur/VafeiadisP07} (used by \textsc{cave} \cite{DBLP:conf/vmcai/Vafeiadis10}),
	\item \textsc{lrg} \cite{DBLP:conf/popl/Feng09},
	\item Deny-Guarantee \cite{DBLP:conf/esop/DoddsFPV09},
	\item \textsc{cap} \cite{DBLP:conf/ecoop/Dinsdale-YoungDGPV10},
	\item \textsc{hlrg} \cite{DBLP:conf/concur/FuLFSZ10}, and
	\item the work by \citet{DBLP:conf/esop/GotsmanRY13}.
\end{inparaitem}
Program logics are conceptually related to our type system.
However, such logics integrate further ingredients to successfully verify intricate lock-free data structures \cite{DBLP:conf/oopsla/TuronVD14}.
Most importantly, they include memory abstractions, like (concurrent) separation logic \cite{DBLP:conf/csl/OHearnRY01,DBLP:conf/lics/Reynolds02separationlogic,DBLP:conf/concur/OHearn04,DBLP:conf/concur/Brookes04}, and mechanisms to reason about thread interference, like rely-guarantee \cite{DBLP:journals/toplas/Jones83}.
This makes them much more complex than our type system.
We deliberately avoid incorporating a memory abstraction into our type system to keep it as flexible as possible.
Instead, we use annotations to delegate the shape analysis to a back-end verifier, achieving modularity in verifying the data structure and its memory management separately.
Moreover, accounting for thread interference in our type system boils down to defining guarantees as closed sets of locations and removing guarantee $\gactive$ upon exiting atomic blocks.

Oftentimes, invariant-based reasoning about interference turns out too restrictive for verification.
To overcome this, program logics like
\begin{inparaitem}[]
	\item \textsc{caresl} \cite{DBLP:conf/icfp/TuronDB13},
	\item \textsc{fcsl} \cite{DBLP:conf/esop/NanevskiLSD14},
	\item \textsc{icap} \cite{DBLP:conf/esop/SvendsenB14},
	\item \textsc{tada} \cite{DBLP:conf/ecoop/PintoDG14},
	\item \textsc{gps} \cite{DBLP:conf/oopsla/TuronVD14}, and
	\item \textsc{iris} \cite{DBLP:conf/popl/JungSSSTBD15}
\end{inparaitem}
make use of protocols.
A protocol captures possible thread interference, for example, using state transition systems.
(Rely-guarantee is a particular instantiation of a protocol \cite{DBLP:conf/icfp/TuronDB13,DBLP:conf/popl/JungSSSTBD15}.)
In our approach, SMR automata are protocols that govern memory deletions and protections, that is, describe the influence of SMR-related actions among threads.
Our types describe a thread's local, per-pointer perception of that global protocol.

Besides protocols, recent logics like \textsc{caresl}, \textsc{tada}, and \textsc{iris} integrate reasoning in the spirit of atomicity abstraction/refinement \cite{Lipton75,Dijkstra1982}.
Intuitively, they allow the client of a fine-grained module to be verified against a coarse-grained specification of the module.
For example, a client of a data structure can be verified against its abstract data type, provided the data structure refines the abstract data type.
Following \cite{DBLP:journals/pacmpl/MeyerW19}, we use the same idea wrt. SMR algorithms: we consider SMR automata instead of the actual SMR implementations.

Some program logics can also unveil memory leaks \cite{DBLP:conf/esop/GotsmanRY13,DBLP:journals/pacmpl/BizjakGKB19}.

\paragraph{Linearizability}
Linearizability testing \cite{DBLP:conf/pldi/VechevY08,DBLP:conf/fm/LiuCLS09,DBLP:conf/pldi/BurckhardtDMT10,DBLP:conf/cav/CernyRZCA10,DBLP:conf/icse/Zhang11a,DBLP:journals/tse/Liu0L0ZD13,DBLP:conf/hvc/TravkinMW13,DBLP:conf/pldi/EmmiEH15,DBLP:conf/forte/HornK15a,DBLP:journals/concurrency/Lowe17,DBLP:journals/corr/YangKLW17,DBLP:journals/pacmpl/EmmiE18} is a bug hunting technique to find non-linearizable executions in large code bases.
Since we focus on verification, we do not go into the details of linearizability testing.
However, it could be worthwhile to use a linearizability tester instead of a verification back-end in our approach to provide faster feedback during the development process and only use a verifier once the development is considered finished.

Verification techniques for linearizability fall into two categories: manual techniques (including tool-supported but not fully automated techniques) and automatic techniques.
Manual approaches require the human checker to have a deep understanding of the proof techniques as well as the program under scrutiny---in our case, this includes a deep understanding of the lock-free data structure as well as the SMR implementation.
This may be the reason why many manual proofs do not consider reclamation \cite{DBLP:journals/entcs/ColvinDG05,DBLP:conf/cav/ColvinGLM06,DBLP:conf/iceccs/Groves07,DBLP:conf/cats/Groves08,DBLP:conf/wdag/DohertyM09,DBLP:conf/tacas/ElmasQSST10,DBLP:conf/podc/OHearnRVYY10,DBLP:journals/fac/BaumlerSTR11,DBLP:journals/toplas/DerrickSW11,DBLP:journals/fac/Jonsson12,DBLP:conf/popl/LiangFF12,DBLP:conf/pldi/LiangF13,DBLP:journals/toplas/LiangFF14,DBLP:conf/wdag/HemedRV15,DBLP:conf/pldi/SergeyNB15,DBLP:conf/esop/SergeyNB15,DBLP:conf/cav/BouajjaniEEM17,DBLP:conf/ecoop/DelbiancoSNB17,DBLP:conf/esop/KhyzhaDGP17,DBLP:conf/cav/SchellhornWD12,DBLP:conf/concur/HenzingerSV13}.
There are fewer works that consider reclamation \cite{DBLP:conf/forte/DohertyGLM04,DBLP:conf/popl/DoddsHK15,DBLP:journals/pacmpl/KrishnaSW18,DBLP:conf/popl/ParkinsonBO07,DBLP:conf/concur/FuLFSZ10,DBLP:conf/ictac/TofanSR11,DBLP:conf/esop/GotsmanRY13}.
(The work by \citet{DBLP:conf/esop/GotsmanRY13} checks memory safety and discovers memory leaks as well.)
For a more detailed overview of manual techniques, we refer to the survey by \mbox{\citet{DBLP:journals/corr/DongolD14}}.

The landscape of related work for automated linearizability proofs is similar to its manual counterpart. 
Most automated approaches ignore memory reclamation, that is, assume a garbage collector \cite{DBLP:conf/cav/AmitRRSY07,DBLP:conf/cav/BerdineLMRS08,DBLP:conf/aplas/SegalovLMGS09,DBLP:conf/spin/VechevYY09,DBLP:conf/vmcai/Vafeiadis10,DBLP:conf/cav/Vafeiadis10,DBLP:conf/spin/SethiTM13,DBLP:conf/cav/ZhuPJ15,DBLP:conf/sas/AbdullaJT16}.
When reclamation is not considered, memory abstractions are simpler and more efficient, they can exploit ownership guarantees~\cite{DBLP:conf/sas/Boyland03,DBLP:conf/popl/BornatCOP05} and the resulting thread-local reasoning techniques~\cite{DBLP:conf/csl/OHearnRY01,DBLP:conf/lics/Reynolds02separationlogic}.
Few works \cite{DBLP:conf/tacas/AbdullaHHJR13,DBLP:conf/vmcai/HazizaHMW16,DBLP:conf/sas/HolikMVW17,DBLP:journals/pacmpl/MeyerW19} address the challenge of verifying lock-free data structures under manual memory management.
Besides \citet{DBLP:journals/pacmpl/MeyerW19}, they use hand-crafted semantics that allow for accessing deleted memory.
The work by \citet{DBLP:journals/pacmpl/MeyerW19} is the closest related.
We build on their programming model and their reduction result as discussed in \Cref{sec:preliminaries,sec:prf}, respectively.
Moreover, we rely to their results for proving an SMR implementation against an SMR automaton.

\paragraph{Moverness}
Movers where first introduced by \citet{Lipton75}.
They were later generalized to arbitrary safety properties~\cite{DBLP:conf/popl/Doeppner77,DBLP:conf/parle/Back89,lamport1989pretending}.
Movers are a widely applied enabling technique for verification.
To ease the verification task, the program is made \emph{more atomic} without cutting away behavior.
Because we use standard moverness arguments, we do not give an extensive overview.
\citet{DBLP:conf/pldi/FlanaganQ03,DBLP:journals/toplas/FlanaganFLQ08}~use a type system to find movers in Java programs.
The \textsc{calvin} tool \cite{DBLP:conf/cav/FlanaganQS02,DBLP:journals/tcs/FlanaganFQS05,DBLP:journals/jot/FreundQ04} applies movers to establish pre/post conditions of functions in concurrent programs using sequential verifiers.
Similarly, \textsc{qed} \cite{DBLP:conf/popl/ElmasQT09} rewrites concurrent code into sequential code based on movers.
These approaches are similar to ours in spirit: they take the verification task to a much simpler semantics.
However, movers are not a key aspect of our approach.
We employ them only to increase the applicability of our tool in case of benign pointer races.
\citet{DBLP:conf/tacas/ElmasQSST10}~extend \textsc{qed} to establish linearizability for simple lock-free data structures.
\textsc{qed} is superseded by \textsc{civl}~\cite{DBLP:conf/cav/HawblitzelPQT15,DBLP:conf/cav/KraglQ18}.
\textsc{civl} proves programs correct by repeatedly applying movers to a program until its specification is obtained.
The approach is semi-automatic, it takes as input a so-called layered program that contains intermediary steps guiding the transformation~\cite{DBLP:conf/cav/KraglQ18}.
Movers were also applied in the context of relaxed memory \cite{BEMT18}.


\begin{acks}
	We thank the POPL'20 reviewers for their valuable feedback and suggestions for improvements.
\end{acks}

\bibliography{bibliography}


\clearpage
\newpage
\appendix


\section{Example for Hazard Pointers}
\label{appendix:hpexample}

We give a brief example on how our type system infers that a pointer can be accessed safely, i.e., how guarantee $\gsafeaccess$ is obtained, when hazard pointers are used.
\Cref{fig:hp-typing-example} depicts a common usage pattern and its typing.
HP is specified by the SMR automaton $\baseobs\times\hpobs$.
For $\hpobs$ and the HP-specific guarantees consider \Cref{fig:hp-observer}.
We use two HP-specific guarantees $\gilluinv$ and $\gilluisu$ that encode the fact that \code{protect} for the $0$-indexed hazard pointer has been invoked and returned, respectively.

\begin{wrapfigure}{r}{5.7cm}
  \begin{tcolorbox}
\begin{lstlisting}[style=typing]
  // ...
`{ ptr:$\emptyset$ }`
  enter protect(ptr, 0);
`{ ptr:$\gilluinv$ }`
  exit protect;
`{ ptr:$\gilluisu$ }`
  // ...
`{ ptr:$\gilluisu$ }`
  @inv active(ptr)
`{ ptr:$\gilluisu\wedge\gactive$ }`
`{ ptr:$\gilluisu\wedge\gactive\wedge\gsafeaccess$ }`
  // ...
\end{lstlisting}%
	\caption{%
		A typing example for a typical usage pattern of hazard pointers.
	}
	\label{fig:hp-typing-example}
  \end{tcolorbox}
\end{wrapfigure}%
In the beginning, the type of \code{ptr} is $\emptyset$.
There are no guarantees.
Next, \code{protect} is invoked with \code{ptr} as a parameter.
Using Rule~\ref{rule:enter} we obtain type $\gilluinv$ for \code{ptr}, stating that we have definitely invoked \code{protect} for \code{ptr}.
We already discussed in \Cref{sec:example:type-inference} how this is computed.
After \code{protect} returns, \code{ptr} obtains type $\gilluisu$ from an application of Rule~\ref{rule:exit}.
It encodes the fact that a protection has been issued.
Accessing \code{ptr}, however, is not safe at this point since we do not know whether \code{ptr} has been retired in the meantime.
To ensure that the protection has been successful, the subsequent code contains an active annotation.
It adds $\gactive$ to the type of \code{ptr}, Rule~\ref{rule:active}.
An application of Rule~\ref{rule:infer} allows us to infer $\gsafeaccess$ for \code{ptr}.
The type system successfully discovered that it is safe to access \code{ptr}.

\begin{figure}
  \begin{tcolorbox}
	\definecolor{colorF}{RGB} {230,40,40}
	\definecolor{colorR}{RGB} {51,34,136}
	\definecolor{colorINV}{RGB} {226,224,135}
	\definecolor{colorISU}{RGB} {157,199,220}
	\definecolor{colorS}{RGB} {112,181,169}
	\definecolor{colorA}{RGB} {52, 235, 52}

	\newcommand{\EnterA}{P0}
	\newcommand{\EnterB}{P1}
	\newcommand{\EnterAo}{!P0}
	\newcommand{\EnterBo}{!P1}
	\newcommand{\ExitA}{E0}
	\newcommand{\ExitB}{E1}
	\newcommand{\Retire}{R}
	\newcommand{\Free}{F}

	\definecolor{colorEnterA}{RGB} {160,160,40}
	\definecolor{colorEnterAo}{RGB} {160,160,40}
	\definecolor{colorExitA}{RGB} {160,160,40}
	\definecolor{colorEnterB}{RGB} {40,170,140}
	\definecolor{colorEnterBo}{RGB} {40,170,140}
	\definecolor{colorExitB}{RGB} {40,170,140}
	\definecolor{colorRetire}{RGB} {0,0,200}
	\definecolor{colorFree}{RGB} {200,0,0}

	\center
	\begin{tikzpicture}[->,>=stealth',shorten >=1pt,auto,node distance=1.7cm and 2.2cm,thick,initial text={}]
		\node [xshift=-.2cm,yshift=2.0cm,draw,thin] {$\hpobs$};

		\node[initial,state] (S1) {\mkstatename{obs:hp:s1}};
		\node[state] (S2) [right=of S1] {\mkstatename{obs:hp:s2}};
		\node[state] (S3) [right=of S2] {\mkstatename{obs:hp:s3}};
		\node[state] (S4) [right=of S3] {\mkstatename{obs:hp:s4}};
		\node[state] (S5) [below=of S3] {\mkstatename{obs:hp:s5}};
		\node[state] (S6) [below=of S4] {\mkstatename{obs:hp:s6}};
		\node[state] (S16) [below=of S2] {\mkstatename{obs:hp:s16}};
		\node[state] (S7) [below=of S5] {\mkstatename{obs:hp:s7}};
		\node[state] (S8) [below=of S6] {\mkstatename{obs:hp:s8}};
		\coordinate [below=of S7,yshift=-4.5mm,xshift=-4.5mm] (XX) {};
		\node[state] (S9) [left=of XX] {\mkstatename{obs:hp:s9}};
		\node[state] (S10) [below=of S7] {\mkstatename{obs:hp:s10}}; 
		\node[state] (S11) [below=of S8] {\mkstatename{obs:hp:s11}};
		\node[state] (S17) [below=of S9] {\mkstatename{obs:hp:s17}};
		\node[state] (S12) [below=of S10] {\mkstatename{obs:hp:s12}};
		\node[state] (S13) [below=of S11] {\mkstatename{obs:hp:s13}};
		\node[state] (S14) [below=of S12] {\mkstatename{obs:hp:s14}};
		\node[state] (S15) [below=of S13] {\mkstatename{obs:hp:s15}};
		\node[accepting,state,double=white] (SF) [right=of S4] {\mkstatename{obs:hp:sf}};

		\draw[colorEnterB] (S2) edge node[anchor=center,fill=colorEnterB,text=black] {\EnterB} (S16);
		\draw[colorExitA] (S16) edge node[anchor=center,fill=colorExitA,text=black] {\ExitA} (S5);
		\draw[colorExitB] (S16) edge[out=270,in=130] node[anchor=center,fill=colorExitB,text=black] {\ExitB} (S12);

		\draw[colorEnterA] (S9) edge node[anchor=center,fill=colorEnterA,text=black] {\EnterA} (S17);
		\draw[colorExitB] (S17) edge node[anchor=center,fill=colorExitB,text=black,pos=.4] {\ExitB} (S12);
		\draw[colorExitA] (S17) edge[out=60,in=250] node[anchor=center,fill=colorExitA,text=black,pos=.57] {\ExitA} (S5);

		\draw[colorEnterA] (S1) edge node[anchor=center,fill=colorEnterA,text=black] {\EnterA} (S2);
		\draw[colorExitA] (S2) edge node[anchor=center,fill=colorExitA,text=black] {\ExitA} (S3);
		\draw[colorRetire] (S3) edge node[anchor=center,fill=colorRetire,text=white] {\Retire} (S4);
		\draw[colorFree] (S4) edge node[anchor=center,fill=colorFree,text=black] {\Free} (SF);
		\draw[colorEnterB] (S3) edge node[anchor=center,fill=colorEnterB,text=black,pos=.7] {\EnterB} (S5);
		\draw[colorEnterB] (S4) edge node[anchor=center,fill=colorEnterB,text=black] {\EnterB} (S6);
		\draw[colorExitB] (S5) edge node[anchor=center,fill=colorExitB,text=black] {\ExitB} (S7);
		\draw[colorExitB] (S6) edge node[anchor=center,fill=colorExitB,text=black] {\ExitB} (S8);
		\draw[colorRetire] (S5) edge node[anchor=center,fill=colorRetire,text=white] {\Retire} (S6);
		\draw[colorRetire] (S7) edge node[anchor=center,fill=colorRetire,text=white] {\Retire} (S8);
		\draw[colorFree] (S6) edge node[anchor=center,fill=colorFree,text=black] {\Free} (SF);
		\draw[colorFree] (S8) edge node[anchor=center,fill=colorFree,text=black] {\Free} (SF);

		\draw[colorEnterAo,densely dashed] (S3) edge[bend right=23,above right] node[anchor=center,fill=colorEnterAo,text=black] {\EnterAo} (S1);
		\draw[colorEnterAo,densely dashed] (S4) edge[bend right,above right] node[anchor=center,fill=colorEnterAo,text=black] {\EnterAo} (S1);
		\draw[colorEnterAo,densely dashed] (S7) edge node[anchor=center,fill=colorEnterAo,text=black,pos=.65] {\EnterAo} (S10);
		\draw[colorEnterAo,densely dashed] (S8) edge node[anchor=center,fill=colorEnterAo,text=black] {\EnterAo} (S11);
		\draw[colorEnterBo,densely dashed] (S7) edge[bend left] node[anchor=center,fill=colorEnterBo,text=black,pos=.4] {\EnterBo} (S3); 
		\draw[colorEnterBo,densely dashed] (S8) edge[bend left] node[anchor=center,fill=colorEnterBo,text=black,pos=.4] {\EnterBo} (S4); 

		\draw[colorEnterB] (S1) edge[bend right=10] node[anchor=center,fill=colorEnterB,text=black] {\EnterB} (S9);
		\draw[colorExitB] (S9) edge node[anchor=center,fill=colorExitB,text=black,pos=.7] {\ExitB} (S10);
		\draw[colorRetire] (S10) edge node[anchor=center,fill=colorRetire,text=white] {\Retire} (S11);
		\draw[colorRetire] (S12) edge node[anchor=center,fill=colorRetire,text=white] {\Retire} (S13);
		\draw[colorRetire] (S14) edge node[anchor=center,fill=colorRetire,text=white] {\Retire} (S15);
		\draw[colorEnterA] (S10) edge node[anchor=center,fill=colorEnterA,text=black] {\EnterA} (S12);
		\draw[colorEnterA] (S11) edge node[anchor=center,fill=colorEnterA,text=black,pos=.6] {\EnterA} (S13);
		\draw[colorExitA] (S12) edge node[anchor=center,fill=colorExitA,text=black] {\ExitA} (S14);
		\draw[colorExitA] (S13) edge node[anchor=center,fill=colorExitA,text=black,pos=.3] {\ExitA} (S15);
		\draw[colorFree] (S11) edge[bend right=12] node[anchor=center,fill=colorFree,text=black] {\Free} (SF);
		\draw[colorFree] (S13) edge[bend right=12] node[anchor=center,fill=colorFree,text=black] {\Free} (SF);
		\draw[colorFree] (S15) edge[bend right=12] node[anchor=center,fill=colorFree,text=black] {\Free} (SF);

		\draw[colorEnterBo,densely dashed] (S10) edge[bend left=15] node[anchor=center,fill=colorEnterBo,text=black] {\EnterBo} (S1); 
		\draw[colorEnterBo,densely dashed] (S11) edge[bend left=15] node[anchor=center,fill=colorEnterBo,text=black] {\EnterBo} (S1); 
		\draw[colorEnterAo,densely dashed] (S14) edge[bend left] node[anchor=center,fill=colorEnterAo,text=black,pos=.6] {\EnterAo} (S10);
		\draw[colorEnterAo,densely dashed] (S15) edge[bend left] node[anchor=center,fill=colorEnterAo,text=black,pos=.875] {\EnterAo} (S11); 
		\draw[colorEnterBo,densely dashed] (S14) edge[bend left] node[anchor=center,fill=colorEnterBo,text=black,pos=.3] {\EnterBo} (S3); 
		\draw[colorEnterBo,densely dashed] (S15) edge node[anchor=center,fill=colorEnterBo,text=black] {\EnterBo} (S3); 

		\newcommand{\drawpartX}[5]{ 
			\draw[fill=#5,draw=none] (#1.center) -- ($(#1.center) + (#2:#4)$) arc (#2:#3:#4) -- cycle;
		}
		\newcommand{\drawtypeEinv}[2][4mm]{\drawpartX{#2}{0}{180}{#1}{colorINV}}
		\newcommand{\drawtypeEisu}[2][4mm]{\drawpartX{#2}{180}{360}{#1}{colorISU}}
		\begin{scope}[on background layer]
			\drawtypeEinv{S2}
			\drawtypeEinv{S3}
			\drawtypeEinv{S4}
			\drawtypeEinv{S5}
			\drawtypeEinv{S6}
			\drawtypeEinv{S7}
			\drawtypeEinv{S8}
			\drawtypeEinv{SF}
			\drawtypeEinv{S16}
			\drawtypeEinv{S17}
			\drawtypeEinv{S12}
			\drawtypeEinv{S13}
			\drawtypeEinv{S14}
			\drawtypeEinv{S15}

			\drawtypeEisu{S3}
			\drawtypeEisu{S4}
			\drawtypeEisu{S5}
			\drawtypeEisu{S6}
			\drawtypeEisu{S7}
			\drawtypeEisu{S8}
			\drawtypeEisu{S14}
			\drawtypeEisu{S15}
			\drawtypeEisu{SF}
		\end{scope}

		\newcommand{\makekey}[3]{ 
			\begin{scope}
				\node[draw,circle,minimum width=5mm] (foo) at (#3,-14.0) {};
				\node[anchor=west] at (#3,-14.0) {$\;\;\,\in\locsof{#1}$};
				\begin{scope}[on background layer]
					#2[2.5mm]{foo}
				\end{scope}
			\end{scope}
		}
		\makekey{\gilluinv}{\drawtypeEinv}{2.5}
		\makekey{\gilluisu}{\drawtypeEisu}{7.0}
	\end{tikzpicture}
	\begin{align*}
		\color{colorEnterA}{\EnterA} ~&:=~ \translab{\evt{\enterof{\guard}}{\athread,\anadr,0}}{\athread=\threadvar\wedge\anadr=\adrvar}
		&
		\color{colorEnterB}{\EnterB} ~&:=~ \translab{\evt{\enterof{\guard}}{\athread,\anadr,1}}{\athread=\threadvar\wedge\anadr=\adrvar}
		\\
		\color{colorEnterAo}{\EnterAo} ~&:=~ \translab{\evt{\enterof{\guard}}{\athread,\anadr,0}}{\athread=\threadvar\wedge\anadr\neq\adrvar}
		&
		\color{colorEnterBo}{\EnterBo} ~&:=~ \translab{\evt{\enterof{\guard}}{\athread,\anadr,1}}{\athread=\threadvar\wedge\anadr\neq\adrvar}
		\\
		\color{colorExitA}{\ExitA} ~&:=~ \translab{\evt{\exitof{\guard}}{\athread}}{\athread=\threadvar}
		&
		\color{colorExitB}{\ExitB} ~&:=~ \translab{\evt{\exitof{\guard}}{\athread}}{\athread=\threadvar}
		\\
		\color{colorRetire}{\Retire} ~&:=~ \translab{\evt{\enterof{\retire}}{\athread,\anadr}}{\anadr=\adrvar}
		&
		\color{colorFree}{\Free} ~&:=~ \translab{\freeof{\anadr}}{\anadr=\adrvar}
	\end{align*}
	\caption{%
		The SMR automaton $\hpobs$ specifies the Hazard Pointer method for two hazard pointers per thread.
	}
	\label{fig:hp-observer}
  \end{tcolorbox}
\end{figure}

\section{Missing Details}
\label{appendix:missing_details}

\subsection{Definitions}
\label{appendix:definitions}

\begin{definition}
	The liberal semantics is defined by the following rules, assuming $\tau\in\asem[\aprog]{X}{Y}$ and that $\anact$ respects the control flow of $\aprog$.
	\begin{description}
		\item[(Assign1)]
			If $\anact=(\athread,\psel{\apvar}:=\apvarp,[\psel{\anadr}\mapsto\anadrp])$ then $\heapcomputof{\tau}{\apvar}=\anadr$ and $\heapcomputof{\tau}{\apvarp}=\anadrp$. 
		\item[(Assign2)]
			If $\anact=(\athread,\apvar:=\apvarp,[\apvar\mapsto\heapcomputof{\tau}{\apvarp}])$.
		\item[(Assign3)]
			If $\anact=(\athread,\apvar:=\psel{\apvarp},[\apvar\mapsto\heapcomputof{\tau}{\psel{\anadr}}])$ with $\heapcomputof{\tau}{\apvarp}=\anadr\in\adr$.
		\item[(Assign4)]
			If $\anact=(\athread,\advar:=\opof{\advarp_1,\dots, \advarp_n},[\advar\mapsto\advalue])$ with $\advalue=\opof{\heapcomputof{\tau}{\advarp_1},\dots, \heapcomputof{\tau}{\advarp_n}}$.
		\item[(Assign5)]
			If $\anact=(\athread,\dsel{\apvar}:=\advarp,[\dsel{\anadr}\mapsto\heapcomputof{\tau}{\advarp}])$ with $\heapcomputof{\tau}{\apvar}=\anadr\in\adr$.
		\item[(Assign6)]
			If $\anact=(\athread,\advar:=\dsel{\apvarp},[\advar\mapsto\heapcomputof{\tau}{\dsel{\anadr}}])$ with $\heapcomputof{\tau}{\apvarp}=\anadr\in\adr$.
		\item[(Assume)]
			If $\anact=(\athread,\assumeof{\lhs\triangleq\rhs},\emptyset)$ then $\heapcomputof{\tau}{\lhs}\triangleq\heapcomputof{\tau}{\rhs}$.
		\item[(Malloc)]
			If $\anact=(\athread,\apvar:=\malloc, \anup)$, then $\anup$ has the form $\apvar\mapsto\anadr,\psel{\anadr}\mapsto\segval,\dsel{\anadr}\mapsto\advalue$ so that $\anadr\in\freshof{\tau}$ or $\anadr\in\freedof{\tau}\cap Y$.
		\item[(Atomic)]
			If $\anact=(\athread,\atomicbegin,\emptyset)$ or $\anact=(\athread,\atomicend,\emptyset)$.
		\item[(Free)]
			If $\anact=(\bot,\freeof{\anadr},\emptyset)$ then $\anadr\in X$.
		\item[(Enter)]
			$\anact=(\athread,\enterof{\afuncof{\vecof{\apvar},\vecof{\advar}}},\emptyset)$, then $\heapcomputof{\tau}{\apvar}\neq\segval$ for every $\apvar$ in $\vecof{\apvar}$.
		\item[(Exit)]
			$\anact=(\athread,\exitof{\afunc},\emptyset)$.
		\item[(Invariant1)]
			$\anact=(\athread,\ghostof{\chooseof{\aghostvar}},\emptyset)$.
		\item[(Invariant2)]
			$\anact=(\athread,\invariantof{\containsof{\apvar}{\aghostvar}},\emptyset)$.
		\item[(Invariant3)]
			$\anact=(\athread,\invariantof{\activeof{\apvar}},\emptyset)$.
		\item[(Invariant4)]
			$\anact=(\athread,\invariantof{\apvar=\apvarp},\emptyset)$.
	\end{description}
\end{definition}

\begin{definition}
	\label{def:historyof}
	The history induced by a computation $\tau$, denoted $\historyof{\tau}$, is defined by:
	\begin{align*}
		\historyof{\epsilon} = &~ \epsilon \\
		\historyof{\tau.(\athread,\freeof{\anadr},\anup,\apc)} = &~ \historyof{\tau} \hconcat \freeof{\anadr} \\
		\historyof{\tau.(\athread,\enterof{\afuncof{\vecof{\apvar},\vecof{\advar}}},\anup,\apc)} = &~ \historyof{\tau} \hconcat \evt{\afunc}{\athread,\heapcomputof{\tau}{\vecof{\apvar}},\heapcomputof{\tau}{\vecof{\advar}}} \\
		\historyof{\tau.(\athread,\exitof{\afunc},\anup,\apc)} = &~ \historyof{\tau} \hconcat \exitof{\afuncof{\athread}} \\
		\historyof{\tau.\anact} = &~  \historyof{\tau} &&\text{otherwise.}
	\end{align*}
\end{definition}

\begin{definition}
	\label{def:freshof}
	The fresh addresses in a computation $\tau$, denoted by $\freshof{\tau}$, are defined by:
	\begin{align*}
		\freshof{\epsilon} = &~ \adr
		\\
		\freshof{\tau.\anact} = &~ \freshof{\tau}\setminus\set{\anadr}
		&&\text{if } \comof{\anact}\equiv\freeof{\anadr}
		\\
		\freshof{\tau.\anact} = &~ \freshof{\tau}\setminus\set{\anadr}
		&&\text{if } \comof{\anact}\equiv\apvar:=\malloc \wedge \heapcomputof{\tau.\anact}{\apvar}=\anadr
		\\
		\freshof{\tau.\anact} = &~ \freshof{\tau}
		&&\text{otherwise.}
	\end{align*}
	The definition carries over naturally to histories.
\end{definition}

\begin{definition}
	\label{def:freedof}
	The freed addresses in a computation $\tau$, denoted by $\freedof{\tau}$, are defined by:
	\begin{align*}
		\freedof{\epsilon} = &~ \emptyset
		\\
		\freedof{\tau.\anact} = &~ \freedof{\tau} \cup \set{\anadr}
		&&\text{if } \comof{\anact}\equiv\freeof{\anadr}
		\\
		\freedof{\tau.\anact} = &~ \freedof{\tau} \setminus \set{\anadr}
		&&\text{if } \comof{\anact}\equiv\apvar:=\malloc \wedge \heapcomputof{\tau.\anact}{\apvar}=\anadr
		\\
		\freedof{\tau.\anact} = &~ \freedof{\tau}
		&&\text{otherwise.}
	\end{align*}
\end{definition}

\begin{definition}
	\label{def:retiredof}
	The retired addresses in a computation $\tau$, denoted by $\retiredof{\tau}$, are defined by:
	\begin{align*}
			\retiredof{\epsilon} = &~ \emptyset
			\\
			\retiredof{\tau.\anact} = &~ \retiredof{\tau} \cup \set{\anadr}
			&&\text{if } \comof{\anact}\equiv\enterof{\retireof{\apvar}}\wedge\anadr=\heapcomputof{\tau}{\apvar}
			\\
			\retiredof{\tau.\anact} = &~ \retiredof{\tau} \setminus \set{\anadr}
			&&\text{if } \comof{\anact}\equiv\freeof{\anadr}
			\\
			\retiredof{\tau.\anact} = &~ \retiredof{\tau}
			&&\text{otherwise.}
	\end{align*}
\end{definition}

\begin{definition}
	\label{def:activeof}
	The active addresses in a computations $\tau$ are: \[ \activeofcomp{\tau}:=\adr\setminus(\freedof{\tau}\cup\retiredof{\tau}) \ .\]
\end{definition}

\begin{definition}[Angel Denotation]
	Consider some $\tau\in\allsem$.
	Let $\invholdsof{\tau}$ have the prenex normal form $\exists \aghostvar_1\ldots \exists \aghostvar_n.\phi$, where $\phi$ is quantifier-free.
	Let $\aghostvar_n$ be the instance of angel $\aghostvar$ resulting from the last allocation in $\tau$.
	The set of addresses possibly represented by angel $\aghostvar$ after computation $\tau$ is
	\begin{align*}
		\denotationof{\tau}{\aghostvar}
		:=
		\setcond{\anadr\in\adr}{
			\exists A_1, \ldots, A_n\subseteq \adr.\anadr\in A_n\wedge (A_1, \ldots, A_n)
			\models
			\phi
		}
		\ .
	\end{align*}
\end{definition}

\begin{definition}[Valid Expressions]
	\label{Definition:Validity}
	The \emph{valid pointer expressions} in a computation $\tau\in\allsemobs$, denoted by $\validof{\tau}\subseteq\pexp$, are defined by:
	\begin{align*}
		\validof{\epsilon}&:=\pvars
			\\
		\validof{\tau.(\athread, \apvar:=\apvarp, \anup)}
			&:= \validof{\tau}\cup\set{\apvar}
			&&\text{if }\apvarp\in \validof{\tau}
			\\
		\validof{\tau.(\athread, \apvar:=\apvarp, \anup)}
			&:= \validof{\tau}\setminus\set{\apvar}
			&&\text{if }\apvarp\notin \validof{\tau}
			\\
		\validof{\tau.(\athread, \psel{\apvar}:=\apvarp, \anup)}
			&:= \validof{\tau}\cup\set{\psel{\anadr}}
			&&\text{if }\heapcomputof{\tau}{\apvar}=\anadr\in\adr \wedge \apvarp\in \validof{\tau}
			\\
		\validof{\tau.(\athread, \psel{\apvar}:=\apvarp, \anup)}
			&:= \validof{\tau}\setminus\set{\psel{\anadr}}
			&&\text{if }\heapcomputof{\tau}{\apvar}=\anadr\in\adr \wedge \apvarp\notin \validof{\tau}
			\\
		\validof{\tau.(\athread, \apvar:=\psel{\apvarp}, \anup)}
			&:= \validof{\tau}\cup\set{\apvar}
			&&\text{if }\heapcomputof{\tau}{\apvarp}=\anadr\in\adr \wedge \psel{\anadr}\in\validof{\tau}
			\\
		\validof{\tau.(\athread, \apvar:=\psel{\apvarp}, \anup)}
			&:= \validof{\tau}\setminus\set{\apvar}
			&&\text{if }\heapcomputof{\tau}{\apvarp}=\anadr\in\adr \wedge \psel{\anadr}\notin\validof{\tau}
			\\
		\validof{\tau.(\athread, \freeof{\anadr}, \anup)}
			&:=\validof{\tau}\setminus \invalidof{\anadr}
			\\
		\validof{\tau.(\athread, \apvar:=\malloc, \anup)}
			&:=\validof{\tau}\cup\set{\apvar,\psel{\anadr}}
			&&\text{if }[\apvar\mapsto\anadr]\in\anup
			\\
		\validof{\tau.(\athread, \assumeof{\apvar=\apvarp}, \anup)}
			&:=\validof{\tau}\cup\set{\apvar,\apvarp}
			&&\text{if }\set{\apvar,\apvarp}\cap\validof{\tau}\neq\emptyset
			\\
		\validof{\tau.(\athread, \anact, \anup)}
			&:=\validof{\tau}
			&&\text{otherwise.}
	\end{align*}
	We have $\invalidof{\anadr}:=\setcond{\apvar}{\heapcomputof{\tau}{\apvar}=\anadr} \cup \setcond{\psel{\anadrp}}{\heapcomputof{\tau}{\psel{\anadrp}}=\anadr} \cup \set{\psel{\anadr}}$.
\end{definition}

\begin{definition}[Observer Behavior]
	The behavior allowed by $\anobs$ on address $\anadr$ after history $\ahist$, denoted by $\freeableof[\anobs]{\ahist}{\anadr}$, is the set
	$\freeableof[\anobs]{\ahist}{\anadr} := \setcond{\ahistp}{\ahist.\ahistp\in\specof{\anobs}\wedge\freesof{\ahistp}\subseteq{\anadr}}$.
	If clear from the context, we just write $\freeableof[]{\ahist}{\anadr}$.
\end{definition}

\begin{definition}[Unsafe Access]
	A computation $\tau.\anact$ raises an \emph{unsafe access} if $\comof{\anact}$ contains $\dsel{\apvar}$ or $\psel{\apvar}$ with $\apvar\notin\validof{\tau}$.
\end{definition}

\begin{definition}[Unsafe Assumption]
	\label{def:unsafe-assumption}
	A computation $\tau.\anact$ raises an \emph{unsafe assumption} if $\comof{\anact}$ is of the form $\assumeof{\apvar=\apvarp}$ with $\set{\apvar,\apvarp}\not\subseteq\validof{\tau}$.
\end{definition}


\begin{definition}[Unsafe Retire]
	A computation $\tau.\anact$ raises an \emph{unsafe retire} if $\comof{\anact}$ is of the form $\enterof{\retireof{\apvar}}$ with $\apvar\notin\validof{\tau}$.
\end{definition}

\begin{definition}[Pointer Race]
	A computation $\tau.\anact$ raises a \emph{pointer race (PR)} if it raises
	\begin{inparaenum}[(i)]
		\item an unsafe access,
		\item an unsafe assumption,
		\item an unsafe call, or
		\item an unsafe retire.
	\end{inparaenum}
	It is \emph{pointer race free (PRF)} if none of its prefixes raises a PR.
\end{definition}

\begin{definition}[Renaming]
	\label{def:renaming}
	A \emph{renaming of address $\anadr$ and $\anadrp$ in a history $\ahist$}, denoted by $\renamingof{\ahist}{\anadr}{\anadrp}$, replaces in $\ahist$ every occurrence of $\anadr$ with $\anadrp$, and vice versa, as follows:
	\begin{align*}
		\renamingof{\epsilon}{\anadr}{\anadrp} = &\;\epsilon
		\\
		\renamingof{\big(\ahist.\afunc(\vecof{\anadrpp},\vecof{\advalue})\big)}{\anadr}{\anadrp} = &\;\big(\renamingof{\ahist}{\anadr}{\anadrp}\big).\big(\afunc(\renamingof{\vecof{\anadrpp}}{\anadr}{\anadrp},\vecof{\advalue})\big)
		\\
		\renamingof{\big(\ahist.\freeof{\anadrpp}\big)}{\anadr}{\anadrp} = &\;\big(\renamingof{\ahist}{\anadr}{\anadrp}\big).\big(\freeof{\renamingof{\anadrpp}{\anadr}{\anadrp}}\big)
		\\
		\renamingof{\ahist.\anevent}{\anadr}{\anadrp} = &\;\renamingof{\ahist}{\anadr}{\anadrp}.\anevent \qquad\text{otherwise.}
	\end{align*}
	where $\renamingof{\anadr}{\anadr}{\anadrp}=\anadrp$, $\renamingof{\anadrp}{\anadr}{\anadrp}=\anadr$, and $\renamingof{\anadrpp}{\anadr}{\anadrp}=\anadrpp$ for all $\anadr\neq\anadrpp\neq\anadrp$.
\end{definition}

\begin{definition}[Elision Support]
	\label{def:elision-support}
	Observer $\anobs$ \emph{supports elision of memory reuse} if
	\begin{compactenum}[(i)]
		\item \label{def:elision-support:frees} $\freeableof{\ahist.\freeof{\anadr}}{\anadrp}=\freeableof{\ahist}{\anadrp}$ for all $\ahist,\anadr,\anadrp$ with $\anadr\neq\anadrp$ and $\ahist.\freeof{\anadr}\in\specof{\anobs}$,
		\item $\freeableof{\ahist}{\anadrpp}=\freeableof{\renamingof{\ahist}{\anadr}{\anadrp}}{\anadrpp}$ for all $\ahist,\anadr,\anadrp,\anadrpp$ with $\anadr\neq\anadrpp\neq\anadrp$,
		\item \label{def:elision-support:fresh-brandnew} $\freeableof{\ahist}{\anadr}\subseteq\freeableof{\renamingof{\ahist}{\anadr}{\anadrp}}{\anadr}$ for all $\ahist,\anadr,\anadrp$ with $\anadr\notin\retiredof{\ahist}$ and $\anadrp\in\freshof{\ahist}$, and
		\item \label{def:elision-support:generosity} $\freeableof{\ahist.\freeof{\anadr}}{\anadr}\subseteq\freeableof{\ahist}{\anadr}$ for all $\ahist,\anadr$.
	\end{compactenum}
\end{definition}





\begin{definition}
	The \emph{domain} of a type environment $\env$ is defined by $\domof{\env}=\setcond{\apavar}{\exists\atype.~\typeof{\apavar}{\atype}\in\env}$.
\end{definition}

\begin{definition}
	Consider some $\env$ and $\vecof{\apvar}=\apvar_1,\dots,\apvar_n$ with $\envof{\apvar_i}=\atype_i$.
	Then we define:
	\begin{align*}
		&\safecallof{\env}{\afuncof{\vecof{\apvar},\vecof{\advar}}} = \mathit{true}
		\\\text{iff}~~~
		&\forall\ahist~\forall\vecof{\anadr},\vecof{\anadrp},\anadrpp,\vecof{\advalue}.
		~~
		\big(
			\forall i.~ 
			(
			\anadr_i=\anadrpp\vee\isvalidof{\atype_i}
			)
			\implies
			\anadr_i=\anadrp_i
		\big)
		\\&\qquad\qquad\quad~\wedge
		\freeableof{\ahist.\afunc(\athread,\vecof{\anadrp},\vecof{\advalue})}{\anadrpp}
		\not\subseteq
		\freeableof{\ahist.\afunc(\athread,\vecof{\anadr},\vecof{\advalue})}{\anadrpp}
		\ .
	\end{align*}
\end{definition}




\begin{definition}[Post Image]
	The post image for pointers $\apvar$ and angles $\aghostvar$ is defined by:
	\begin{align*}
		\lpostof{\apvar}{\acom}{\alocset}
		:=&
		\setcond{\alocationp}{
			\exists\,\alocation~\exists\,\varphi~\exists\,\aheap.~~
			(\alocation,\varphi)\trans{\aheap(\acom_\athread)}(\alocationp,\varphi)
			\:\wedge\:
			\alocation\in\alocset
			\:\wedge\:
			\varphi(\threadvar)=\athread
			\:\wedge\:
			\varphi(\adrvar)=\aheap(\apvar)
		}
		\\
		\lpostof{\aghostvar}{\acom}{\alocset}
		:=&
		\setcond{\alocationp}{
			\exists\,\alocation~\exists\,\varphi~\exists\,\aheap.~~
			(\alocation,\varphi)\trans{\aheap(\acom_\athread)}(\alocationp,\varphi)
			\:\wedge\:
			\alocation\in\alocset
			\:\wedge\:
			\varphi(\threadvar)=\athread
		}
	\end{align*}
	where $\aheap(\acom_\athread)$ means the event that results from thread $\athread$ executing $\acom$ under memory $\aheap$.
\end{definition}

\begin{definition}[Relaxed Unsafe Assumption]
	\label{def:unsafe-assumption}
	A computation $\tau.\anact$ raises a \emph{relaxed unsafe assumption} if $\comof{\anact}$ is $\assumeof{\apvar=\apvarp}$ such that there is $\apavar,\apavarp\in\set{\apvar,\apvarp}$ with $\apvar\not\equiv\apavarp$ and $\apavar\notin\validof{\tau}$ and $\freeof{\heapcomputof{\tau}{\apavarp}}\in\historyof{\tau}$.
\end{definition}

\begin{definition}[Relaxed Pointer Race]
	\label{def:relaxed-pointer-race}
	A computation $\tau.\anact$ is a \emph{relaxed pointer race (RPR)} if $\anact$ is
	\begin{inparaenum}[(i)]
		\item an unsafe access,
		\item a relaxed unsafe assumption,
		\item an unsafe call, or
		\item an unsafe retire.
	\end{inparaenum}
\end{definition}

\begin{theorem}[Generalization of \Cref{thm:PRF-guarantee}]
	\label{thm:generalized-RPRF-guarantee}
	If $\anobs$ supports elision and the semantics $\anobs\freesem$ is relaxed-pointer-race-free, then $\anobs\allsem\computrel\nosem$.
\end{theorem}

\begin{remark}[\Cref{thm:generalized-RPRF-guarantee}]
	In practice, programs use a \code{null} constant.
	Since \code{null} is not part of our command language, a program needs to define it itself.
	Ensuring that \code{null} is never written to nor retired can be done easily, using a syntactic checks and assertions, respectively.
	Then, according to \Cref{thm:PRF-guarantee}, any pointer may be compared to \code{null} without risking a (relaxed) pointer race.
	This can increase the applicability of the type system and ease the implementation of tools.
\end{remark}


\subsection{Reduction}

\begin{definition}
	We write $\aheap(e)=\bot$ if $e\notin\domof{\aheap}$.
\end{definition}

\begin{definition}[In-Use Addresses]
	An address $\anadr$ is \emph{in-use} in memory $\aheap$ if $\aheap$ contains a pointer to $\anadr$.
	Formally, the addresses in-use are \(\adrof{\aheap}:=(\rangeof{\aheap}\cup\domof{\aheap})\cap\adr\) where we use $\set{\psel{\anadr}}\cap\adr=\anadr$ and likewise for data selectors.
\end{definition}

\begin{definition}[Restrictions]
	A restriction of $\aheap$ to a set $P\subseteq\pexp$, denoted by $\restrict{\aheap}{P}$, is a new $\aheapp$ with
	\(\domof{\aheapp} := P \cup \dvars \cup \setcond{\dsel{\anadr}\in\dexp}{\anadr\in\aheap(P)}\)
	and $\aheap(e) = \aheapp(e)$ for all $e\in\domof{\aheapp}$.
\end{definition}

\begin{definition}[Computation similarity]
	Two computations $\tau$ and $\sigma$ are similar, denoted by $\tau\computequiv\sigma$, if $\controlof{\tau}=\controlof{\sigma}$ and $\restrict{\heapcomput{\tau}}{\validof{\tau}}=\restrict{\heapcomput{\sigma}}{\validof{\sigma}}$.
\end{definition}

\begin{definition}[Observer Behavior Inclusion]
	Consider $\tau,\sigma\in\allsemobs$.
	Then, $\sigma$ \emph{includes the (observer) behavior} of $\tau$, denoted by $\tau\obsrel\sigma$,
	if $\freeableof{\tau}{\anadr}\subseteq\freeableof{\sigma}{\anadr}$ holds for all $\anadr\in\vadrof{\tau}$.
\end{definition}

\begin{definition}[Computation Relation]
	Two computations are in \textit{computation relation}, denoted $\tau\computrel\sigma$, if $\controlof{\tau}=\controlof{\sigma}$ and $\restrict{\heapcomput{\tau}}{\validof{\tau}}=\restrict{\heapcomput{\sigma}}{\validof{\tau}}$.
\end{definition}


\begin{lemma}
	\label{thm:prefix-closure}
	$\asem[\aprog]{X}{Y}$ is prefix closed by \Cref{assumption:observers-accepting-state}.
\end{lemma}

\begin{lemma}
	\label{thm:adrof-valid-heap-restriction}
	Consider $\tau\in\allsem$.
	Then, $\vadrof{\tau}=(\validof{\tau}\cap\adr)\cup\heapcomputof{\tau}{\validof{\tau}}$.
\end{lemma}

\begin{lemma}
	\label{thm:valideq}
	Consider $\tau,\sigma\in\allsemobs$.
	If $\tau\computequiv\sigma$, then $\validof{\tau}=\validof{\sigma}$.
\end{lemma}

\begin{lemma}
	\label{thm:move-event-from-freeable}
	Let $\anevent=\afuncof{\athread,\vecof{\anadr},\vecof{\advalue}}$.
	Then, $\ahist_1\in\freeableof{\ahist_2.\anevent}{\anadr}$ iff $\anevent.\ahist_1\in\freeableof{\ahist_2}{\anadr}$.
\end{lemma}

\begin{lemma}
	\label{thm:invalidpointers-adr}
	Consider $\tau\in\freesem$ PRF.
	Then, $\adrof{\restrict{\heapcomput{\tau}}{\validof{\tau}}}\cap\heapcomputof{\tau}{\pexp\setminus\validof{\tau}}=\emptyset$.
\end{lemma}

\begin{lemma}
	\label{thm:pointers-to-freed-are-invalid}
	If $\tau\in\allsem$ and $\heapcomputof{\tau}{\apexp}\in\freedof{\tau}$, then $\apexp\notin\validof{\tau}$.
\end{lemma}

\begin{lemma}
	\label{thm:elision-support-lift}
	Let $\tau\in\allsemobs$, $\anadr\in(\freshof{\tau}\cup\freedof{\tau})\setminus\retiredof{\tau}$, and $\anadrp\in\freshof{\tau}$.
	If $\anobs$ supports elision, then $\freeableof{\tau}{\anadr}\subseteq\renamingof{\freeableof{\tau}{\anadrp}}{\anadrp}{\anadr}$.
\end{lemma}

\begin{lemma}
	\label{thm:replacing-addresses-in-computations-new}
	Assume $\anobs$ supports elision.
	Let $\tau\in\freesemobs[\aprog]$, $\anadr\notin\vadrof{\tau}$, and $A\subseteq\adr$ with $\cardof{A}<\infty$.
	Then there is $\sigma\in\freesemobs[\aprog]$ and $\anadrp\in\freshof{\tau}\setminus A$ with:
	\begin{compactitem}
		\item $\tau\computequiv\sigma$ and $\tau\obsrel\sigma$ and $\retiredof{\tau}\subseteq\retiredof{\sigma}\cup\set{\anadr}$ and $\anadr\in\freshof{\sigma}$,
		\item $\anadrp\in\freshof{\sigma}\iff\anadr\in\freshof{\tau}$ and $\freshof{\sigma}\setminus\set{\anadr,\anadrp}=\freshof{\tau}\setminus\set{\anadr,\anadrp}$,
		\item $\freeableof{\tau}{\anadr}=\renamingof{\freeableof{\sigma}{\anadrp}}{\anadrp}{\anadr}$ and $\renamingof{\freeableof{\tau}{\anadrp}}{\anadrp}{\anadr}=\freeableof{\sigma}{\anadr}$,
		\item $\forall\anadrpp.~\anadr\neq\anadrpp\neq\anadrp\implies\freeableof{\tau}{\anadrpp}=\freeableof{\sigma}{\anadrpp}$, and
		\item $\forall\anexp,\anexpp\in\pvars\cup\setcond{\psel{\anadrpp}}{\anadrpp\in\heapcomputof{\tau}{\validof{\tau}}}.~\heapcomputof{\tau}{\anexp}\neq\heapcomputof{\tau}{\anexpp}\implies\heapcomputof{\sigma}{\anexp}\neq\heapcomputof{\sigma}{\anexpp}$.
	\end{compactitem}
\end{lemma}

\begin{lemma}
	\label{thm:elision-computations-new}
	Assume that $\anobs$ supports elision and that $\freesemobs[\aprog]$ is PRF.
	Then, for every $\tau\in\allsemobs[\aprog]$ there is some $\sigma\in\freesemobs[\aprog]$ with $\tau\computequiv\sigma$, $\tau\obsrel\sigma$, and $\retiredof{\tau}\subseteq\retiredof{\sigma}$.
	Moreover, $\heapcomputof{\tau}{\anexp}\neq\heapcomputof{\tau}{\anexpp}$ implies $\heapcomputof{\sigma}{\anexp}\neq\heapcomputof{\sigma}{\anexpp}$ for all $\anexp,\anexpp\in\pvars\cup\setcond{\psel{\anadrp}}{\anadrp\in\heapcomputof{\tau}{\validof{\tau}}}$.
\end{lemma}

\begin{lemma}
	\label{thm:fresh-notin-range}
	If $\tau\in\allsem$ and $\anadr\in\freshof{\tau}$, then $\anadr\notin\rangeof{\heapcomput{\tau}}$.
\end{lemma}

\begin{lemma}
	\label{thm:disjoint-fresh-retired}
	If $\tau\in\allsem$, then $\freshof{\tau}\cap\retiredof{\tau}=\emptyset$.
\end{lemma}

\begin{lemma}
	\label{thm:invalid-freed-or-fresh}
	If $\tau\in\freesem$ PRF and $\apexp\in\pexp$ with $\apexp\notin\validof{\tau}$, then $\heapcomputof{\tau}{\apexp}\in\freedof{\tau}$ or $\apexp\equiv\psel{\anadr}\wedge\anadr\in\freshof{\tau}\cup\freedof{\tau}$.
\end{lemma}

\begin{lemma}
	\label{thm:invalid-freed}
	If $\tau\in\freesem$ PRF and $\apvar\in\pvars$ with $\apvar\notin\validof{\tau}$, then $\heapcomputof{\tau}{\apvar}\in\freedof{\tau}$.
\end{lemma}


\begin{lemma}
	\label{thm:freeable-vs-deletable-new}
	Assume $\anobs$ supports elision.
	Consider $\tau.\anact\in\allsemobs$ with $\comof{\anact}=\freeof{\anadr}$.
	Let $\historyof{\tau}=\ahist$.
	Then, $\specof{\ahist.\freeof{\anadr}}\subseteq\specof{\ahist}$.
\end{lemma}

\begin{lemma}
	\label{thm:retired-state-retired}
	Consider some $\tau\in\freesemobs$ and some $\anadr\in\adr$.
	Let $\varphi=\set{\adrvar\mapsto\anadr}$.
	Then:
	\begin{align*}
		(\ref{obs:base:init},\varphi)\trans{\historyof{\tau}}(\ref{obs:base:retired},\varphi) &\iff \anadr\in\retiredof{\tau}
		\\
		(\ref{obs:base:init},\varphi)\trans{\historyof{\tau}}(\ref{obs:base:init},\varphi) &\iff \anadr\notin\retiredof{\tau}
		\\
		\anadr\in\activeofcomp{\tau} &\implies (\ref{obs:base:init},\varphi)\trans{\historyof{\tau}}(\ref{obs:base:init},\varphi)
	\end{align*}
\end{lemma}

\begin{lemma}
	\label{thm:retired-before-freed}
	Let $\tau.(\athread,\freeof{\anadr},\anup)\in\freesemobs$.
	If $\anobs=\baseobs\times\implobs$, then $\anadr\in\retiredof{\tau}$.
\end{lemma}

\begin{lemma}
	\label{thm:deletable-computations}
	Assume $\anobs$ supports elision.
	Then, for every $\tau\in\freesemobs$ there is $\sigma\in\nosem$ with:
	\begin{inparaenum}[(i)]
		\item $\controlof{\tau}=\controlof{\sigma}$,
		\item $\heapcomput{\tau}=\heapcomput{\sigma}$, and
		\item $\freshof{\tau}\subseteq\freshof{\sigma}$.
	Moreover, if $\anobs=\baseobs\times\implobs$, then 
		\item $\retiredof{\tau}\subseteq\retiredof{\sigma}$,
		\item $\freedof{\tau}\subseteq\retiredof{\sigma}$, and
		\item $\invholdsof{\sigma}\implies\invholdsof{\tau}$.
	\end{inparaenum}
\end{lemma}

\begin{theorem}[Formalization of \Cref{thm:PRF-guarantee}]
	\label{thm:PRF-guarantee-formal}
	If $\anobs$ supports elision and the semantics $\anobs\freesem$ is pointer-race-free, then $\anobs\allsem\computrel\nosem$.
\end{theorem}


\subsection{Type System}

In this section we assume, if not stated otherwise, a fixed program $\aprog$ and a fixed SMR automaton $\anobs$ to avoid notational clutter.
A generalization to arbitrary programs is straight forward.
Recall from \Cref{sec:preliminaries} that we assume that $\anobs$ is of the form $\anobs=\baseobs\times\implobs$ for some SMR automaton $\implobs$.

\begin{definition}[$\cskip$]
	We use the $\cskip$ as syntactic sugar for a command that has no effect, for example, $\assumeof{\advar=\advar}$ where $\advar$ is some data variable.
	We assume that $\TYPESTMT{\env}{\cskip}{\env}$ holds for all~$\env$.
\end{definition}

\begin{definition}
	Indexing a (pointer/angel/data) variable $\avar$ by a thread $\athread$ yields a new variable~$\renamevarof{\avar}{\athread}$.
	Indexing all non-shared variables $\avar\notin\svars$ by $\athread$ in $\aprog$ gives a new program $\renameof{\aprog}{\athread}$.
\end{definition}

\begin{definition}
	The thread-local variables of $\athread$ is the set $\lvars{\athread}=\setcond{\renamevarof{\apvar}{\athread}}{\apvar\notin\svars}$ of non-shared variables indexed by $\athread$.
\end{definition}

\begin{definition}
	The initial program counter is $\pcinit$ with $\pcinit(\athread)=\renameof{\aprog}{\athread}$ for all threads~$\athread$.
\end{definition}

\begin{assumption}
	\label{assumption:ghost-varibles-local}
	We assume that ghost variables $\aghostvar$ are local, that is, $\aghostvar\notin\svars$.
\end{assumption}

\begin{definition}
	The initial type environment for $\aprog$ is $\envinit$.
	For $\renameof{\aprog}{\athread}$ it is $\envinitt{\athread}$.
	Formally:
	\begin{align*}
		\envinit&:=\setcond{\typeof{\apavar}{\emptyset}}{\apavar\in\pvars\cup\gvars}
		\\
		\envinitt{\athread}&:=\setcond{\typeof{\apvar}{\emptyset}}{\apvar\in\pvars\cap\svars}\cup\setcond{\typeof{\renamevarof{\apvar}{\athread}}{\emptyset}}{\apvar\in\pvars\setminus\svars}\cup\setcond{\typeof{\renamevarof{\aghostvar}{\athread}}{\emptyset}}{\aghostvar\in\gvars}
	\end{align*}
\end{definition}

\begin{figure}
	\begin{mathpar}
		\infrule{sos1}
			{\anact=(\athread,\acom,\anup)}
			{(\acom,\tau)\step[\athread](\cskip,\tau.\anact)}
		\and
		\infrule{sos2}
			{~}
			{(\cskip;\astmt,\tau)\step[\athread](\astmt,\tau)}
		\and
		\infrule{sos3}
			{i\in\set{1,2}}
			{(\astmt_1\choice\astmt_2,\tau)\step[\athread](\astmt_i,\tau)}
		\and
		\infrule{sos4}
			{\astmt_1\not\equiv\cskip \\ (\astmt_1,\tau)\step[\athread](\astmt_1',\tau')}
			{(\astmt_1;\astmt_2,\tau)\step[\athread](\astmt_1';\astmt_2,\tau')}
		\and
		\infrule{sos5}
			{\astmt'\in\set{\cskip,\astmt,\astmt;\astmt^*}}
			{(\astmt^*,\tau)\step[\athread](\astmt',\tau)}
		\and
		\infrule{sos6}
			{\apc(\athread)=\astmt \\ (\astmt,\tau)\step[\athread](\astmt',\tau') \\ \locksetof{\tau}=\set{\athread}}
			{(\apc,\tau)\step[\athread](\apc[\athread\mapsto\astmt'],\tau')}
		\and
		\infrule{sos7}
			{\apc(\athread)=\astmt \\ (\astmt,\tau)\step[\athread](\astmt',\tau) \\ \locksetof{\tau}=\emptyset}
			{(\apc,\tau)\step[\athread](\apc[\athread\mapsto\astmt'],\tau)}
		\and
		\infrule{sos8}
			{\apc(\athread)=\astmt \\ (\astmt,\tau)\step[\athread](\astmt',\tau.\tau') \\ \locksetof{\tau}=\emptyset \\ \locksetof{\tau.\tau'}=\set{\athread}}
			{(\apc,\tau)\step[\athread](\apc[\athread\mapsto\astmt'],\tau.\tau')}
		\and
		\infrule{sos9}
			{\anact=(\bot,\freeof{\anadr},\emptyset) \\ \locksetof{\tau}=\emptyset}
			{(\apc,\tau)\step[\bot](\apc,\tau.\anact)}
	\end{mathpar}
	\begin{align*}
		\locksetof{\epsilon} &:= \emptyset
		\\
		\locksetof{\tau.\anact} &:= \locksetof{\tau}\cup\set{\athread} &&\text{if } \anact=(\athread,\atomicbegin,\anup)
		\\
		\locksetof{\tau.\anact} &:= \locksetof{\tau}\setminus\set{\athread} &&\text{if } \anact=(\athread,\atomicend,\anup)
		\\
		\locksetof{\tau.\anact} &:= \locksetof{\tau} &&\text{otherwise}
	\end{align*}
	\caption{%
		The SOS rules for the transition relation $\step$ among configurations.
	}
	\label{fig:sos}
\end{figure}

\begin{definition}
	We define $\controlof{\tau}=\setcond{\apc}{(\pcinit,\epsilon)\step^*(\apc,\tau)}$ where $\step$ is the transition relation among configurations from \Cref{fig:sos}.
	Then, $\controlof[\athread]{\tau}=\setcond{\apc(\athread)}{\apc\in\controlof{\tau}}$.
\end{definition}

\begin{assumption}
	\label{assumption:control-flow}
	We assume that computations adhere to the control flow.
	Formally, this means $\controlof{\tau}\neq\emptyset$ for all $\tau\in\allsem$.
\end{assumption}

\begin{remark}
	\Cref{assumption:control-flow} requires that all primitive commands are wrapped inside atomics, that is, occur somewhere between $\atomicbegin$ and $\atomicend$.
\end{remark}

\begin{definition}
	A computation $\tau\in\allsem$ induces a flat line program for thread $\athread$, denoted by $\flatof{\athread}{\tau}$, as follows:
	\begin{align*}
		\flatof{\athread}{\epsilon} &:= \cskip
		\\
		\flatof{\athread}{\tau.\anact} &:= \flatof{\athread}{\tau};\acom &&\text{if } \anact=(\athread,\acom,\anup)
		\\
		\flatof{\athread}{\tau.\anact} &:= \flatof{\athread}{\tau} &&\text{if } \anact=(\athreadp,\acom,\anup)\wedge\athread\neq\athreadp
	\end{align*}
\end{definition}

\begin{definition}
	A pointer $\apvar$ has no valid alias in a computation $\tau$, denoted by $\noaliasof{\tau}{\apvar}$, if $\segval\neq\heapcomputof{\tau}{\apvar}\notin\heapcomputof{\tau}{\validof{\tau}\setminus\set{\apvar}}$.
\end{definition}

\begin{definition}
	Consider some $\tau\in\allsem$.
	Let $\invholdsof{\tau}$ have the prenex normal form $\exists \aghostvar_1\ldots \exists \aghostvar_n.\phi$, where $\phi$ is quantifier-free.
	Let $\aghostvar_n$ be the instance of angel $\aghostvar$ resulting from the last allocation in $\tau$.
	The set of addresses possibly represented by angel $\aghostvar$ after computation $\tau$ is
	\begin{align*}
		\denotationof{\tau}{\aghostvar}
		:=
		\setcond{\anadr\in\adr}{
			\exists A_1, \ldots, A_n\subseteq \adr.~\anadr\in A_n\wedge (A_1, \ldots, A_n)
			\models
			\phi
		}
		\ .
	\end{align*}
\end{definition}


\begin{definition}
	The locations reached in $\anobs$ by a history $\ahist$ wrt. to some thread $\athread$ and some address $\anadr$ is defined by
	\(
		\lreachof[\anobs]{\athread}{\anadr}{\ahist}
		:=
		\setcond{\alocation}{
			\exists\varphi.~~
			(\alocation_\mathit{init},\varphi)\trans{\ahist}(\alocation,\varphi)
			\:\wedge\:
			\varphi(\threadvar)=\athread
			\:\wedge\:
			\varphi(\adrvar)=\anadr
		}
	\)
	where $\alocation_\mathit{init}$ is the initial location in $\anobs$.
	For $\segval$ we define $\lreachof[\anobs]{\athread}{\anadr}{\ahist}=\top$ to contain all locations of $\anobs$.
	The definition of $\lreach$ extends naturally to sets of histories.
\end{definition}



\begin{lemma}
	\label{thm:check-epsilon-transitive}
	If $\checknocomof{\env_1}{\epsilon}{\env_2}$ and $\checknocomof{\env_2}{\epsilon}{\env_3}$, then $\checknocomof{\env_1}{\epsilon}{\env_3}$.
\end{lemma}

\begin{lemma}
	\label{thm:typing-single-steps}
	Consider $\TYPESTMT{\env_1}{\astmt}{\env_2}$ and $(\astmt,\tau)\step[\athread](\astmtp,\tau.\tau')$.
	Then there is $\env$ such that $\TYPESTMT{\env_1}{\flatof{\athread}{\tau'}}{\env}$ and $\TYPESTMT{\env}{\astmtp}{\env_2}$.
\end{lemma}

\begin{lemma}
	\label{thm:typing-computations}
	Let $\TYPESTMT{\envinit}{\aprog}{\env}$.
	Consider $(\pcinit,\epsilon)\step^*(\apc,\tau)$ and some thread $\athread$.
	Then there is $\env_1,\env_2$ with $\TYPESTMT{\envinitt{\athread}}{\flatof{\athread}{\tau}}{\env_1}$ and $\TYPESTMT{\env_1}{\apc(\athread)}{\env_2}$.
\end{lemma}


\begin{lemma}
	\label{thm:env-inbetween-atomics}
	Let $\tau.\anact\in\freesem$ and $\athread\neq\threadof{\anact}$.
	Let $\TYPESTMT{\envinitt{\athread}}{\flatof{\athread}{\tau}}{\env}$ and $\apavar\in\pvars\cup\gvars$.
	Then $\gactive\notin\envof{\apavar}$ and $\apavar\notin\lvars{\athread}\implies\env(\apavar)\cap\set{\glocal,\gsafeaccess}=\emptyset$.
\end{lemma}

\begin{lemma}
	\label{thm:validity-locals}
	\label{thm:noalias-locals}
	Let $\tau.\anact\in\freesem$ and $\athread\neq\threadof{\anact}\neq\bot$.
	Let $\apvar\in\pvars\cap\lvars{\athread}$.
	Then, $\apvar\in\validof{\tau}$ implies $\apvar\in\validof{\tau.\anact}$.
	Moreover, $\noaliasof{\tau}{\apvar}$ implies $\noaliasof{\tau.\anact}{\apvar}$.
\end{lemma}

\begin{lemma}
	\label{thm:invariant-eq-implies-valid}
	Consider $\tau.\anact\in\freesem$ PRF with $\anact=(\athread,\invariantof{\apvar=\apvarp},\anup)$ and $\invholdsof{\tau.\anact}$.
	Then, $\set{\apvar,\apvarp}\cap\validof{\tau}\neq\emptyset$ implies $\set{\apvar,\apvarp}\subseteq\validof{\tau}$.
\end{lemma}

\begin{lemma}
	\label{thm:invariant-active-implies-active-valid-new}
	Consider $\tau.\anact\in\freesemobs$ PRF with $\anact=(\athread,\invariantof{\activeof{\apvar}},\anup)$ and $\invholdsof{\tau.\anact}$.
	Then, $\apvar\in\validof{\tau.\anact}$ and $\lreachof{\athread}{\anadr}{\historyof{\tau.\anact}}\subseteq\locsof{\gactive}$ for $\anadr=\heapcomputof{\tau.\anact}{\apvar}$.
\end{lemma}

\begin{lemma}
	\label{thm:ghost-active-implies-active-notretired-new}
	Consider $\tau.\anact\in\freesemobs$ PRF with $\anact=(\athread,\ghostof{\activeof{\aghostvar}},\anup)$ and $\invholdsof{\tau.\anact}$.
	Then, $\denotationof{\tau.\anact}{\aghostvar}\cap\freedof{\tau.\anact}=\emptyset$ and $\lreachof{\athread}{\anadr}{\historyof{\tau.\anact}}\subseteq\locsof{\gactive}$ for all $\anadr\in\denotationof{\tau.\anact}{\aghostvar}$.
\end{lemma}

\begin{lemma}
	\label{thm:check-preserves-type-guranatees-new-new}
	Let $\tau\in\freesemobs$.
	Let $\env,\envp$ such that $\checknocomof{\env}{\epsilon}{\envp}$.
	Let $\athread$ be some thread.
	Let $\apvar\in\pvars$ and $\anadr=\heapcomputof{\tau}{\apvar}$.
	Let $\aghostvar\in\gvars$ and $\anadrp\in\denotationof{\tau}{\aghostvar}$.
	Then,
	\begin{align*}
		\isvalidof{\env(\apvar)}\implies\apvar\in\validof{\tau}
		&\qquad\text{implies}\qquad
		\isvalidof{\envp(\apvar)}\implies\apvar\in\validof{\tau}
		\\
		\isvalidof{\env(\aghostvar)}\implies\anadrp\notin\freedof{\tau}
		&\qquad\text{implies}\qquad
		\isvalidof{\envp(\aghostvar)}\implies\anadrp\notin\freedof{\tau}
		\\
		\glocal\in\env(\apvar)\implies\noaliasof{\tau}{\apvar}
		&\qquad\text{implies}\qquad
		\glocal\in\envp(\apvar)\implies\noaliasof{\tau}{\apvar}
		\\
		\lreachof{\athread}{\anadr}{\historyof{\tau}}\subseteq\locsof{\envof{\apvar}}
		&\qquad\text{implies}\qquad
		\lreachof{\athread}{\anadr}{\historyof{\tau}}\subseteq\locsof{\envpof{\apvar}}
		\\
		\lreachof{\athread}{\anadrp}{\historyof{\tau}}\subseteq\locsof{\envof{\aghostvar}}
		&\qquad\text{implies}\qquad
		\lreachof{\athread}{\anadrp}{\historyof{\tau}}\subseteq\locsof{\envpof{\aghostvar}}
	\end{align*}
\end{lemma}


\begin{lemma}
	\label{thm:type-guarantees-vs-computations-new}
	Assume $\anobs$ supports elision, and $\invholdsof{\nosem}$.
	Consider some thread~$\athread$, some type environments $\env$, and some $\tau\in\freesemobs$ PRF with $\invholdsof{\tau}$ and $\TYPESTMT{\envinitt{\athread}}{\flatof{\athread}{\tau}}{\env}$.
	Then, for every $\apvar\in\pvars$, we have $\lreachof{\athread}{\anadr}{\historyof{\tau}}\subseteq\locsof{\envof\apvar}$ and $\isvalidof{\env(\apvar)}\implies\apvar\in\validof{\tau}$.
\end{lemma}

\begin{lemma}
	\label{thm:type-check-implies-prf-inv-new}
	Let $\anobs$ supports elision.
	If $\typechecks{\aprog}$ and $\invholdsof{\nosem[\aprog]}$, then $\freesemobs[\aprog]$ PRF and $\invholdsof{\freesemobs[\aprog]}$.
\end{lemma}

\begin{lemma}
	\label{thm:type-check-implies-no-double-retires}
	Let $\anobs$ supports elision.
	If $\typechecks{\aprog}$ and $\invholdsof{\nosem[\aprog]}$, $\allsem$ does not perform double retires.
\end{lemma}

\presection
\section{Proofs}

\subsection{Reduction}

\begin{proof}[Proof of \Cref{thm:prefix-closure}]
	Follows immediately from \Cref{assumption:observers-accepting-state} as it guarantees that continuations of a history not accepted by $\anobs$ are also not accepted.
\end{proof}

\begin{proof}[Proof of \Cref{thm:adrof-valid-heap-restriction}]
	Follows from \cite[Lemma D.5]{DBLP:journals/corr/MeyerW19}. 
\end{proof}

\begin{proof}[Proof of \Cref{thm:valideq}]
	Follows from \cite[Lemma D.7]{DBLP:journals/corr/MeyerW19} 
\end{proof}

\begin{proof}[Proof of \Cref{thm:move-event-from-freeable}]
	Follows from definition.
\end{proof}

\begin{proof}[Proof of \Cref{thm:invalidpointers-adr}]
	Follows from \cite[Lemma D.15]{DBLP:journals/corr/MeyerW19}. 
	That the semantics from \cite{DBLP:journals/corr/MeyerW19} sets selectors to $\bot$ for $\free$ commands does not affect the result.
\end{proof}

\begin{proof}[Proof of \Cref{thm:pointers-to-freed-are-invalid}]
	To the contrary, assume there is a shortest $\tau.\anact\in\allsem$ with some address $\anadr\in\freedof{\tau.\anact}$ and $\anadr\in\heapcomputof{\tau.\anact}{\validof{\tau.\anact}}$.
	Note that $\tau.\anact$ is indeed the shortest such computation since the claim is vacuously true for $\epsilon$.
	
	First, consider the case where we have $\anadr\notin\freedof{\tau}$.
	Then, $\anact$ must execute the command $\freeof{\anadr}$.
	As a consequence, we get $\apexp\notin\validof{\tau.\anact}$ for all $\apexp$ with $\heapcomputof{\tau}{\apexp}=\anadr$.
	So $\anadr\notin\heapcomputof{\tau.\anact}{\validof{\tau.\anact}}$.
	Since this contradicts the assumption, we must have $\anadr\in\freedof{\tau}$.
	
	Now, consider the case where we have $\anadr\in\freedof{\tau}$.
	By definition, there is some $\apexp\in\validof{\tau.\anact}$ with $\heapcomputof{\tau.\anact}{\apexp}=\anadr\neq\segval$.
	We get $\apexp\notin\validof{\tau}$ by minimality of $\tau.\anact$.
	That is, $\anact$ validates $\apexp$.
	To do so, $\anact$ must be an assignment, an allocation, or an assertion:
	\begin{itemize}
		\item
			If $\anact$ is of the form $\anact=(\athread,\apexp:=\apexpp,\anup)$, then $\apexpp\in\validof{\tau}$ and $\heapcomputof{\tau}{\apexpp}=\anadr$ must hold in order to establish the desired properties of $\apexp$.
			However, $\heapcomputof{\tau}{\apexpp}$ leads to $\apexpp\notin\validof{\tau}$ by minimality of $\tau.\anact$.
			Hence, $\anact$ cannot be an assignment.
		\item
			If $\anact$ is of the form $\anact=(\athread,\apexp:=\malloc,\anup)$, then $\anadr\notin\freedof{\tau.\anact}$ by definition.
			Hence, $\anact$ cannot be an allocation targeting $\apexp$.
		\item
			If $\anact$ is of the form $\anact=(\athread,\apvar:=\malloc,\anup)$ with $\heapcomputof{\tau.\anact}{\apvar}=\anadrp$ and $\apexp\equiv\psel{\anadrp}$, then $\heapcomputof{\tau.\anact}{\apexp}=\segval\neg\anadr$.
			Hence, $\anact$ cannot be an allocation.
		\item
			If $\anact$ is of the form $\anact=(\athread,\assumeof{\apvar=\apvarp},\anup)$, then wlog. $\apexp\equiv\apvar$ and $\apvarp\in\validof{\tau}$ and $\anadr=\heapcomputof{\tau.\anact}{\apexp}=\heapcomputof{\tau}{\apexp}=\heapcomputof{\tau}{\apvarp}$ must hold.
			Again by minimality, we get $\apvarp\notin\validof{\tau}$ which contradicts the assumption.
			Hence, $\anact$ cannot be an assertion.
	\end{itemize}
	The above case distinction is complete and thus concludes the claim.
\end{proof}

\begin{proof}[Proof of \Cref{thm:elision-support-lift}]
	Let $\tau\in\allsemobs$, $\anadr\in(\freshof{\tau}\cup\freedof{\tau})\setminus\retiredof{\tau}$, and $\anadrp\in\freshof{\tau}$.
	Let $\historyof{\tau}=\ahist$.
	We have $\renamingof{\freeableof{\ahist}{\anadrp}}{\anadrp}{\anadr}=\freeableof{\renamingof{\ahist}{\anadrp}{\anadr}}{\anadr}$ according to \cite[Lemma D.26]{DBLP:journals/corr/MeyerW19}. 
	If $\anadr\in\freedof{\tau}$, then \Cref{def:elision-support}\ref{def:elision-support:fresh-brandnew} yields $\freeableof{\ahist}{\anadr}\subseteq\freeableof{\renamingof{\ahist}{\anadrp}{\anadr}}{\anadr}$.
	Thus, $\freeableof{\ahist}{\anadr}\subseteq\renamingof{\freeableof{\ahist}{\anadrp}}{\anadrp}{\anadr}$ as desired.
	Otherwise, we have $\anadr\in\freshof{\tau}$.
	This means $\renamingof{\ahist}{\anadrp}{\anadr}=\ahist$.
	So, $\freeableof{\renamingof{\ahist}{\anadrp}{\anadr}}{\anadr}=\freeableof{\ahist}{\anadr}$.
	Hence, $\freeableof{\ahist}{\anadr}=\renamingof{\freeableof{\ahist}{\anadrp}}{\anadrp}{\anadr}$ as desired.
\end{proof}

\begin{proof}[Proof of \Cref{thm:replacing-addresses-in-computations-new}]
	Follows from \cite[Lemmas D.26, D.27 and D.28]{DBLP:journals/corr/MeyerW19}. 
	That the semantics from \cite{DBLP:journals/corr/MeyerW19} sets selectors to $\bot$ for $\free$ commands does not affect the result.
\end{proof}

\begin{proof}[Proof of \Cref{thm:elision-computations-new}]
	Follows from \cite[Proofs of Proposition C.14, Lemma C.16, and Lemma D.16]{DBLP:journals/corr/MeyerW19}. 
\end{proof}


\begin{proof}[Proof of \Cref{thm:fresh-notin-range}]
	Follows from \cite[Lemma D.9]{DBLP:journals/corr/MeyerW19}. 
\end{proof}

\begin{proof}[Proof of \Cref{thm:disjoint-fresh-retired}]
	The claim holds for the empty computation $\epsilon$.
	To the contrary, assume the claim does not hold.
	Then, there must be a shortest computation $\tau.\anact\in\allsem$ such that $\freshof{\tau.\anact}\cap\retiredof{\tau.\anact}\neq\emptyset$.
	Let $\anadr\in\freshof{\tau.\anact}\cap\retiredof{\tau.\anact}$.
	By minimality of $\tau.\anact$, we must have $\anadr\notin\freshof{\tau}$ or $\anadr\notin\retiredof{\tau}$.
	If $\anadr\notin\freshof{\tau}$, then we have $\anadr\notin\freshof{\tau.\anact}$ by definition.
	Since this contradicts the assumption, we must have $\anadr\in\freshof{\tau}$ and $\anadr\notin\retiredof{\tau}$.
	To arrive at $\anadr\in\retiredof{\tau.\anact}$, $\anact$ must be of the form $\anact=(\athread,\enterof{\retireof{\apvar}},\anup)$ with $\heapcomputof{\tau}{\apvar}=\anadr$.
	By the contrapositive of \Cref{thm:fresh-notin-range}, we have $\anadr\notin\freshof{\tau}$.
	As before, this gives $\anadr\notin\freshof{\tau.\anact}$ and contradicts the assumption.
\end{proof}

\begin{proof}[Proof of \Cref{thm:invalid-freed-or-fresh}]
	The claim holds for the empty computation $\epsilon$ by definition.
	Consider some $\tau.\anact\in\freesem$ PRF with $\anact=(\athread,\acom,\anup)$ such that for every $\apexpp\in\pexp$ we have $\apexpp\notin\validof{\tau}$ implies $\heapcomputof{\tau}{\apexpp}\in\freedof{\tau}$ or $\apexpp\notin\pvars\wedge\set{\apexpp}\cap\adr\subseteq\freshof{\tau}\cup\freedof{\tau}$.
	Let $\apexp\in\pexp$ be some pointer expression with $\apexp\notin\validof{\tau.\anact}$.
	We do a case distinction.
	\begin{compactitem}
		\item
			Consider $\acom$ being an assignment.
			If $\acom$ does not contain $\apexp$ at the left-hand side, then $\apexp\notin\validof{\tau}$.
			Hence, we have either $\heapcomputof{\tau.\anact}{\apexp}=\heapcomputof{\tau}{\apexp}\in\freedof{\tau}=\freedof{\tau.\anact}$ or $\apexp\notin\pvars\wedge\set{\apexp}\cap\adr\subseteq\freshof{\tau}\cup\freedof{\tau}=\freshof{\tau.\anact}=\freedof{\tau.\anact}$.
			Otherwise, $\acom\equiv\apexp:=\apexpp$.
			If $\apexpp\in\pvars$, then $\apexpp\notin\validof{\tau}$.
			Thus, $\heapcomputof{\tau}{\apexpp}\in\freedof{\tau}$.
			We get $\heapcomputof{\tau.\anact}{\apexp}\in\freedof{\tau.\anact}$.
			Otherwise, $\apexpp\equiv\psel{\apvar}$ with $\heapcomputof{\tau}{\apvar}=\anadr$ and $\psel{\anadr}\notin\validof{\tau}$.
			By \Cref{thm:fresh-notin-range} we have $\anadr\notin\freshof{\tau}$.
			So we get $\heapcomputof{\tau}{\apexpp}\in\freedof{\tau}$ or $\anadr\in\freedof{\tau}$.
			The latter cannot apply since \Cref{thm:pointers-to-freed-are-invalid} gives $\apvar\notin\validof{\tau}$ which means $\tau.\anact$ raises a pointer race contradicting the assumption.
			In the remaining case we get $\heapcomputof{\tau.\anact}{\apexp}\in\freedof{\tau.\anact}$ as desired.

		\item
			Consider $\acom$ being an $\assume$, an $\invariant$, an $\enter$, or an $\exit$.
			Then, we know that $\apexp\notin\validof{\tau}$, $\heapcomputof{\tau}{\apexp}=\heapcomputof{\tau.\anact}{\apexp}$, $\freshof{\tau}=\freshof{\tau.\anact}$, and $\freedof{\tau}=\freedof{\tau.\anact}$.
			This implies the desired property.

		\item
			Consider $\acom$ being an allocation.
			Let $\acom\equiv\apvar:=\malloc$.
			The update is of the form $\anup=[\apvar\mapsto\anadr,\psel{\anadr}\mapsto\segval,dots]$ for some $\anadr\in\adr$.
			By the semantics, we have $\anadr\in\freshof{\tau}$.
			We get $\freshof{\tau.\anact}=\freshof{\tau}\setminus\set{\anadr}$ and $\freedof{\tau.\anact}=\freedof{\tau}\setminus\set{\anadr}$.
			Since $\apexp\notin\validof{\tau.\anact}$, we have $\apexp\notin\set{\apvar,\psel{\anadr}}$.
			So $\heapcomputof{\tau}{\apexp}=\heapcomputof{\tau.\anact}{\apexp}$.
			If $\heapcomputof{\tau}{\apexp}\in\freedof{\tau}$ we get $\heapcomputof{\tau}{\apexp}\notin\freshof{\tau}$ and thus $\heapcomputof{\tau.\anact}{\apexp}\in\freedof{\tau.\anact}$.
			Otherwise, $\apexp\equiv\psel{\anadrp}$ and $\anadrp\in\freshof{\tau}\cup\freedof{\tau}$ with $\anadr\neq\anadrp$.
			So we get the desired $\anadrp\in\freshof{\tau.\anact}\cup\freedof{\tau.\anact}$.

		\item
			Consider $\acom$ being a free.
			Let $\acom\equiv\freeof{\anadr}$.
			If $\apexp\in\validof{\tau}$, then we have $\heapcomputof{\tau}{\apexp}=\anadr$ or $\apexp\equiv\psel{\anadr}$ since $\anact$ invalidates $\apexp$.
			We get either $\heapcomputof{\tau.\anact}{\apexp}\in\freedof{\tau.\anact}$ or $\apexp\equiv\psel{\anadr}\wedge\anadr\in\freedof{\tau.\anact}$.
			Otherwise, we have $\apexp\notin\validof{\tau}$.
			If $\heapcomputof{\tau}{\apexp}\in\freedof{\tau}$, then $\heapcomputof{\tau.\anact}{\apexp}\in\freedof{\tau.\anact}$.
			Otherwise, $\apexp\equiv\psel{\anadrp}\wedge\anadrp\in\freshof{\tau}\cup\freedof{\tau}$.
			This gives $\anadrp\in\freshof{\tau.\anact}\cup\freedof{\tau.\anact}$ because $\anadr=\anadrp$ yields $\anadrp\in\freedof{\tau.\anact}$ and $\anadr\neq\anadrp$ does not affect the freshness/freedness of $\anadrp$.
	\end{compactitem}
	This concludes the claim.
\end{proof}

\begin{proof}[Proof of \Cref{thm:invalid-freed}]
	Follows from \Cref{thm:invalid-freed-or-fresh}.
\end{proof}

\begin{proof}[Proof of \Cref{thm:freeable-vs-deletable-new}]
	By assumption we have
	\begin{auxiliary}
		\forall\ahist\,\forall\anadr,\anadrp.&~\anadr\neq\anadrp\implies\freeableof{\ahist.\freeof{\anadr}}{\anadrp}=\freeableof{\ahist}{\anadrp}
		\label[aux]{proof:freeable:aux1}
		\\
		\forall\ahist\,\forall\anadr.&~\freeableof{\ahist.\freeof{\anadr}}{\anadr}\subseteq\freeableof{\ahist}{\anadr}
		\label[aux]{proof:freeable:aux2}
	\end{auxiliary}
	Consider some $\ahist.\freeof{\anadr}$.
	We show $\specof{\ahist.\freeof{\anadr}}\subseteq\specof{\ahist}$ as this implies the claim.
	Towards a contradiction, assume the inclusion does not hold.
	That is, there is a shortest $\ahistp\in\specof{\ahist.\freeof{\anadr}}$ with $\ahistp\notin\specof{\ahist}$.
	If $\freesof{\ahistp}=\emptyset$, then we get $\ahistp\in\freeableof{\ahist.\freeof{\anadr}}{\anadr}$.
	By \Cref{proof:freeable:aux2} we have $\ahistp\in\freeableof{\ahist}{\anadr}$.
	This means $\ahistp\in\specof{\ahist}$ by definition.
	This contradicts the choice of $\ahistp$.
	So we must have $\freesof{\ahistp}\neq\emptyset$.

	Consider now $\freesof{\ahistp}\neq\emptyset$.
	That is, there is a decomposition of $\ahistp$ of the form $\ahistp=\ahist_1.\freeof{\anadrp}.\ahist_2$ with $\freesof{\ahist_2}=\emptyset$.
	We derive the following.
	\begin{compactitem}
		\item
		We have $\ahist_1.\freeof{\anadrp}.\ahist_2\notin\specof{\ahist}$.
		That is, $\ahist.\ahist_1.\freeof{\anadrp}.\ahist_2\notin\specof{\anobs}$.
		By definition, this means $\ahist_2\notin\freeableof{\ahist.\ahist_1.\freeof{\anadrp}}{\anadrpp}$ with $\anadrpp\neq\anadrp$.
		Now, \Cref{proof:freeable:aux1} yields $\ahist_2\notin\freeableof{\ahist.\ahist_1}{\anadrpp}$.
		Since $\freesof{\ahist_2}=\emptyset$, we must have $\ahist.\ahist_1.\ahist_2\notin\specof{\anobs}$.
		That is, we get $\ahist_1.\ahist_2\notin\specof{\ahist}$.

		\item
		We have $\ahist_1.\freeof{\anadrp}.\ahist_2\in\specof{\ahist.\freeof{\anadr}}$.
		So, $\ahist.\freeof{\anadr}.\ahist_1.\freeof{\anadrp}.\ahist_2\in\specof{\anobs}$.
		By $\freesof{\ahist_2}$, we get $\ahist_2\in\freeableof{\ahist.\freeof{\anadr}.\ahist_1.\freeof{\anadrp}}{\anadrp}$.
		Then, \Cref{proof:freeable:aux2} yields $\ahist_2\in\freeableof{\ahist.\freeof{\anadr}.\ahist_1}{\anadrp}$.
		So, $\ahist.\freeof{\anadr}.\ahist_1.\ahist_2\in\specof{\anobs}$.
		That is, we get $\ahist_1.\ahist_2\in\specof{\ahist.\freeof{\anadr}}$.
	\end{compactitem}
	Altogether, this means we have $\ahist_1.\ahist_2\in\specof{\ahist.\freeof{\anadr}}$ and $\ahist_1.\ahist_2\notin\specof{\ahist}$ with $\ahist_1.\ahist_2$ being shorter than $\ahist_1.\freeof{\anadrp}.\ahist_2$.
	This contradicts the minimality of $\ahistp$ and thus concludes the claim.
\end{proof}


\begin{proof}[Proof of \Cref{thm:retired-state-retired}]
	Let $\anadr\in\adr$ and $\varphi=\set{\adrvar\mapsto\anadr}$.
	The claim holds for $\epsilon$.
	Towards a contradiction, assume there is a shortest $\tau.\anact\in\freesemobs$ with $(\ref{obs:base:init},\varphi)\trans{\historyof{\tau.\anact}}(\ref{obs:base:retired},\varphi)$ and $\anadr\notin\retiredof{\tau.\anact}$.
	If $\historyof{\tau.\anact}=\historyof{\tau}$, then we have $\retiredof{\tau.\anact}=\retiredof{\tau}$.
	This contradicts the assumption that $\tau.\anact$ is the shortest such computation.
	So $\historyof{\tau.\anact}$ is of the form $\historyof{\tau.\anact}=\ahist.\anevent$ with $\ahist=\historyof{\tau}$.
	We do a case distinction on the state after $\tau$.
	\begin{compactitem}
		\item
			Consider $(\ref{obs:base:init},\varphi)\trans{\ahist}(\ref{obs:base:final},\varphi)$.
			By definition of $\baseobs$, there is not step $(\ref{obs:base:final},\varphi)\trans{\anevent}(\ref{obs:base:retired},\varphi)$.
			Hence, this case cannot apply.
		\item
			Consider $(\ref{obs:base:init},\varphi)\trans{\ahist}(\ref{obs:base:init},\varphi)$.
			Then we must have $(\ref{obs:base:init},\varphi)\trans{\anevent}(\ref{obs:base:retired},\varphi)$.
			This means $\anevent$ is of the form $\anevent=\retireof{\athread,\anadr}$ for some thread $\athread$.
			By definition, $\anact=(\athread,\retireof{\apvar},\emptyset)$ with $\heapcomputof{\tau}{\apvar}=\anadr$.
			Thus, $\anadr\in\retiredof{\tau.\anact}$.
			This contradicts the assumption.
		\item
			Consider $(\ref{obs:base:init},\varphi)\trans{\ahist}(\ref{obs:base:retired},\varphi)$.
			By minimality, we have $\anadr\in\retiredof{\tau}$.
			To arrive at $\anadr\notin\retiredof{\tau}$, we must have $\anevent=\freeof{\anadr}$.
			This, however, leads to $(\ref{obs:base:retired},\varphi)\trans{\anevent}(\ref{obs:base:init},\varphi)$.
			So, $(\ref{obs:base:init},\varphi)\trans{\historyof{\tau.\anact}}(\ref{obs:base:init},\varphi)$.
			This contradicts the assumption.
	\end{compactitem}
	The above case distinction is complete and thus proves that $(\ref{obs:base:init},\varphi)\trans{\historyof{\tau}}(\ref{obs:base:retired},\varphi)$ implies $\anadr\in\retiredof{\tau}$.
	Consider now the reverse direction.
	To that end, consider some $\tau\in\freesemobs$ and some $\anadr\notin\retiredof{\tau}$.
	Using the contrapositive of the above, we get $(\ref{obs:base:init},\varphi)\trans{\historyof{\tau}}(\alocation,\varphi)$ with $\alocation\neq\ref{obs:base:retired}$.
	By \Cref{Assumption:Product}, $\alocation\neq\ref{obs:base:final}$ as for otherwise $\tau\notin\freesemobs$.
	Hence, $\alocation=\ref{obs:base:init}$ must hold.
	This establishes the first equivalence.
	The second equivalence follow analogously.
	The remaining property follows from the second equivalence together with the fact that $\anadr\in\activeofcomp{\tau}$ implies $\anadr\notin\retiredof{\tau}$.
\end{proof}

\begin{proof}[Proof of \Cref{thm:retired-before-freed}]
	Let $\tau.\anact\in\freesemobs$ with $\anact=(\athread,\freeof{\anadr},\anup)$.
	Let $\varphi=\set{\adrvar\mapsto\anadr}$.
	We have $\historyof{\tau.\anact}=\ahist.\freeof{\anadr}$ with $\ahist=\historyof{\tau}$.
	By definition, $\ahist.\freeof{\anadr}\in\specof{\baseobs}$.
	So we must have $(\ref{obs:base:init},\varphi)\trans{\ahist}(\ref{obs:base:retired},\varphi)$ as for otherwise $\freeof{\anadr}$ would take $\baseobs$ to $\ref{obs:base:init}$ and thus give $\tau.\anact\notin\freesemobs$.
	Now \Cref{thm:retired-state-retired} yields the desired $\anadr\in\retiredof{\tau}$.
\end{proof}

\begin{proof}[Proof of \Cref{thm:deletable-computations}]
	For $\tau=\epsilon$ we choose $\epsilon=\sigma\in\nosem$.
	Then, $\sigma$ satisfies the desired properties.
	Consider now $\tau.\anact\in\freesemobs$.
	Assume we already constructed $\sigma\in\nosem$~with:
	\begin{enumerate}[label=({P}\arabic*),leftmargin=1.3cm,parsep=.5ex,topsep=-1ex] 
		\item \label[property]{proof:deletable:control} $\controlof{\tau}=\controlof{\sigma}$,
		\item \label[property]{proof:deletable:heap} $\heapcomput{\tau}=\heapcomput{\sigma}$,
		\item \label[property]{proof:deletable:fresh} $\freshof{\tau}\subseteq\freshof{\sigma}$,
		\item \label[property]{proof:deletable:retired} $\retiredof{\tau}\subseteq\retiredof{\sigma}$, and
		\item \label[property]{proof:deletable:freed} $\freedof{\tau}\subseteq\retiredof{\sigma}$.
		\item \label[property]{proof:deletable:inv} $\invholdsof{\sigma}\implies\invholdsof{\tau}$.
	\end{enumerate}\vspace{2mm}
	First, assume $\comof{\anact}\not\equiv\freeof{\anadr}$.
	We get $\sigma.\anact\in\nosem$.
	The reason for this is that $\sigma$:
	\begin{compactitem}
	 	\item performs the same updates in assignments due to \Cref{proof:deletable:heap},
	 	\item allows for the same $\assume$ commands due to \Cref{proof:deletable:heap},
	 	\item emits the same events due to \Cref{proof:deletable:heap},
	 	\item allows for the same $\malloc$ commands due to \Cref{proof:deletable:fresh}, and
	 	\item nothing can be removed from $\retiredof{\sigma}$.
	\end{compactitem}
	Moreover, $\sigma$ satisfies the desired properties.
	This is because the same update is applied after $\tau$ and $\sigma$.
	To see that $\invholdsof{\sigma.\anact}\implies\invholdsof{\tau.\anact}$, let $\anact$ execute an annotation.
	There are two cases:
	\begin{compactitem}
		\item
			Consider $\comof{\anact}$ is of the form $\ghostof{\chooseof{\aghostvar}}$.
			By definition, $\invholdsof{\tau.\anact}=\exists\aghostvar.~F_\tau$ and $\invholdsof{\sigma.\anact}=\exists\aghostvar.~F_\sigma$.
			By induction, we get $\invholdsof{\sigma.\anact}\implies\invholdsof{\tau.\anact}$.
		\item
			Consider $\comof{\anact}\not\equiv\ghostof{\chooseof{\aghostvar}}$.
			By definition then, $\invholdsof{\tau.\anact}=\invholdsof{\tau}\wedge F_\tau$ and $\invholdsof{\sigma.\anact}=\invholdsof{\sigma}\wedge F_\sigma$.
			If $\comof{\anact}\notin\set{\invariantof{\activeof{\apvar}},\ghostof{\activeof{\aghostvar}}}$, we have $F_\tau\equiv F_\sigma$ by \Cref{proof:deletable:heap}.
			Otherwise, we have $F_\sigma\implies F_\tau$ because \Cref{proof:deletable:retired,proof:deletable:freed} gives $\activeofcomp{\sigma}\subseteq\activeofcomp{\tau}$.
			By induction, we get $\invholdsof{\sigma.\anact}\implies\invholdsof{\tau.\anact}$.
	\end{compactitem}
	The case distinction is complete and thus concludes the claim.

	If $\comof{\anact}\equiv\freeof{\anadr}$, then $\sigma$ already has the desired properties.
	This is the case because $\freeof{\anadr}$ does not affect the memory nor the control.
	The $\freeof{\anadr}$ may remove $\anadr$ from $\freshof{\tau}$, thus maintaining $\freshof{\tau.\anact}\subseteq\freshof{\sigma}$.
	Similarly, we maintain $\retiredof{\tau.\anact}\subseteq\retiredof{\sigma}$.
	Last, we have $\freedof{\tau.\anact}=\freedof{\tau}\cup\set{\anadr}$.
	It remains to show that $\anadr\in\retiredof{\sigma}$.
	This follows from \Cref{thm:retired-before-freed} together with $\retiredof{\tau}\subseteq\retiredof{\sigma}$ from induction.
\end{proof}

\begin{proof}[Proof of \Cref{thm:PRF-guarantee-formal}]
	Follows from \Cref{thm:elision-computations-new,thm:deletable-computations}.
\end{proof}

\begin{proof}[Proof of \Cref{thm:PRF-guarantee}]
	Follows from \Cref{thm:PRF-guarantee-formal}.
\end{proof}

\begin{proof}[Proof of \Cref{thm:generalized-RPRF-guarantee}]
	Follows from \Cref{thm:elision-computations-new,thm:deletable-computations}.
	Here, we rely on a generalization of \Cref{thm:elision-computations-new} that assumes only RPRF in its premise.
	To prove the generalization of that \namecref{thm:elision-computations-new}, we have to consider one additional case, namely the case where $\tau.\anact\in\allsemobs$ is a PR but RPRF.
	(During the proof, we have already constructed $\sigma\in\freesemobs$ with $\tau\computequiv\sigma$.)
	That is, $\anact$ is of the form $\assumeof{\apvar=\apvarp}$ such that, without loss of generality, $\apvar\notin\validof{\tau}$ and $\freeof{\anadr}\in\historyof{\tau}$ where $\anadr=\heapcomputof{\tau}{\apvarp}$.
	This means $\tau.\anact\in\asemobs{\adr\setminus\set{\anadr}}{\adr\setminus\set{\anadr}}{\anobs}$.
	So $\apvarp\in\validof{\tau}$.
	By \cite[Lemma D.15]{DBLP:journals/corr/MeyerW19} we have $\heapcomputof{\tau}{\apvar}\neq\anadr=\heapcomputof{\tau}{\apvarp}$.
	Hence, the case cannot apply.
	(The remaining cases are identical to the non-generalized version, \Cref{thm:elision-computations-new}).
\end{proof}


\subsection{Type System}

\begin{proof}[Proof of \Cref{thm:check-epsilon-transitive}]
	Immediately follows from definition.
\end{proof}

\begin{proof}[Proof of \Cref{thm:typing-single-steps}]
	We do an induction over the derivation depth of $\TYPESTMT{\env_1}{\astmt}{\env_2}$.
	\begin{description}[labelwidth=6mm,leftmargin=8mm,itemindent=0mm]
		\item[IB:]
			The derivation is due to one rule application.
			This means the derivation is not due to one of: \ref{rule:infer}, \ref{rule:seq}, \ref{rule:choice}, or \ref{rule:loop}.
			For the remaining, applicable rules we have $\astmt\equiv\acom$ and $(\acom,\tau)\step[\athread](\cskip,\tau.\tau')$.
			Thus, $\flatof{\athread}{\tau'}=\astmt$.
			We immediately get $\TYPESTMT{\env_1}{\flatof{\athread}{\tau'}}{\env_2}$.
			Moreover, $\TYPESTMT{\env_2}{\cskip}{\env_2}$ holds by definition.
			That is, we choose $\env=\env_2$ as it has the desired properties.

		\item[IH:]
			If $\TYPESTMT{\env_1}{\astmt}{\env_2}$ and $(\astmt,\tau)\step[\athread](\astmtp,\tau.\tau')$ hold, then there is some $\env$ with $\TYPESTMT{\env_1}{\flatof{\athread}{\tau'}}{\env}$ and $\TYPESTMT{\env}{\astmtp}{\env_2}$.

		\item[IS:]
			Consider a composed derivation of $\TYPESTMT{\env_1}{\astmt}{\env_2}$.
			Let $(\astmt,\tau)\step[\athread](\astmtp,\tau.\tau')$.
			We do a case distinction on the first rule of the derivation.
			\begin{compactdesc}
				\item[\textnormal{\textit{Rule \ref{rule:seq}, part 1.}}]
					Consider $\astmt\equiv\cskip;\astmt_2$.
					Then, $\astmtp=\astmt_2$ and $\tau'=\epsilon$.
					Moreover, there exists some $\env$ such that $\TYPESTMT{\env_1}{\cskip}{\env}$ and $\TYPESTMT{\env}{\astmt_2}{\env_2}$ due to the type rules.
					Note that $\flatof{\athread}{\tau'}=\cskip$.
					That is, $\env$ has the desired properties.

				\item[\textnormal{\textit{Rule \ref{rule:seq}, part 2.}}]
					Consider $\astmt\equiv\astmt_1;\astmt_2$.
					Let $(\astmt_1,\tau)\step[\athread](\astmt_1',\tau.\tau'')$.
					By definition, $\astmtp=\astmt_1';\astmt_2$ and $\tau'=\tau''$.
					The type rules give some $\envp$ with $\TYPESTMT{\env_1}{\astmt_1}{\envp}$ and $\TYPESTMT{\envp}{\astmt_2}{\env_2}$.
					By induction, there is $\env$ with $\TYPESTMT{\env_1}{\flatof{\athread}{\tau'}}{\env}$ and $\TYPESTMT{\env}{\astmt_1'}{\envp}$.
					The latter gives $\TYPESTMT{\env}{\astmt_1';\astmt_2}{\env_2}$.
					That is, $\env$ has the desired properties.

				\item[\textnormal{\textit{Rule \ref{rule:choice}.}}]
					We have $\astmt\equiv\astmt_1\choice\astmt_2$.
					Then, $\tau'=\epsilon$ and $\astmtp=\astmt_i$ for some $i\in\set{1,2}$.
					Due to the type rules, we have $\TYPESTMT{\env_1}{\astmt_i}{\env_2}$.
					Moreover, $\flatof{\athread}{\tau'}=\cskip$ gives $\TYPESTMT{\env_1}{\flatof{\athread}{\tau'}}{\env_1}$.
					That is, $\env=\env_1$ is an adequate choice with the desired properties.

				\item[\textnormal{\textit{Rule \ref{rule:loop}.}}]
					We have $\astmt\equiv\astmt_1^*$.
					Then, $\tau'=\epsilon$ and $\astmtp=\astmt_i$ for some $i\in\set{1,2}$.
					Due to the type rules, we have $\TYPESTMT{\env_1}{\astmt_i}{\env_2}$.
					Moreover, $\flatof{\athread}{\tau'}=\cskip$ gives $\TYPESTMT{\env_1}{\flatof{\athread}{\tau'}}{\env_1}$.
					That is, $\env=\env_1$ is an adequate choice with the desired properties.

				\item[\textnormal{\textit{Rule \ref{rule:infer}.}}]
					There are type environments $\env_3$ and $\env_4$ such that $\checknocomof{\env_1}{\epsilon}{\env_3}$, $\TYPESTMT{\env_3}{\astmt}{\env_4}$, and $\checknocomof{\env_4}{\epsilon}{\env_2}$.
					By induction, there is $\env$ with $\TYPESTMT{\env_3}{\flatof{\athread}{\tau'}}{\env}$ and $\TYPESTMT{\env}{\astmtp}{\env_4}$.
					Applying Rule~\ref{rule:infer} we get $\TYPESTMT{\env_1}{\flatof{\athread}{\tau'}}{\env}$ and $\TYPESTMT{\env}{\astmtp}{\env_2}$ as desired.
			\end{compactdesc}
	\end{description}
	The above induction concludes the claim.
\end{proof}

\begin{proof}[Proof of \Cref{thm:typing-computations}]
	Let $\TYPESTMT{\envinit}{\aprog}{\env}$.
	We proceed by induction over the SOS transitions.
	\begin{description}[labelwidth=6mm,leftmargin=8mm,itemindent=0mm]
		\item[IB:]
			We have $(\pcinit,\epsilon)\step^0(\apc,\tau)$.
			That is, $\apc=\pcinit$ and $\tau=\epsilon$.
			Consider some thread $\athread$.
			We have $\flatof{\athread}{\tau}=\cskip$ and $\apc(\athread)=\renameof{\aprog}{\athread}$.
			The former gives $\TYPESTMT{\envinitt{\athread}}{\flatof{\athread}{\tau}}{\envinitt{\athread}}$.
			The latter gives $\TYPESTMT{\envinitt{\athread}}{\apc(\athread)}{\env}$ due to the premise.
			So we can choose $\env_1=\envinitt{\athread}$ and $\env_2=\env$.

		\item[IH:]
			Let the claim hold for sequences of up to $n$ steps.

		\item[IS:]
			Consider now $(\pcinit,\epsilon)\step^n(\apc,\tau)\step[\athreadp](\apcp,\tau.\tau')$.
			Let $\athread\neq\bot$ be some arbitrary thread.
			By induction, there are $\envp_1,\env_2$ with \[
				\TYPESTMT{\envinitt{\athread}}{\flatof{\athread}{\tau}}{\envp_1}
				\qquad\text{and}\qquad
				\TYPESTMT{\envp_1}{\apc(\athread)}{\env_2}
				\ .
			\]
			First, assume $\athread\neq\athreadp$.
			Then, $\flatof{\athread}{\tau.\anact}=\flatof{\athread}{\tau}$ and $\apcp(\athread)=\apc(\athread)$.
			Thus, the claim follows immediately by induction.
			So consider $\athread=\athreadp$ now.
			We have $(\apc(\athread),\tau)\step[\athread](\apcp(\athread),\tau.\tau')$ by definition.
			\Cref{thm:typing-single-steps} yields $\env_1$ with \[
				\TYPESTMT{\envp_1}{\flatof{\athread}{\tau'}}{\env_1}
				\qquad\text{and}\qquad
				\TYPESTMT{\env_1}{\apcp(\athread)}{\env_2}
				\ .
			\]
			Altogether, we get the desired: \[
				\TYPESTMT{\envinitt{\athread}}{\flatof{\athread}{\tau.\anact}}{\env_1}
				\qquad\text{and}\qquad
				\TYPESTMT{\env_1}{\apcp(\athread)}{\env_2}
				\ .
			\]
	\end{description}
	The above induction concludes the claim.
\end{proof}

\begin{proof}[Proof of \Cref{thm:env-inbetween-atomics}]
	Let $\tau.\anact\in\smash{\freesem}$ and $\athread,\athreadp$ threads with $\athread\neq\athreadp=\threadof{\anact}$.
	Consider some $\TYPESTMT{\envinitt{\athread}}{\flatof{\athread}{\tau}}{\env}$ and $\apavar\in\pvars\cup\gvars$.
	We have $\locksetof{\tau}=\set{\athreadp}$ or $\locksetof{\tau.\anact}=\set{\athreadp}$.
	Hence, $\athread$ has not yet contributed any actions to $\tau$ or it has finished an atomic section.
	We do a case distinction.
	\begin{compactitem}
		\item
			Consider the case $\flatof{\athread}{\tau}=\cskip$.
			By the type rules there is $\envp$ with:
			\begin{align*}
				\checknocomof{\envinitt{\athread}}{\epsilon}{\envp}
				\qquad
				\TYPECOM{\envp}{\cskip}{\envp}
				\qquad
				\checknocomof{\envp}{\epsilon}{\env}
			\end{align*}
			By definition, $\envinitt{\athread}(\apavar)=\emptyset$.
			So $\neg\isvalidof{\envinitt{\athread}(\apavar)}$.
			Hence, $\neg\isvalidof{\envp(\apavar)}$ and $\neg\isvalidof{\env(\apavar)}$ follow from the definition of type inference.
			So we conclude $\env(\apavar)\cap\set{\gactive,\glocal,\gsafeaccess}=\emptyset$.

		\item
			Consider the case $\flatof{\athread}{\tau.\anact}=\astmt;\atomicend$ for some statement $\astmt$.
			By the typing rules there are $\env_1,\env_2,\env_3$ with:
			\begin{align*}
				\TYPESTMT{\envinitt{\athread}}{\astmt}{\env_1}
				\qquad
				\checknocomof{\env_1}{\epsilon}{\env_2}
				\qquad
				\TYPESTMT{\env_2}{\atomicend}{\env_3}
				\qquad
				\checknocomof{\env_3}{\epsilon}{\env}
			\end{align*}
			where the derivation $\TYPESTMT{\env_2}{\atomicend}{\env_3}$ is due to Rule~\ref{rule:end}.
			This means, $\env_3=\rmtransientof{\env_2}$.
			By definition, $\gactive\notin\env_3(\apavar)$.
			Hence, type inference provides $\gactive\notin\env(\apvar)$ as desired.
			
			Now, consider the case $\apvar\notin\lvars{\athread}$.
			There are two cases.
			First, assume $\apvar\in\svars$.
			$\env_3(\apvar)=\emptyset$ by definition.
			Second, assume $\apvar\notin\svars$.
			That is, $\apavar$ is local to another thread $\athreadpp\neq\athread$: $\apvar\in\lvars{\athreadpp}$.
			Then, the claim follows because the initial type binding does not contain local pointers of other threads and the type rules never add type bindings.
	\end{compactitem}
	The above case distinction is complete and concludes the claim thus.
\end{proof}

\begin{proof}[Proof of \Cref{thm:validity-locals,thm:noalias-locals}]
	Let $\tau.\anact\in\freesem$ and $\athread,\athreadp$ threads with $\athread\neq\athreadp=\threadof{\anact}\neq\bot$.
	Let $\apvar\in\pvars\cap\lvars{\athread}$.
	By definition, $\lvars{\athread}\cap\lvars{\athreadp}=\emptyset$.
	Due to the semantics, $\apvar$ does not occur in $\comof{\anact}$.
	Hence, $\apvar\in\validof{\tau}\iff\apvar\in\validof{\tau.\anact}$.
	Moreover, $\heapcomputof{\tau}{\apvar}=\heapcomputof{\tau.\anact}{\apvar}$.
	So every valid alias created by $\anact$ requires a valid alias in $\tau$.
	This is not possible since $\noaliasof{\tau}{\apvar}$.
\end{proof}

\begin{proof}[Proof of \Cref{thm:invariant-eq-implies-valid}]
	Let $\tau.\anact\in\freesem$ PRF with $\anact=(\athread,\invariantof{\apvar=\apvarp},\emptyset)$ and $\invholdsof{\tau.\anact}$.
	The latter gives $\heapcomputof{\tau}{\apvar}=\heapcomputof{\tau}{\apvarp}$.
	Towards a contradiction, assume the claim does not hold.
	Wlog. $\apvar\notin\validof{\tau}$ and $\apvarp\in\validof{\tau}$.
	\Cref{thm:invalidpointers-adr} gives $\heapcomputof{\tau}{\apvar}\neq\heapcomputof{\tau}{\apvarp}$.
	This contradicts the premise and thus concludes the claim.
\end{proof}

\begin{proof}[Proof of \Cref{thm:invariant-active-implies-active-valid-new}]
	Let $\tau.\anact\in\freesem$ PRF with $\anact=(\athread,\invariantof{\activeof{\apvar}},\anup)$ and $\invholdsof{\tau.\anact}$.
	By definition, this means $\heapcomputof{\tau}{\apvar}\in\activeofcomp{\tau}$.
	That is, $\heapcomputof{\tau}{\apvar}\notin\freedof{\tau}$.
	By the contrapositive of \Cref{thm:invalid-freed}: $\apvar\in\validof{\tau}$.
	So $\apvar\in\validof{\tau.\anact}$ by definition.
	Moreover, we get $\heapcomputof{\tau.\anact}{\apvar}\in\activeofcomp{\tau.\anact}$.
	Hence, the remaining property follows from \Cref{thm:retired-state-retired}.
\end{proof}

\begin{proof}[Proof of \Cref{thm:ghost-active-implies-active-notretired-new}]
	Let $\tau.\anact\in\freesem$ PRF with $\anact=(\athread,\ghostof{\activeof{\aghostvar}},\anup)$ and $\invholdsof{\tau.\anact}$.
	By definition, we have $\denotationof{\tau}{\aghostvar}\subseteq\activeofcomp{\tau}$.
	Hence, $\denotationof{\tau}{\aghostvar}\cap\freedof{\tau.\anact}=\emptyset$.
	Moreover, we have $\denotationof{\tau.\anact}{\aghostvar}\subseteq\activeofcomp{\tau.\anact}$ by definition.
	Let $\anadr\in\denotationof{\tau.\anact}{\aghostvar}$.
	This means $\anadr\in\activeofcomp{\tau.\anact}$.
	Then, the remaining property follows from \Cref{thm:retired-state-retired}.
\end{proof}

\begin{proof}[Proof of \Cref{thm:check-preserves-type-guranatees-new-new}]
	Let $\tau\in\freesemobs$.
	Let $\env,\envp$ be two type environments with $\checknocomof{\env}{\epsilon}{\envp}$.
	Let $\apvar\in\pvars$ be a pointer with $\anadr=\heapcomputof{\tau}{\apvar}$.
	Let $\aghostvar\in\gvars$ a ghost variable and $\anadrp\in\denotationof{\tau}{\aghostvar}$.
	Consider the possible assumptions:
	\begin{align}
		\isvalidof{\env(\apvar)}&\implies\apvar\in\validof{\tau}
		\label{proof:check-preserves-type-guarantess:assume-valid-pointer}
		\\
		\isvalidof{\env(\aghostvar)}&\implies\anadrp\notin\freedof{\tau}
		\label{proof:check-preserves-type-guarantess:assume-valid-ghost}
		\\
		\glocal\in\env(\apvar)&\implies\noaliasof{\tau}{\apvar}
		\label{proof:check-preserves-type-guarantess:assume-local}
		\\
		\lreachof{\athread}{\anadr}{\historyof{\tau}}&\subseteq\locsof{\envof{\apvar}}
		\label{proof:check-preserves-type-guarantess:assume-history-ptr}
		\\
		\lreachof{\athread}{\anadrp}{\historyof{\tau}}&\subseteq\locsof{\envof{\aghostvar}}
		\label{proof:check-preserves-type-guarantess:assume-history-ghost}
	\end{align}
	The definition of $\checknocomof{\env}{\epsilon}{\envp}$ gives:
	\begin{align}
		\isvalidof{\envpof{\apvar}}&\implies\isvalidof{\envof{\apvar}}
		\label{proof:check-preserves-type-guarantess:check-valid-pointer}
		\\
		\isvalidof{\envpof{\aghostvar}}&\implies\isvalidof{\envof{\aghostvar}}
		\label{proof:check-preserves-type-guarantess:check-valid-ghost}
		\\
		\glocal\in\envpof{\apvar}&\implies\glocal\in\envof{\apvar}
		\label{proof:check-preserves-type-guarantess:check-local}
		\\
		\locsof{\envof{\apvar}}&\subseteq\locsof{\envpof{\apvar}}
		\label{proof:check-preserves-type-guarantess:check-history-ptr}
		\\
		\locsof{\envof{\aghostvar}}&\subseteq\locsof{\envpof{\aghostvar}}
		\label{proof:check-preserves-type-guarantess:check-history-ghost}
	\end{align}
	Then, we combine
	\begin{compactitem}
	 	\item If \eqref{proof:check-preserves-type-guarantess:assume-valid-pointer} holds, then \eqref{proof:check-preserves-type-guarantess:check-valid-pointer} gives $\isvalidof{\envp(\apvar)}\implies\apvar\in\validof{\tau}$,
	 	\item If \eqref{proof:check-preserves-type-guarantess:assume-valid-ghost} holds, then \eqref{proof:check-preserves-type-guarantess:check-valid-ghost} gives $\isvalidof{\envp(\aghostvar)}\implies\anadrp\notin\freedof{\tau}$,
	 	\item If \eqref{proof:check-preserves-type-guarantess:assume-local} holds, then \eqref{proof:check-preserves-type-guarantess:check-local} gives $\glocal\in\envp(\apvar)\implies\noaliasof{\tau}{\apvar}$, and
	 	\item If \eqref{proof:check-preserves-type-guarantess:assume-history-ptr} holds, then \eqref{proof:check-preserves-type-guarantess:check-history-ptr} gives $\lreachof{\athread}{\anadr}{\historyof{\tau}}\subseteq\locsof{\envpof{\apvar}}$.
	 	\item If \eqref{proof:check-preserves-type-guarantess:assume-history-ghost} holds, then \eqref{proof:check-preserves-type-guarantess:check-history-ghost} gives $\lreachof{\athread}{\anadrp}{\historyof{\tau}}\subseteq\locsof{\envpof{\aghostvar}}$.
	\end{compactitem}
	This concludes the claim.
\end{proof}



\begin{proof}[Proof of \Cref{thm:type-guarantees-vs-computations-new}]
	Let $\tau\in\freesemobs$.
	We show:
	\begin{align*}
		&\freedof{\tau}\cap\retiredof{\tau}=\emptyset
		\\
		\text{and}\qquad
		\forall\athread~
		\forall\env
		.~~~
		&\TYPESTMT{\envinitt{\athread}}{\flatof{\athread}{\tau}}{\env}
		\\&\implies
		\forall\apvar,\aghostvar.~~
		\left(
		\begin{aligned}
			&\lreachof{\athread}{\heapcomputof{\tau}{\apvar}}{\historyof{\tau}}\subseteq\locsof{\envof{\apvar}}
			\\\wedge~
			&\isvalidof{\env(\apvar)}\implies\apvar\in\validof{\tau}
			\\\wedge~
			&\glocal\in\env(\apvar)\implies\noaliasof{\tau}{\apvar}
			\\\wedge~
			&\forall\anadr\in\denotationof{\tau}{\aghostvar}.~
			 \lreachof{\athread}{\anadr}{\historyof{\tau}}\subseteq\locsof{\envof{\aghostvar}}
			\\\wedge~
			&\isvalidof{\env(\aghostvar)}\implies\denotationof{\tau}{\aghostvar}\cap\freedof{\tau}=\emptyset
		\end{aligned}
		\right)
	\end{align*}
	We proceed by induction over the structure of $\tau$.
	\begin{description}[labelwidth=6mm,leftmargin=8mm,itemindent=0mm]
		\item[IB:]
			Let $\tau=\epsilon$.
			Let $\athread$ be some thread and let $\env$ be some type environment such that $\TYPESTMT{\envinitt{\athread}}{\flatof{\athread}{\tau}}{\env}$.
			Note that $\flatof{\athread}{\tau}=\cskip$.
			By definition, $\freedof{\tau}\cap\retiredof{\tau}=\emptyset$.
			Consider some thread $\athread$,  some $\apvar\in\pvars$ and some $\aghostvar\in\gvars$.
			By definition, $\apvar\in\validof{\tau}$.
			This gives the desired implication $\isvalidof{\envof{\apvar}}\implies\apvar\in\validof{\tau}$.
			By the type rules we have:
			\begin{align*}
				\checkof{\envinitt{\athread}}{\epsilon}{\env_1}
				\qquad
				\TYPECOM{\env_1}{\cskip}{\env_1}
				\qquad
				\checkof{\env_1}{\epsilon}{\env}
			\end{align*}
			Since $\glocal\notin\envinitt{\athread}(\apvar)$, we get $\glocal\notin\envof{\apvar}$ by the definition of type inference.
			So we satisfy the implication $\glocal\in\env(\apvar)\implies\noaliasof{\tau}{\apvar}$.
			Moreover, $\freedof{\tau}=\emptyset$.
			So $\isvalidof{\env(\aghostvar)}\implies\anadr\notin\freedof{\tau}$ is satisfied for every $\anadr\in\adr$ as well.
			By \Cref{thm:check-epsilon-transitive} we have $\checknocomof{\envinitt{\athread}}{\epsilon}{\env}$.
			That is, $\locsof{\envinitt{\athread}(\apvar)}\subseteq\locsof{\envof{\apvar}}$.
			By definition, $\envinitt{\athread}(\apvar)=\emptyset$.
			That is, $\locsof{\envinitt{\athread}(\apvar)}=\top\times\top$.
			Consequently, $\lreachof{\athread}{\heapcomputof{\tau}{\apvar}}{\historyof{\tau}}\subseteq\locsof{\envinitt{\athread}(\apvar)}$.
			Similarly for $\aghostvar$.
			Altogether, this concludes the base case.


		\item[IH:]
			Let the claim hold for $\tau$.

		\item[IS:]
			Consider $\tau'=\tau.\anact$ with $\anact=(\athreadp,\acom,\anup)$, $\tau.\anact$ PRF, and $\invholdsof{\tau.\anact}$.
			Let $\athread$ be some arbitrary thread we establish the claim for.
			We do a case distinction on $\athreadp$.

			\ad{$\athread=\athreadp\neq\bot$}
			We have \[
				\flatof{\athread}{\tau.\anact} = \flatof{\athread}{\tau};\acom
				\qquad\text{and}\qquad
				\freedof{\tau.\anact}\subseteq\freedof{\tau}
				\ .
			\]
			Assume that $\tau.\anact$ can be typed for $\athread$ as nothing needs to be shown otherwise.
			That is, assume there are $\env_3$ such that \[
				\TYPESTMT{\envinitt{\athread}}{\flatof{\athread}{\tau.\anact}}{\env_3}
				\ .
			\]
			Due to the type rules and the above equality, we know that there are some $\env_1,\env_2$ such that:
			\begin{align*}
				\TYPESTMT{\envinitt{\athread}}{\flatof{\athread}{\tau}}{\env_0}
				\qquad
				\checknocomof{\env_0}{\epsilon}{\env_1}
				\qquad
				\TYPESTMT{\env_1}{\acom}{\env_2}
				\qquad
				\checknocomof{\env_2}{\epsilon}{\env_3}
			\end{align*}
			where $\TYPESTMT{\env_1}{\acom}{\env_2}$ is derived by neither Rule \ref{rule:seq} nor Rule \ref{rule:choice} nor Rule \ref{rule:loop} nor Rule \ref{rule:infer}.
			By induction, the claim holds for $\env_0$.
			So by \Cref{thm:check-preserves-type-guranatees-new-new} the claim also holds for $\env_1$.
			If the claim holds for $\env_2$, then the claim follow for $\env_3$ from \Cref{thm:check-preserves-type-guranatees-new-new} again.
			So it remains to show that the claim holds for $\env_2$ relying on $\env_1$.
			We do a case distinction over the type rules applied for the derivation $\TYPECOM{\env_1}{\acom}{\env_2}$.
			To that end, let $\apvar\in\pvars$ be some arbitrary pointer variable and let $\aghostvar\in\gvars$ be some arbitrary angel.
			\newcommand{\thep}{\anadrpp_\apvar}
			\newcommand{\thea}{\anadrpp_\aghostvar}
			Let $\heapcomputof{\tau}{\apvar}=\thep$ and let $\thea\in\denotationof{\tau}{\aghostvar}$.
			We show
			\begin{enumerate}[label=({G}\arabic*),leftmargin=1.3cm] 
				\item
					\label[property]{proof:type-guarantees-vs-computations:goal-t:histories-ptr}
					$\lreachof{\athread}{\thep}{\historyof{\tau.\anact}}\subseteq\locsof{\env_2(\apvar)}$
				\item
					\label[property]{proof:type-guarantees-vs-computations:goal-t:histories-ghost}
					$\lreachof{\athread}{\thea}{\historyof{\tau.\anact}}\subseteq\locsof{\env_2(\aghostvar)}$
				\item
					\label[property]{proof:type-guarantees-vs-computations:goal-t:valid-ptr}
					$\isvalidof{\env_2(\apvar)}\implies\apvar\in\validof{\tau.\anact}$
				\item
					\label[property]{proof:type-guarantees-vs-computations:goal-t:local-ptr}
					$\glocal\in\env_2(\apvar)\implies\noaliasof{\tau.\anact}{\apvar}$
				\item
					\label[property]{proof:type-guarantees-vs-computations:goal-t:valid-angle}
					$\isvalidof{\env_2(\aghostvar)}\implies\thea\notin\freedof{\tau.\anact}$
			\end{enumerate}
			\medskip

			\begin{casedistinction}
				\item[Rule \ref{rule:begin}]
					By definition, we have $\env_2=\env_1$, $\heapcomput{\tau}=\heapcomput{\tau.\anact}$, $\validof{\tau}=\validof{\tau.\anact}$, $\freedof{\tau}=\freedof{\tau.\anact}$, $\denotation{\tau}=\denotation{\tau.\anact}$, and $\historyof{\tau}=\historyof{\tau.\anact}$.
					Hence, the claim follows by induction.

				\item[Rule \ref{rule:end}]
					By definition, we have $\env_2=\rmtransientof{\env_1}$.
					By definition, $\env_2(\apvar)\subseteq\env_1(\apvar)$ and $\env_2(\aghostvar)\subseteq\env_1(\aghostvar)$.
					This means we have $\locsof{\env_2(\apvar)}\supseteq\locsof{\env_1(\apvar)}$ and $\locsof{\env_2(\aghostvar)}\supseteq\locsof{\env_1(\aghostvar)}$.
					As in the previous case, we have $\heapcomput{\tau}=\heapcomput{\tau.\anact}$, $\validof{\tau}=\validof{\tau.\anact}$, $\freedof{\tau}=\freedof{\tau.\anact}$, $\denotation{\tau}=\denotation{\tau.\anact}$, and $\historyof{\tau}=\historyof{\tau.\anact}$.
					So the claim follows by induction.


				\item[Rule \ref{rule:assign1}]
					We have $\freedof{\tau}=\freedof{\tau.\anact}$ and $\invholdsof{\tau}\equiv\invholdsof{\tau.\anact}$.
					Hence, \Cref{proof:type-guarantees-vs-computations:goal-t:valid-angle} holds by induction.
					If $\glocal\in\env_2(\apvar)$, then $\apvar$ does not appear in $\acom$ by definition of Rule \ref{rule:assign1}.
					Hence, no alias of $\apvar$ is created by $\acom$ and \Cref{proof:type-guarantees-vs-computations:goal-t:local-ptr} continues to hold by induction.
					If $\isvalidof{\env_2}$ then there are two cases.
					First, $\acom\equiv\apvar:=\apvarp$.
					Then, $\env_2(\apvar)=\env_1(\apvarp)\setminus\set{\glocal}$.
					Hence, $\apvarp\in\validof{\tau}$ by induction.
					This leads to $\apvar\in\validof{\tau.\anact}$ as desired.
					Second, $\apvar$ is not assigned to by $\acom$.
					Then, $\env_2(\apvar)\subseteq\env_1(\apvar)$.
					Hence, $\apvar\in\validof{\tau}$ by induction and $\apvar\in\validof{\tau.\anact}$ thus.
					This concludes \Cref{proof:type-guarantees-vs-computations:goal-t:valid-ptr}.
					Note that $\historyof{\tau.\anact}=\historyof{\tau}$.
					\Cref{proof:type-guarantees-vs-computations:goal-t:histories-ghost} remains to hold by induction since $\aghostvar$ is not affected by $\acom$.
					It remains to establish \Cref{proof:type-guarantees-vs-computations:goal-t:histories-ptr}.
					If $\apvar$ does not occur in $\acom$, nothing needs to be show.
					So assume $\apvar$ occurs in $\acom$.
					In the first case, $\apvar$ occurs on the right-hand side of the assignment in $\acom$.
					Then, $\env_2(\apvar)=\env_1(\apvar)\setminus\set{\glocal}$.
					That is, $\locsof{\env_1(\apvar)}\subseteq\locsof{\env_2(\apvar)}$.
					Then, \Cref{proof:type-guarantees-vs-computations:goal-t:histories-ptr} follows by induction.
					Otherwise, $\apvar$ appears on the left-hand side of the assignment in $\acom$.
					So $\acom\equiv\apvar:=\apvarp$ for some $\apvarp$.
					Then, $\env_2(\apvar)=\env_1(\apvarp)\setminus\set{\glocal}$.
					Moreover, $\heapcomputof{\tau.\anact}{\apvar}=\heapcomputof{\tau}{\apvarp}=\thep$.
					By induction, we have $\lreachof{\athread}{\thep}{\historyof{\tau}}\subseteq\locsof{\env_1(\apvarp)}$.
					Hence, $\lreachof{\athread}{\thep}{\historyof{\tau}}\subseteq\locsof{\env_1(\apvarp)\setminus\set{\glocal}}$.
					We get the desired $\lreachof{\athread}{\thep}{\historyof{\tau.\anact}}\subseteq\locsof{\env_2(\apvarp)}$ by definition.

				\item[Rule \ref{rule:assign2}]
					Analogous to the previous case for \ref{rule:assign1}.

				\item[Rule \ref{rule:assign3}]
					Analogous to the previous case for \ref{rule:assign1}.
				

				\item[Rule \ref{rule:assign4},\ref{rule:assign5},\ref{rule:assign6},\ref{rule:assume2}]
					We have $\env_2=\env_1$, $\heapcomput{\tau}=\heapcomput{\tau.\anact}$, $\denotation{\tau}=\denotation{\tau.\anact}$, $\validof{\tau}=\validof{\tau.\anact}$, $\freedof{\tau}=\freedof{\tau.\anact}$, and $\historyof{\tau}=\historyof{\tau.\anact}$.
					Hence, the claim follows by induction.

				\item[Rule \ref{rule:assume1}]
					We have $\freedof{\tau}=\freedof{\tau.\anact}$ and $\invholdsof{\tau}\equiv\invholdsof{\tau.\anact}$.
					Hence, \Cref{proof:type-guarantees-vs-computations:goal-t:valid-angle} holds by induction.
					If $\glocal\in\env_2(\apvar)$, then $\glocal\in\env_1(\apvar)$ due to the type rule.
					This means $\noaliasof{\tau}{\apvar}$.
					Note that $\heapcomput{\tau}=\heapcomput{\tau.\anact}$.
					Moreover, since $\tau.\anact$ is PRF we know that the pointers in $\acom$ are valid.
					Hence, $\validof{\tau}=\validof{\tau.\anact}$.
					So we get $\noaliasof{\tau.\anact}{\apvar}$ by definition.
					This establishes \Cref{proof:type-guarantees-vs-computations:goal-t:local-ptr}.
					If $\isvalidof{\env_2(\apvar)}$, then there are two cases.
					First, $\isvalidof{\env_1(\apvar)}$ holds.
					This means we have $\apvar\in\validof{\tau}$ by induction.
					As stated above, this results in $\apvar\in\validof{\tau.\anact}$.
					Second, $\neg\isvalidof{\env_1(\apvar)}$ holds.
					Then, $\apvar$ is validated by $\acom$.
					For this to happen, $\apvar$ must appear in $\acom$.
					Since $\tau.\anact$ is assumed to be PRF, we know $\apvar\in\validof{\tau}$ must hold.
					So $\apvar\in\validof{\tau.\anact}$ as before.
					This gives \Cref{proof:type-guarantees-vs-computations:goal-t:valid-ptr}.
					Note that we have $\historyof{\tau.\anact}=\historyof{\tau}$ and $\denotation{\tau}=\denotation{\tau.\anact}$.
					So It remains to establish \Cref{proof:type-guarantees-vs-computations:goal-t:histories-ghost} continues to hold by induction since $\aghostvar$ is not affected.
					It remains to establish \Cref{proof:type-guarantees-vs-computations:goal-t:histories-ptr}.
					If $\apvar$ does not occur in $\acom$ nothing needs to be show.
					So assume $\apvar$ appears in $\acom$.
					Wlog. $\acom$ is of the form $\acom\equiv\assumeof{\apvar=\apvarp}$.
					Due to the type rule, we have $\env_2(\apvar)=\env_1(\apvar)\setminus\set{\glocal}$.
					By the semantics, we have $\thep=\heapcomputof{\tau.\anact}{\apvar}=\heapcomputof{\tau}{\apvar}=\heapcomputof{\tau}{\apvarp}=\heapcomputof{\tau.\anact}{\apvarp}$.
					So by induction, $\lreachof{\athread}{\thep}{\historyof{\tau}}\subseteq\locsof{\env_1(\apvar)}\cap\locsof{\env_1(\apvar)}=\locsof{\env_1(\apvar)\wedge\env_2(\apvar)}=\locsof{\env_2(\apvar)}$.
					This concludes \Cref{proof:type-guarantees-vs-computations:goal-t:histories-ptr}.

				\item[Rule \ref{rule:equal}]
					If $\isvalidof{\env_2(\apvar)}$, then there are two cases.
					First, $\isvalidof{\env_1(\apvar)}$ holds.
					This means we have $\apvar\in\validof{\tau}$ by induction.
					Then, $\apvar\in\validof{\tau.\anact}$ because $\validof{\tau}=\validof{\tau.\anact}$ by definition.
					Second, $\neg\isvalidof{\env_1(\apvar)}$ holds.
					Then, $\apvar$ is validated by $\acom$.
					For this to happen, $\acom$ must be of the form $\acom\equiv\invariantof{\apvar=\apvarp}$ with $\isvalidof{\apvarp}$.
					By induction, we have $\apvarp\in\validof{\tau}$.
					Then, \Cref{thm:invariant-eq-implies-valid} gives $\apvar\in\validof{\tau}$.
					Hence, $\apvar\in\validof{\tau.\anact}$ as before.
					This concludes \Cref{proof:type-guarantees-vs-computations:goal-t:valid-ptr}.
					The remaining properties follow analogously to the previous case for \ref{rule:assume1}.
					For \Cref{proof:type-guarantees-vs-computations:goal-t:histories-ptr} note that $\denotation{\tau}\iff\denotation{\tau.\anact}$ by definition together with the fact that $\invholdsof{\tau.\anact}$ holds.

				\item[Rule \ref{rule:active} for pointers]
					If $\isvalidof{\env_2(\apvar)}$, then there are two cases.
					First, $\isvalidof{\env_1(\apvar)}$ holds.
					This means $\apvar\in\validof{\tau}$ by induction.
					Then, $\apvar\in\validof{\tau.\anact}$ because $\validof{\tau}=\validof{\tau.\anact}$.
					Second, $\neg\isvalidof{\env_1(\apvar)}$ holds.
					Then, $\apvar$ is validated by $\acom$.
					So $\acom$ must be of the form $\acom\equiv\invariantof{\activeof{\apvar}}$.
					Since the invariants hold by assumption, $\invholdsof{\tau.\anact}$, we can invoke \Cref{thm:invariant-active-implies-active-valid-new}.
					It gives $\apvar\in\validof{\tau.\anact}$.
					Altogether, this concludes \Cref{proof:type-guarantees-vs-computations:goal-t:valid-ptr}.
					If $\glocal\in\env_2(\apvar)$, then $\glocal\in\env_1(\apvar)$ due to the type rule.
					Hence, the induction hypothesis together with $\heapcomput{\tau}=\heapcomput{\tau.\anact}$ and $\validof{\tau}=\validof{\tau.\anact}$ gives \Cref{proof:type-guarantees-vs-computations:goal-t:local-ptr}.
					For \Cref{proof:type-guarantees-vs-computations:goal-t:valid-angle} note that $\env_2(\aghostvar)=\env_1(\aghostvar)$ and that $\denotationof{\tau}{\aghostvar}=\denotationof{\tau.\anact}{\aghostvar}$ because $\invholdsof{\tau.\anact}$ by assumption.
					So \Cref{proof:type-guarantees-vs-computations:goal-t:valid-angle} follows by induction.
					Note that we have $\historyof{\tau.\anact}=\historyof{\tau}$.
					Since $\aghostvar$ is not affected, \Cref{proof:type-guarantees-vs-computations:goal-t:histories-ghost} follows by induction.
					We show \Cref{proof:type-guarantees-vs-computations:goal-t:histories-ptr}.
					If $\acom$ does not contain $\apvar$, nothing needs to be show.
					Otherwise, $\acom$ is of the form $\acom\equiv\invariantof{\activeof{\apvar}}$.
					From \Cref{thm:invariant-active-implies-active-valid-new} we get $\lreachof{\athread}{\thep}{\historyof{\tau.\anact}}\subseteq\locsof{\gactive}$.
					From induction, we get $\lreachof{\athread}{\thep}{\historyof{\tau.\anact}}\subseteq\locsof{\env_1(\apvar)}$.
					By definition, this establishes the desired $\lreachof{\athread}{\thep}{\historyof{\tau.\anact}}\subseteq\locsof{\env_1(\apvar)\wedge\gactive}=\locsof{\env_2(\apvar)}$ and thus concludes \Cref{proof:type-guarantees-vs-computations:goal-t:histories-ptr}.

				\item[Rule \ref{rule:active} for angels]
					Using \Cref{thm:ghost-active-implies-active-notretired-new}, this case is analogous to the previous, pointer case.

				\item[Rule \ref{rule:malloc}]
					Recall that $\tau.\anact\in\freesemobs$.
					So $\anact$ allocates a fresh address.
					Hence, $\freedof{\tau}=\freedof{\tau.\anact}$ by definition.
					Moreover, $\invholdsof{\tau}\equiv\invholdsof{\tau.\anact}$.
					Then, \Cref{proof:type-guarantees-vs-computations:goal-t:valid-angle} holds by induction.
					\Cref{proof:type-guarantees-vs-computations:goal-t:histories-ghost} remains to hold by induction since $\aghostvar$ is not affected.
					Consider \Cref{proof:type-guarantees-vs-computations:goal-t:histories-ptr}.
					If $\apvar$ does not appear in $\acom$, nothing needs to be shown.
					Otherwise, $\acom\equiv\apvar:=\malloc$.
					From \Cref{thm:disjoint-fresh-retired} we get $\thep\notin\retiredof{\tau}$.
					By definition then, $\thep\notin\retiredof{\tau.\anact}$.
					Then, \Cref{thm:retired-state-retired} gives $\lreach{\athread}{\thep}{\historyof{\tau.\anact}}\in\locsof{\glocal}$.
					And by definition we have $\env_2(\apvar)=\set{\glocal}$.
					This concludes \Cref{proof:type-guarantees-vs-computations:goal-t:histories-ptr}.
					
					For the remaining properties, we do a case distinction on $\apvarp$.
					First, consider the case $\apvar\neq\apvarp$.
					The, $\env_2(\apvar)=\env_1(\apvar)$ and $\apvar\in\validof{\tau.\anact}\iff\apvar\in\validof{\tau}$.
					So \Cref{proof:type-guarantees-vs-computations:goal-t:valid-ptr} follows by induction.
					Also by induction, we have $\noaliasof{\tau}{\apvar}$.
					Due to $\anup$, we have $\heapcomputof{\tau.\anact}{\apvar}=\heapcomputof{\tau}{\apvar}\neq\segval$.
					Towards a contradiction, assume $\neq\noaliasof{\tau.\anact}{\apvar}$.
					By definition, there is some pointer expression $\apexp\in\validof{\tau.\anact}\setminus\set{\apvar}$ with $\heapcomputof{\tau.\anact}{\apvar}=\heapcomputof{\tau.\anact}{\apexp}$.
					Since $\heapcomputof{\tau.\anact}{\apvarp}\in\freshof{\tau}$ we have $\heapcomputof{\tau.\anact}{\apvarp}\notin\rangeof{\heapcomput{\tau}}$ due to \Cref{thm:fresh-notin-range}.
					Hence, $\apexp\neq\apvarp$.
					Moreover, $\apexp\not\equiv\psel{\heapcomputof{\tau.\anact}{\apvarp}}$ because $\heapcomputof{\tau.\anact}{\psel{\heapcomputof{\tau.\anact}{\apvarp}}}=\segval$.
					Consequently, $\apexp$ is not affected by $\anact$.
					This means we have $\heapcomputof{\tau.\anact}{\apexp}=\heapcomputof{\tau}{\apexp}$ and $\apexp\in\validof{\tau}\setminus\set{\apvar}$.
					That is, $\neq\noaliasof{\tau}{\apvar}$.
					Since this contradicts induction, we conclude the desired $\noaliasof{\tau.\anact}{\apvar}$.
					This establishes \Cref{proof:type-guarantees-vs-computations:goal-t:local-ptr}.

					Second, consider the case $\apvar=\apvarp$.
					Then, $\env_2(\apvar)=\set{\glocal}$.
					By \Cref{thm:pointers-to-freed-are-invalid,thm:fresh-notin-range} we have $\anadr\notin\heapcomputof{\tau}{\validof{\tau}}$.
					Hence, we get $\anadr\notin\heapcomputof{\tau.\anact}{\validof{\tau.\anact}\setminus\set{\apvar}}$.
					This means $\noaliasof{\tau.\anact}{\apvar}$ holds by definition.
					This establishes \Cref{proof:type-guarantees-vs-computations:goal-t:local-ptr}.
					And by definition $\apvar\in\validof{\tau.\anact}$.
					So \Cref{proof:type-guarantees-vs-computations:goal-t:valid-ptr} holds as well.

				\item[Rule \ref{rule:enter}]
					Note that we have $\heapcomput{\tau}=\heapcomput{\tau.\anact}$ and $\denotation{\tau}=\denotation{\tau.\anact}$ by definition.
					By induction, we have $\lreachof{\athread}{\thep}{\historyof{\tau}}\subseteq\locsof{\env_1(\apvar)}$.
					Type inference gives $\lpostof{\apvar}{\acom}{\locsof{\env_1(\apvar)}}\subseteq\locsof{\env_2(\apvar)}$.
					Since $\historyof{\tau.\anact}\in\specof{\anobs}$ due to the semantics, we get $\historyof{\tau.\anact}\in\lpostof{\apvar}{\acom}{\locsof{\env_1(\apvar)}}$.
					Hence, $\lpostof{\apvar}{\acom}{\locsof{\env_2(\apvar)}}\subseteq\locsof{\env_2(\apvar)}$ as desired.
					This concludes \Cref{proof:type-guarantees-vs-computations:goal-t:histories-ptr}.
					\Cref{proof:type-guarantees-vs-computations:goal-t:histories-ghost} follows along the same lines.
					If $\isvalidof{\env_2(\apvar)}$, then $\isvalidof{\env_1(\apvar)}$ by the definition of type inference.
					Moreover, $\validof{\tau}=\validof{\tau.\anact}$.
					Hence, \Cref{proof:type-guarantees-vs-computations:goal-t:valid-ptr} follows by induction.
					Similarly, $\isvalidof{\env_2(\aghostvar)}$ implies $\isvalidof{\env_1(\aghostvar)}$.
					So \Cref{proof:type-guarantees-vs-computations:goal-t:valid-angle} follows by induction together with $\freedof{\tau}=\freedof{\tau.\anact}$ and $\invholdsof{\tau}\equiv\invholdsof{\tau.\anact}$.
					If $\glocal\in\env_2(\apvar)$, then $\glocal\in\env_1(\apvar)$ by definition.
					Note that $\heapcomput{\tau}=\heapcomput{\tau.\anact}$.
					Hence, $\noaliasof{\tau.\anact}{\apvar}$ follows by induction.
					This establishes \Cref{proof:type-guarantees-vs-computations:goal-t:local-ptr}.

				\item[Rule \ref{rule:exit}]
					Analogously to the previous case for \ref{rule:enter}.

				\item[Rule \ref{rule:angel}]
					By definition, $\heapcomput{\tau}=\heapcomput{\tau.\anact}$, $\validof{\tau}=\validof{\tau.\anact}$, and $\freedof{\tau}=\freedof{\tau.\anact}$.
					Moreover, the type rule gives $\env_2(\apvar)=\env_1(\apvar)$.
					So \Cref{proof:type-guarantees-vs-computations:goal-t:valid-ptr,proof:type-guarantees-vs-computations:goal-t:local-ptr} follow by induction.
					If $\isvalidof{\env_2(\aghostvar)}$, then $\aghostvar$ does not appear in $\acom$.
					So by the type rule we have $\env_2(\aghostvar)=\env_1(\aghostvar)$.
					Hence, \Cref{proof:type-guarantees-vs-computations:goal-t:valid-angle} follows by induction.
					Since $\apvar$ is not affected, we get \Cref{proof:type-guarantees-vs-computations:goal-t:histories-ptr} from induction.
					It remains to consider \Cref{proof:type-guarantees-vs-computations:goal-t:histories-ghost}.
					If $\aghostvar$ does not appear in $\acom$, nothing needs to be show because $\denotationof{\tau.\anact}{\aghostvar}=\denotationof{\tau}{\aghostvar}$.
					Otherwise, we have $\env_2(\aghostvar)=\emptyset$.
					This concludes \Cref{proof:type-guarantees-vs-computations:goal-t:histories-ghost}.

				\item[Rule \ref{rule:member}]
					By definition, $\heapcomput{\tau}=\heapcomput{\tau.\anact}$, $\validof{\tau}=\validof{\tau.\anact}$, and $\freedof{\tau}=\freedof{\tau.\anact}$.
					If $\isvalidof{\env_2(\aghostvar)}$, then $\isvalidof{\env_1(\aghostvar)}$ holds since $\env_2(\aghostvar)=\env_1(\aghostvar)$.
					So \Cref{proof:type-guarantees-vs-computations:goal-t:valid-angle} follows by induction.
					If $\glocal\in\env_2(\apvar)$, then $\glocal\in\env_1(\apvar)$ since angels cannot acquire guarantee $\glocal$ due to the type rules. 
					Hence, \Cref{proof:type-guarantees-vs-computations:goal-t:local-ptr} follows by induction.
					If $\isvalidof{\env_2(\apvar)}$, then there are two cases.
					First, $\isvalidof{\env_1(\apvar)}$ holds.
					Then, $\apvar\in\validof{\tau}$ by induction and thus $\apvar\in\validof{\tau.\anact}$.
					Second, $\neg\isvalidof{\env_1(\apvar)}$ holds.
					Then, $\apvar$ is validated by $\acom$.
					For this to happen, we must have $\acom\equiv\invariantof{\containsof{\apvar}{\aghostvarp}}$ with $\isvalidof{\env_1(\aghostvarp)}$.
					Since $\invholdsof{\tau.\anact}$ hold, we get $\heapcomputof{\tau}{\apvar}\in\denotationof{\tau.\anact}{\aghostvarp}$.
					Moreover, $\isvalidof{\env_1(\aghostvarp)}$ gives $\isvalidof{\env_2(\aghostvarp)}$.
					So the already established \Cref{proof:type-guarantees-vs-computations:goal-t:valid-angle} yields $\heapcomputof{\tau}{\apvar}\notin\freedof{\tau.\anact}$.
					Hence, $\heapcomputof{\tau.\anact}{\apvar}\notin\freedof{\tau.\anact}$.
					Now, $\apvar\in\validof{\tau.\anact}$ by the contrapositive of \Cref{thm:invalid-freed}.
					This establishes \Cref{proof:type-guarantees-vs-computations:goal-t:valid-ptr}.
					Consider now \Cref{proof:type-guarantees-vs-computations:goal-t:histories-ghost}.
					If $\aghostvarp$ does not appear in $\acom$, nothing needs to be shown.
					Otherwise, $\acom\equiv\invariantof{\containsof{\apvarp}{\aghostvar}}$.
					Due to assumption of $\invholdsof{\tau.\anact}$, we have $\heapcomputof{\tau.\anact}{\apvarp}\in\denotationof{\tau.\anact}{\aghostvar}$.
					By definition, $\heapcomputof{\tau.\anact}{\apvarp}=\heapcomputof{\tau}{\apvarp}$.
					Again by definition, we get $\heapcomputof{\tau}{\apvarp}\in\denotationof{\tau}{\aghostvar}$ because $\denotationof{\tau}{\aghostvar}$ is defined to be the maximal set.
					This means $\denotationof{\tau}{\aghostvar}=\denotationof{\tau.\anact}{\aghostvar}$.
					Hence, we can conclude \Cref{proof:type-guarantees-vs-computations:goal-t:histories-ghost} by induction.
					It remains to show \Cref{proof:type-guarantees-vs-computations:goal-t:histories-ptr}.
					If $\apvar$ does not occur in $\acom$, nothing needs to be shown.
					Otherwise, $\acom\equiv\invariantof{\containsof{\apvar}{\aghostvarp}}$.
					Similarly to the above, we get $\thep\in\denotationof{\tau}{\aghostvarp}$.
					So we get by induction: \[
						\lreachof{\athread}{\thep}{\historyof{\tau}}
						\subseteq\locsof{\env_1(\apvar)}\cap\locsof{\env_1(\aghostvarp)}
						=\locsof{\env_1(\apvar)\wedge\env_1(\aghostvarp)}
						=\locsof{\env_2(\apvar)}
						\ .
					\]
					This concludes \Cref{proof:type-guarantees-vs-computations:goal-t:histories-ptr} by induction together with $\historyof{\tau.\anact}=\historyof{\tau}$.
			\end{casedistinction}
			\smallskip
			The above case distinction is complete and shows that the claim holds for $\athread$ and $\env_2$.
			Hence, the claim holds for $\env_3$ as reasoned above.

			Now, we show that $\freedof{\tau.\anact}\cap\retiredof{\tau.\anact}=\emptyset$ holds.
			We have $\freedof{\tau.\anact}\subseteq\freedof{\tau}$ by the fact that $\anact$ cannot be a free due to the fact that $\athreadp\neq\bot$.
			If $\retiredof{\tau.\anact}\subseteq\retiredof{\tau}$ holds, then the claim follows by induction.
			Otherwise, we have $\retiredof{\tau.\anact}=\retiredof{\tau}\cup\set{\anadr}$ with $\acom\equiv\enterof{\retireof{\apvar}}$ and $\heapcomputof{\tau}{\apvar}=\anadr$.
			Since $\TYPECOM{\env_1}{\acom}{\env_2}$ holds, we know $\apvar\in\validof{\tau}$.
			By the contrapositive of \Cref{thm:pointers-to-freed-are-invalid}, $\anadr\notin\freedof{\tau}$.
			So by induction, $\freedof{\tau.\anact}\cap\retiredof{\tau.\anact}=\emptyset$.

			\medskip
			\ad{$\athread\neq\athreadp\neq\bot$}
			We have \[
				\TYPESTMT{\envinitt{\athread}}{\flatof{\athread}{\tau}}{\env_4}
				\quad\text{and}\quad
				\flatof{\athread}{\tau.\anact} = \flatof{\athread}{\tau}
				\quad\text{and}\quad
				\TYPESTMT{\envinitt{\athread}}{\flatof{\athread}{\tau.\anact}}{\env_4}
				\ .
			\]
			The induction hypothesis applies to $\TYPESTMT{\envinitt{\athread}}{\flatof{\athread}{\tau}}{\env_4}$.
			We have to show that the desired properties are stable under interference.

			Consider some $\apvar\in\domof{\env_4}\cap\pvars$.
			If $\glocal\in\env_4(\apvar)$, then $\apvar\in\lvars{\athread}$ by the contrapositive of \Cref{thm:env-inbetween-atomics}.
			By induction, we have $\noaliasof{\tau}{\apvar}$.
			Then, \Cref{thm:noalias-locals} gives $\noaliasof{\tau.\anact}{\apvar}$.
			If $\isvalidof{\env_4(\apvar)}$, then $\apvar\in\validof{\tau}$ by induction.
			We invoke \Cref{thm:noalias-locals} and get $\apvar\in\validof{\tau.\anact}$.

			Consider some $\aghostvar\in\domof{\env_4}\cap\gvars$.
			If $\isvalidof{\env_4(\aghostvar)}$, then $\aghostvar\in\lvars{\athread}$ by the contrapositive of \Cref{thm:env-inbetween-atomics}.
			By induction we get $\denotationof{\tau}{\aghostvar}\cap\freedof{\tau}=\emptyset$.
			Due to $\aghostvar$ being local to a thread other than the one executing $\anact$, $\aghostvar$ cannot occur in $\acom$.
			Consequently, $\denotationof{\tau.\anact}{\aghostvar}=\denotationof{\tau}{\aghostvar}$.
			So $\freedof{\tau.\anact}\subseteq\freedof{\tau}$ due to $\athreadp\neq\bot$ gives $\denotationof{\tau.\anact}{\aghostvar}\cap\freedof{\tau.\anact}=\emptyset$ as desired.


			It remains to establish $\historyof{\tau.\anact}\subseteq\locsof{\env_4}$.
			If $\historyof{\tau.\anact}=\historyof{\tau}$, then the claim follows by induction.
			So let $\historyof{\tau.\anact}=\ahist.\anevent$ with $\historyof{\tau}=\ahist$.
			We have $\project{\anevent}{\athread}=\epsilon$.
			By definition of closedness under interference, we have for all $\anadr$ and $i$:
			\begin{align*}
				\lreachof{\athread}{\anadr}{\ahist}\subseteq\locsof{\gsafeaccess}
				&\implies
				\lreachof{\athread}{\anadr}{\ahist.\anevent}\subseteq\locsof{\gsafeaccess}
				\\\text{and}\qquad
				\lreachof{\athread}{\anadr}{\ahist}\subseteq\locsof{\gcustom{i}}
				&\implies
				\lreachof{\athread}{\anadr}{\ahist.\anevent}\subseteq\locsof{\gcustom{i}}
				\ .
			\end{align*}
			Consider some $\apavar\in\domof{\env_4}$.
			By \Cref{thm:env-inbetween-atomics} we know $\gactive\notin\env_4(\apvar)$.
			If $\glocal\in\env_4(\apavar)$, then we have $\apavar\in\pvars$ since the type rules do not allow angels to carry $\glocal$.
			Moreover, induction gives $\lreachof{\athread}{\heapcomputof{\tau}{\apavar}}{\ahist}\subseteq\locsof{\glocal}$.
			To the contrary, assume $\lreachof{\athread}{\heapcomputof{\tau}{\apavar}}{\ahist.\anevent}\not\subseteq\locsof{\glocal}$.
			By definition, this means that $\anevent$ makes $\baseobs$ leave its initial location.
			Since $\ahist.\anevent\in\specof{\anobs}$ due to the semantics, we must have $\anevent=\enterof{\retireof{\athreadp,\heapcomputof{\tau}{\apavar}}}$.
			That is, $\acom\equiv\enterof{\retireof{\apvarp}}$ with $\heapcomputof{\tau}{\apvarp}=\heapcomputof{\tau}{\apavar}$.
			Since $\tau.\anact$ is assumed to be PRF, we must have $\apvarp\in\validof{\tau}$.
			This results in $\neg\noaliasof{\tau}{\apavar}$ and resembles a contradiction.
			So we conclude $\lreachof{\athread}{\heapcomputof{\tau}{\apavar}}{\ahist.\anevent}\subseteq\locsof{\glocal}$.
			By the contrapositive of \Cref{thm:env-inbetween-atomics} we also know that $\apavar\in\lvars{\athread}$.
			So $\heapcomputof{\tau}{\apavar}=\heapcomputof{\tau.\anact}{\apavar}$.
			That is, $\lreachof{\athread}{\heapcomputof{\tau.\anact}{\apavar}}{\ahist.\anevent}\not\subseteq\locsof{\glocal}$.
			Moreover, if $\apavar\in\gvars$, then $\apavar\in\lvars{\athread}$ by \Cref{assumption:ghost-varibles-local}.
			This means $\denotationof{\tau}{\apavar}=\denotationof{\tau.\anact}{\apavar}$ since $\apavar$ does not appear in $\acom$.
			So we get $\lreachof{\athread}{\anadr}{\ahist.\anevent}\subseteq\locsof{\env_4(\apavar)}$ for all $\anadr\in\denotationof{\tau.\anact}{\apavar}$ by induction.

			The remaining $\freedof{\tau.\anact}\cap\retiredof{\tau.\anact}=\emptyset$ follows as in the previous case for $\athread=\athreadp\neq\bot$.
			This concludes the case.

			\medskip
			\ad{$\athreadp=\bot$}
			We have $\anact=(\bot,\freeof{\anadr},\emptyset)$.
			Consider some thread $\athread$ and type environment $\env$ with $\TYPESTMT{\envinitt{\athread}}{\flatof{\athread}{\tau}}{\env}$.
			By definition, $\flatof{\athread}{\tau.\anact}=\flatof{\athread}{\tau}$.
			So $\TYPESTMT{\envinitt{\athread}}{\flatof{\athread}{\tau.\anact}}{\env}$.
			We show that $\env$ satisfies the claim.
			Let $\historyof{\tau}=\ahist$.
			Then, $\historyof{\tau.\anact}=\ahist.\freeof{\anadr}$.
			By the semantics, we have $\ahist.\freeof{\anadr}\in\specof{\anobs}$.
			
			Consider some $\apvar\in\domof{\env}\cap\pvars$.
			If $\glocal\in\env(\apvar)$, then $\noaliasof{\tau}{\apvar}$.
			Since $\heapcomput{\tau}=\heapcomput{\tau.\anact}$ and $\validof{\tau.\anact}\subseteq\validof{\tau}$, we get $\noaliasof{\tau.\anact}{\apvar}$.
			If $\isvalidof{\env(\apvar)}$, then $\set{\gactive,\glocal,\gsafeaccess}\cap\env(\apvar)\neq\emptyset$.
			By induction, we have $\apvar\in\validof{\tau}$.
			Note that we have $\lreachof{\athread}{\heapcomputof{\tau}{\apvar}}{\ahist}\subseteq\locsof{\env(\apvar)}$ by induction.
			Hence, $\anadr\neq\heapcomputof{\tau}{\apvar}$ must hold as for otherwise $\ahist.\freeof{\anadr}\notin\specof{\anobs}$ by the definition of $\set{\gactive,\glocal,\gsafeaccess}$.
			So, we get $\apvar\in\validof{\tau.\anact}$ by definition.

			Consider some $\aghostvar\in\domof{\env}\cap\gvars$.
			If $\isvalidof{\env(\aghostvar)}$, then $\denotationof{\tau}{\aghostvar}\cap\freedof{\tau}=\emptyset$ by induction.
			By definition, we have $\denotationof{\tau}{\aghostvar}=\denotationof{\tau.\anact}{\aghostvar}$.
			Moreover, $\freedof{\tau.\anact}=\freedof{\tau}\cup\set{\anadr}$.
			So in order to arrive at $\denotationof{\tau.\anact}{\aghostvar}\cap\freedof{\tau.\anact}=\emptyset$, it suffices to establish $\anadr\notin\denotationof{\tau}{\aghostvar}$.
			To the contrary, assume $\anadr\in\denotationof{\tau}{\aghostvar}$.
			Then, $\lreachof{\athread}{\anadr}{\ahist}\subseteq\locsof{\env(\aghostvar)}$ by induction.
			However, similar to the pointer case above, this means $\ahist.\freeof{\anadr}\notin\specof{\anobs}$ because of $\isvalidof{\env(\aghostvar)}$.
			Hence, $\anadr\notin\denotationof{\tau}{\aghostvar}$ must hold as desired.

			It remains to show that $\lreachof{\athread}{\anadrp}{\ahist.\freeof{\anadr}}\subseteq\locsof{\env(\apavar)}$ for every $\apavar\in\domof{\env_4}$ with $\anadr\neq\heapcomputof{\tau}{\apavar}$ for $\apavar\in\pvars$ and $\anadr\notin\denotationof{\tau}{\apavar}$ for $\apavar\in\gvars$.
			By definition we have:
			\begin{align*}
				\lreachof{\athread}{\anadrp}{\ahist}\subseteq\locsof{\gcustom{i}} &\implies \lreachof{\athread}{\anadrp}{\ahist.\freeof{\anadr}}\subseteq\locsof{\gcustom{i}}
				\\
				\lreachof{\athread}{\anadrp}{\ahist}\subseteq\locsof{\gsafeaccess} &\implies \lreachof{\athread}{\anadrp}{\ahist.\freeof{\anadr}}\subseteq\locsof{\gsafeaccess}
				\\
				\lreachof{\athread}{\anadrp}{\ahist}\subseteq\locsof{\gactive} &\implies \lreachof{\athread}{\anadrp}{\ahist.\freeof{\anadr}}\subseteq\locsof{\gactive} &\text{if }\anadr\neq\anadrp
				\\
				\lreachof{\athread}{\anadrp}{\ahist}\subseteq\locsof{\glocal} &\implies \lreachof{\athread}{\anadrp}{\ahist.\freeof{\anadr}}\subseteq\locsof{\glocal} &\text{if }\anadr\neq\anadrp
			\end{align*}
			Recall from above that $\set{\gactive,\glocal}\cap\env(\apavar)\neq\emptyset$ implies $\anadr\neq\heapcomputof{\tau}{\apavar}$ for $\apavar\in\pvars$ and $\anadr\notin\denotationof{\tau}{\apavar}$ for $\apavar\in\gvars$.
			Hence, we conclude the desired $\lreachof{\athread}{\anadrp}{\ahist.\freeof{\anadr}}\subseteq\locsof{\env(\apavar)}$ by induction together with $\heapcomput{\tau}=\heapcomput{\tau.\anact}$ and $\denotation{\tau}=\denotation{\tau.\anact}$.

			Lastly, we show $\freedof{\tau.\anact}\cap\retiredof{\tau.\anact}=\emptyset$.
			We have $\freedof{\tau.\anact}=\freedof{\tau}\cup\set{\anadr}$ and $\retiredof{\tau.\anact}=\retiredof{\tau}\setminus\set{\anadr}$.
			Then the claim follows by induction.
	\end{description}
\end{proof}


\begin{proof}[Proof of \Cref{thm:type-check-implies-prf-inv-new}]
	Let $\TYPESTMT{\envinit}{\aprog}{\env_\aprog}$ and $\invholdsof{\nosem}$.
	Towards a contradiction, assume the claim does not hold.
	That is, there is a shortest computation $\tau.\anact\in\freesemobs$ such that $\tau.\anact$ raises a pointer race or $\neg\invholdsof{\tau.\anact}$.
	By minimality, $\tau$ is PRF and $\invholdsof{\tau}$.
	Let $\anact=(\athread,\acom,\anup)$.
	(Note that $\acom$ is neither $\atomicbegin$ nor $\atomicend$ due to the assumption.)
	By \Cref{thm:typing-computations}, there is $\env_3$ such that $\TYPESTMT{\envinitt{\athread}}{\flatof{\athread}{\tau.\anact}}{\env_3}$.
	We have $\flatof{\athread}{\tau.\anact}=\flatof{\athread}{\tau};\acom$ by definition.
	So by the type rules there is $\env_0,\env_1,\env_2$
	\begin{align*}
		\TYPESTMT{\envinitt{\athread}}{\flatof{\athread}{\tau}}{\env_0}
		\qquad
		\checknocomof{\env_0}{\epsilon}{\env_1}
		\qquad
		\TYPECOM{\env_1}{\acom}{\env_2}
		\qquad
		\checknocomof{\env_2}{\epsilon}{\env_3}
	\end{align*}
	First, consider the case where $\tau.\anact$ raises a pointer race.
	By definition of pointer races, $\anact$ is on of: an unsafe access, an unsafe assumption, an unsafe $\enter$, or an unsafe retire.
	\begin{compactitem}
		\item
			If $\anact$ is an unsafe access, then $\acom$ contains $\psel{\apvar}$ or $\dsel{\apvar}$ with $\apvar\notin\validof{\tau}$.
			That is, $\TYPECOM{\env_1}{\acom}{\env_2}$ is derived using on of the following rules: \ref{rule:assign2}, \ref{rule:assign3}, \ref{rule:assign5}, or \ref{rule:assign6}.
			Since the derivation is defined, we must have $\env_1(\apvar)=\atype$ with $\isvalidof{\atype}$.
			By \Cref{thm:type-guarantees-vs-computations-new}, this means $\apvar\in\validof{\tau}$.
			Since this contradicts the assumption of $\anact$ raising a pointer race, this case cannot apply.
	
		\item
			If $\anact$ is an unsafe assumption, then $\acom\equiv\assumeof{\apvar=\apvarp}$ with $\set{\apvar,\apvarp}\not\subseteq\validof{\tau}$.
			That is, $\TYPECOM{\env_1}{\acom}{\env_2}$ is derived using rule \ref{rule:assume1}.
			By definition, we have $\env_1(\apvar)=\atype,\env_1(\apvarp)=\atypep$ with $\isvalidof{\atype}$ and $\isvalidof{\atypep}$.
			From \Cref{thm:type-guarantees-vs-computations-new} we get $\set{\apvar,\apvarp}\subseteq\validof{\tau}$.
			Since this contradicts the assumption of $\anact$ raising a pointer race, this case cannot apply.

		\item
			Consider an unsafe $\enter$, i.e., $\acom\equiv\enterof{\afuncof{\vecof{\apvar},\vecof{\advar}}}$.
			That is, $\TYPECOM{\env_1}{\acom}{\env_2}$ is derived using rule \ref{rule:enter}.
			Let $\heapcomputof{\tau}{\vecof{\apvar}}=\vecof{\anadr}$ and $\heapcomputof{\tau}{\vecof{\advar}}=\vecof{\advalue}$.
			That $\anact$ is an unsafe enter means that there are $\vecof{\anadrp}$ and $\anadrpp$ with:
			\begin{align*}
				&\forall i.~ (\anadr_i=\anadrpp\vee\apvar_i\in\validof{\tau}) \implies \anadr_i=\anadrp_i
				\\\text{and}\qquad
				&\freeableof{\ahist.\afunc(\athread,\vecof{\anadrp},\vecof{\advalue})}{\anadrpp}
				\not\subseteq
				\freeableof{\ahist.\afunc(\athread,\vecof{\anadr},\vecof{\advalue})}{\anadrpp}
				\ .
			\end{align*}
			Let $\vecof{\apvar}=\apvar_1,\dots,\apvar_n$ with $\env(\apvar_i)=\atype_i$.
			From \Cref{thm:type-guarantees-vs-computations-new} we get $\apvar_i\notin\validof{\tau}\implies\neg\isvalidof{\atype_i}$.
			Hence, the following holds:
			\begin{align*}
				&\forall i.~ (\anadr_i=\anadrpp\vee\apvar_i\in\validof{\tau}) \implies \anadr_i=\anadrp_i
				\\\text{implies}\quad
				&\forall i.~ \anadr_i\neq\anadrp_i \implies (\anadr_i\neq\anadrpp\wedge\apvar_i\notin\validof{\tau})
				\\\text{implies}\quad
				&\forall i.~ \anadr_i\neq\anadrp_i \implies (\anadr_i\neq\anadrpp\wedge\neg\isvalidof{\atype_i})
				\\\text{implies}\quad
				&\forall i.~ (\anadr_i=\anadrpp\vee\isvalidof{\atype_i}) \implies \anadr_i=\anadrp_i
			\end{align*}
			So $\safecallof{\env_1}{\afuncof{\vecof{\apvar},\vecof{\advar}}}=\mathit{false}$.
			As this contradicts $\TYPECOM{\env_1}{\acom}{\env_2}$, this case cannot apply.

		\item
			Consider an unsafe retire, i.e., $\acom\equiv\enterof{\retireof{\apvar}}$ with $\apvar\notin\validof{\tau}$.
			As in the previous case, $\TYPECOM{\env_1}{\acom}{\env_2}$ is derived using rule \ref{rule:enter}.
			It gives $\gactive\in\env_1(\apvar)$.
			That is, $\isvalidof{\env_1(\apvar)}=\mathit{true}$.
			Then, \Cref{thm:type-guarantees-vs-computations-new} yields $\apvar\in\validof{\tau}$.
			Hence, this case cannot apply.

	\end{compactitem}
	The above case distinction is complete.
	That is, $\tau.\anact$ cannot raise a pointer race.
	Hence, we must have $\neg\invholdsof{\tau.\anact}$.
	Since we have $\invholdsof{\tau}$ as stated before, $\anact$ must be an annotation that does not hold.
	By \Cref{thm:deletable-computations} for $\tau.\anact$ there is $\sigma\in\nosem$ with $\controlof{\tau.\anact}=\controlof{\sigma}$, $\heapcomput{\tau.\anact}=\heapcomput{\sigma}$, and $\invholdsof{\sigma}{\implies}\invholdsof{\tau.\anact}$.
	The contrapositive of the latter, gives $\neg\invholdsof{\tau.\anact}{\implies}\neg\invholdsof{\sigma}$.
	That is, we must have $\neg\invholdsof{\sigma}$.
	This contradicts the assumption of $\invholdsof{\nosem}$, concluding the claim thus.
\end{proof}

\begin{proof}[Proof of \Cref{thm:type-check-implies-no-double-retires}]
	Let $\smrobs$ supports elision.
	Furthermore, let $\typechecks{\aprog}$ and $\invholdsof{\nosem[\aprog]}$.
	By \Cref{thm:type-check-implies-prf-inv-new} we know that $\freesemobs$ is PRF and that $\invholdsof{\freesemobs}$ holds.
	Now, to the contrary, assume the overall claim does not hold.
	Then there is $\tau.\anact\in\allsem$ with $\anadr\in\retiredof{\tau}$, $\anact=(\athread,\acom,\anup)$, $\acom\equiv\enterof{\retireof{\apvar}}$, and $\heapcomputof{\tau}{\apvar}=\anadr$.
	Then, \Cref{thm:elision-computations-new} yields $\sigma\in\freesemobs$ such that $\tau\computequiv\sigma$ and $\retiredof{\tau}\subseteq\retiredof{\sigma}$.
	From the former we get $\sigma.\anact\in\freesemobs$.
	Moreover, together with $\freesemobs$ PRF and thus $\apvar\in\validof{\sigma}$, we get $\heapcomputof{\tau}{\apvar}=\anadr=\heapcomputof{\sigma}{\apvar}$.
	The latter gives $\anadr\in\retiredof{\sigma}$.

	From \Cref{thm:typing-computations} we get some $\env_3$ with $\TYPESTMT{\envinitt{\athread}}{\flatof{\athread}{\sigma.\anact}}{\env_3}$.
	Then, $\flatof{\athread}{\sigma.\anact}=\flatof{\athread}{\sigma};\acom$ by definition.
	So the typing rules give some $\env_0,\env_1,\env_2$ with
	\begin{align*}
		\TYPESTMT{\envinitt{\athread}}{\flatof{\athread}{\sigma}}{\env_0}
		\qquad
		\checknocomof{\env_0}{\epsilon}{\env_1}
		\qquad
		\TYPECOM{\env_1}{\acom}{\env_2}
		\qquad
		\checknocomof{\env_2}{\epsilon}{\env_3}
	\end{align*}
	where the derivation for $\TYPECOM{\env_1}{\acom}{\env_2}$ is due to Rule~\ref{rule:enter}.
	By definition, this means we have $\gactive\in\env_1(\apvar)$.
	Note that $\sigma$ is PRF and $\invholdsof{\sigma}$.
	Moreover, $\TYPESTMT{\envinitt{\athread}}{\flatof{\athread}{\sigma}}{\env_1}$ holds due to Rule~\ref{rule:infer}.
	Now, \Cref{thm:type-guarantees-vs-computations-new} yields $\lreachof{\athread}{\anadr}{\historyof{\sigma}}\subseteq\locsof{\env_1(\apvar)}$.
	In particular, this means $\lreachof{\athread}{\anadr}{\historyof{\sigma}}\subseteq\locsof{\gactive}$.
	Hence, $(\ref{obs:base:init},\varphi)\trans{\ahist}(\ref{obs:base:init},\varphi)$ with $\varphi=\set{\adrvar\mapsto\anadr}$ and $\ahist=\historyof{\sigma}$.
	Then, \Cref{thm:retired-state-retired} gives $\anadr\notin\retiredof{\sigma}$.
	This contradicts the previous $\anadr\in\retiredof{\sigma}$.
\end{proof}

\begin{proof}[Proof of \Cref{Theorem:Soundness}]
	Follows from \Cref{thm:type-check-implies-prf-inv-new,thm:type-check-implies-no-double-retires}.
\end{proof}


\section{Type Checking}
\begin{lemma}\label{Lemma:Lattice}
$\actypes:=(\factorize{\alltypes}{\leadsto\cap \leadsto^{-1}}, \leadsto)$ is a complete lattice. 
\end{lemma} 
\begin{proof}
The least element is the most precise type containing every guarantee.
The join of two guarantees $\gcustom{\alocset}\sqcup \gcustom{\alocsetp}$ is characterized by the union $\alocset\cup\alocsetp$, which is again closed under interference. 
The meet of two guarantees $\gcustom{\alocset}\sqcap \gcustom{\alocsetp}$ requires care.
The intersection $\alocset\cap\alocsetp$ may not be closed under interference, in which case a corresponding guarantee $\gcustom{\alocset\cap\alocsetp}$ does not exist. 
Instead, we take the largest set of locations that is closed under interference and lives inside the intersection. 
It is guaranteed to exist.
\end{proof}
\begin{lemma}\label{Lemma:Monotonicity}
$\spof{\cdot, \acom}$ is monotonic. 
\end{lemma}
\begin{proof}
The case that requires care is Rule~\ref{rule:enter}. 
A larger enriched environment has less guarantees, and eventually $\safecallof{\env}{\afuncof{\vecof{\apvar},\vecof{\advar}}}$ may fail.
In this case, however, we return $\fail$. 
\end{proof}
\begin{proof}[Proof of \Cref{Proposition:TypeCheck}]
For $\lsolof{X}\sqsupseteq \bigsqcap_{\typejudge{\,\envinit}{\astmt}{\env}}\env$, observe that the least solution to the constraint system yields a type derivation $\typejudge{\envinit}{\astmt}{\lsolof{X}}$. 
Indeed, as the solution assigns a type environment to every control point, it already gives the intermediary environments we should assume for a Rule~\ref{rule:seq}. 
For the reverse direction, consider $\typejudge{\envinit}{\astmt}{\env}$. 
One can show that by assigning the environments encountered in the type derivation to the variables of the corresponding control points, 
we obtain a solution to the constraint system. 
Since $\lsolof{X}$ is the least solution, $\lsolof{X}\sqsubseteq \env$. 
As the reasoning holds for every derivation, $\lsolof{X}$ will lower bound the meet of the environments.
\end{proof}


\begin{lemma}
	The Kleene iteration takes time $\mathcal{O}(n^2)$ where $n$ is the size of the input program.
\end{lemma}
\begin{proof}
	%
	We denote the set of program variables by $\allvars$ and the set of variables from the constraint system by $\cvars=\set{X_1,\dots,X_k}$.
	The Kleene iteration solving the constraint system works over the product lattice \[\left(\left(\actypes^{\allvars} \cup \set{\fail}\right)^\cvars, \sqsubseteq\right)\]
	where $\sqsubseteq$ is the component-wise $\sqsubseteq$ from the original lattice $\allenvsbot$.
	Following the discussion in \Cref{sec:type_checking}, chains in this lattice have a maximal length of $m$ which is quadratic in $n$.
	To compute the Kleene iteration, we employ a worklist algorithm.

	We introduce notations to ease the formal development.
	We use $V_i=(V_i^1,\dots,V_1^k)$ where $V_i^j$ denotes the value of variable $X_j$ in step $i$.
	Moreover, we maintain a worklist $W_i$ the entries of which are variable indices that need to be recomputed after step $i$.
	That is, $j\in W_i$ states that the value $V_j$ of $X_j$ is no longer up to date since the variables that $X_j$ depends on have received new values previously.
	We write $W.W'$ and mean worklist $W$ concatenated with worklist $W'$.
	We define a total function $\IDep:\set{1,\dots,k}\to\powersetof{\set{1,\dots,k}}$ that computes dependencies among the variables from $\cvars$.
	That is, $\IDep(j)$ yields the indices of variables that depend on $X_j$.
	We use $\IDep(j)$ during the Kleene iteration to query those variables that we need to compute a new value for after updating the value of $X_j$.
	For example, the constraint $\spof{X_1, \acom}\sqsubseteq X_2$ implies $2\in \IDep(1)$.
	Similarly, we define $\IReq:\set{1,\dots,k}\to\powersetof{\set{1,\dots,k}}$ that computes requirements.
	That is, $\IReq(j)$ computes the indices of variables that are required to compute $X_j$.
	For example, the constraint $\spof{X_1, \acom}\sqsubseteq X_2$ implies $1\in \IReq(2)$.
	Both $\IDep$ and $\IReq$ can be tabulated (in quadratic time) prior to the Kleene iteration.

	In the beginning $V_0^j=\bot$ and $W_0$ is empty.
	In the first step, a new value $V_{1}^{i}\in\allenvsbot$ for all $X_i\in\cvars$ is computed.
	This is linear in $n$.
	Next, we perform steps until the worklist is empty again or we have performed $m$ steps---in both cases we are guaranteed to have found a fixed point.
	To perform a step $i\to i+1$, let $W_i$ be of the form $W_i=j.W_i^\mathit{tail}$.
	Then, we
	\begin{inparaenum}[(i)]
		\item compute a new value $V_{i+1}^j$ for $X_j$, and
		\item set $W_{i+1}=W_i^\mathit{tail}.\IDep(j)$.
	\end{inparaenum}
	Removing the first element from a list, looking up $\IDep(j)$, and appending two list can be done in constant time.
	Hence, (ii) can be done in constant time.
	It remains to consider (i).
	The updated value is: \[V_{i+1}^j=\bigsqcup_{l\in \IReq(j)} g_{j,l}(V_i^l) \ \] where $g_{j,l}(V)=V$ or $g_{j,l}(V)=\spof{V, \acom}$ for some $\acom$, depending on the constraints collected from the input program.
	Recall from \Cref{sec:type_checking} that $g_{j,l}$ can be computed in constant time.
	For $\IReq(j)$ to yield a set of constant size, we rely on a simple preprocessing of the input program: every primitive command is wrapped inside $\cskip$ commands, that is, command $\acom$ is preprocessed to $(\cskip;\acom);\cskip$.
	This way we introduce fresh variables into the constraint system such that no two branches of the program share constraint system variables and that each nested statement does not share constraint system variables with its parents.
	The program resulting from the preprocessing is linear in the size of the original program, hence our reasoning so far remains valid for the preprocessed program and carries over to the original one.
	Altogether, we arrive at (ii) taking constant time.
	This concludes the claim.
\end{proof}

\end{document}